%% file: main-journal.tex
\documentclass[11pt]{article}

\input{preamble}

\begin{document}

%

	\input{titlepage}

	\input{intro}

	\input{related}

	\input{organization_journal}
	
	\input{preliminaries}

\input{ST-Diameter-LBs}

\input{ST-Diam-Equivalence}

	\input{Eccentricity-Hardness}

	\input{eccentricity-hardness-directed}

	\input{diameter_lb_5_8}
	
	\input{diameter_lb_6_10}

    \input{diameter_lb_k_simple}
	\input{algorithm}

	\input{algorithm_linear}
	\input{corollary}

	\input{ST-Diameter-Algs}
	\input{diameter_alg_sparse}

	\input{diameter-alg}

	\input{diameter-alg-mmult}

		
	\bibliographystyle{alpha}
	\bibliography{references}

	
%
%
%
%
%
%
%
%
%
%
%

\end{document}

%% file: preamble.tex
\usepackage{booktabs}

\usepackage{amsthm}
\usepackage{amsmath}
\usepackage{amssymb}
\usepackage{fullpage}
\usepackage{color}
\usepackage{subcaption}

\usepackage{amsfonts,multirow}
\usepackage{graphicx} 
\usepackage{algorithm}
\usepackage[noend]{algpseudocode}
\usepackage{enumitem}

\usepackage{tabularx}
\makeatletter
\newcommand{\multiline}[1]{%
  \begin{tabularx}{\dimexpr\linewidth-\ALG@thistlm}[t]{@{}X@{}}
    #1
  \end{tabularx}
}
\makeatother

\makeatletter
\def\BState{\State\hskip-\ALG@thistlm}
\makeatother

\newcommand{\eps}{\ensuremath{\varepsilon}}

\DeclareMathOperator*{\argmax}{arg\,max}

\newcommand{\ignore}[1]{}

\def \poly { \text{\rm poly~} }
\def \polylog { \text{\rm polylog~} }

\newtheorem{theorem}{Theorem}

\newtheorem{lemma}[theorem]{Lemma}
\newtheorem{corollary}[theorem]{Corollary}

\newtheorem{observation}[theorem]{Observation}
\newtheorem{claim}[theorem]{Claim}

\newcommand{\Nin}{N^{\text{in}}}
\newcommand{\Nout}{N^{\text{out}}}
\newcommand{\tO}{\tilde{O}}
\newcommand{\Ot}{\tilde{O}}

%% file: titlepage.tex
\begin{titlepage}
	\title{Towards Tight Approximation Bounds for Graph Diameter and Eccentricities\footnote{Appeared in STOC '18}}
	\date{}
	\author{
		Arturs Backurs \footnote{\texttt{backurs@mit.edu}, Supported by an IBM PhD Fellowship, the NSF and the Simons Foundation}\\ MIT 
		\and Liam Roditty \footnote{\texttt{liam.roditty@biu.ac.il}}\\ Bar Ilan University
		\and Gilad Segal\footnote{\texttt{giladsegal123@gmail.com}}\\ Bar Ilan University
		\and Virginia Vassilevska Williams \footnote{\texttt{virgi@mit.edu}, Supported by an NSF CAREER Award, NSF
Grants CCF-1417238, CCF-1528078 and CCF-1514339, a BSF Grant
BSF:2012338 and a Sloan Research Fellowship.}\\ MIT 
		\and Nicole Wein \footnote{\texttt{nwein@mit.edu}, Supported by an NSF Graduate Fellowship and NSF Grant CCF-1514339}\\ MIT
	}
	\clearpage
	\maketitle
	\thispagestyle{empty}
	\begin{abstract}
		Among the most important graph parameters is the Diameter, the largest distance between any two vertices.
		There are no  known very efficient algorithms for computing the Diameter exactly. Thus, much research has been devoted to how fast this parameter can be {\em approximated}. Chechik et al.\ [SODA 2014] showed that the diameter can be approximated within a multiplicative factor of $3/2$ in $\tO(m^{3/2})$ time. 
		Furthermore, Roditty and Vassilevska W.\ [STOC 13] showed that unless the Strong Exponential Time Hypothesis (SETH) fails, no $O(n^{2-\eps})$ time algorithm can achieve an approximation factor better than $3/2$ in sparse graphs. Thus the above algorithm is essentially optimal for sparse graphs for approximation factors less than $3/2$.
		It was, however, completely plausible that a $3/2$-approximation is possible in linear time. In this work we conditionally rule out such a possibility by showing that unless SETH fails no $O(m^{3/2-\eps})$ time algorithm can achieve an approximation factor better than $5/3$.
		
		Another fundamental set of graph parameters are the Eccentricities. The Eccentricity of a vertex $v$ is the distance between $v$ and the farthest vertex from $v$. Chechik et al.\ [SODA 2014] showed that the Eccentricities of all vertices can be approximated within a factor of $5/3$ in $\tO(m^{3/2})$ time and Abboud et al.\ [SODA 2016] showed that no $O(n^{2-\eps})$ algorithm can achieve better than $5/3$ approximation in sparse graphs. We show that the runtime of the $5/3$ approximation algorithm is also optimal by proving that under SETH, there is no $O(m^{3/2-\eps})$ algorithm that achieves a better than $9/5$ approximation. We also show that no near-linear time algorithm can achieve a better than $2$ approximation for the Eccentricities. This is the first lower bound in fine-grained complexity that addresses near-linear time computation. 
		
		We show that our lower bound for near-linear time algorithms is essentially tight by giving an algorithm that approximates Eccentricities within a $2+\delta$ factor in $\tO(m/\delta)$ time for any $0<\delta<1$. This beats all Eccentricity algorithms in Cairo et al.\ [SODA 2016] and is the first constant factor approximation for Eccentricities in directed graphs.
		
		To establish the above lower bounds we study the $S$-$T$ Diameter problem: Given a graph and two subsets $S$ and $T$ of vertices, output the largest distance between a vertex in $S$ and a vertex in $T$.
		We give new algorithms and show tight lower bounds that serve as a starting point for all other hardness results.
		
		Our lower bounds apply only to sparse graphs.
		We show that for dense graphs, there are near-linear time algorithms for $S$-$T$ Diameter, Diameter and Eccentricities, with almost the same approximation guarantees as their $\tilde{O}(m^{3/2})$ counterparts, improving upon the best known algorithms for dense graphs.
		
	\end{abstract}
		\thispagestyle{empty}
\end{titlepage}

%% file: intro.tex
\section{Introduction}
Among the most important graph parameters are the graph's Diameter and the Eccentricities of its vertices.
The Eccentricity of a vertex $v$ is the (shortest path) distance to the furthest vertex from $v$, and the Diameter is the largest Eccentricity over all vertices in the graph.

The Eccentricities and Diameter measure how fast information can spread in networks. Efficient algorithms for their computation are highly desired (see e.g.~\cite{diam-prac1,diam-prac6,diam-prac3}). 
Unfortunately, the fastest known algorithms for these parameters are very slow on large graphs. For unweighted graphs on $n$ vertices and $m$ edges, the fastest Diameter algorithm runs in $\tilde{O}(\min\{mn,n^\omega\})$~\footnote{$\tilde{O}$ notation hides polylogarithmic factors} time~\cite{cyganbaur} where $\omega<2.373$ is the exponent of square matrix multiplication~\cite{williams2012multiplying,legallmult,stothers}.
For weighted graphs, the fastest Eccentricity and Diameter algorithms
actually compute all distances in the graph, i.e. they solve the All-Pairs Shortest Paths (APSP) problem. The fastest known algorithms for APSP in weighted graphs run in $\min\{\tilde{O}(mn),n^3/\exp(\sqrt{\log n})\}$~\cite{ryanapsp,Pettie04,PettieR05}.

Whether one can solve Diameter faster than APSP is a well-known open problem (e.g. see Problem 6.1 in \cite{fanchung} and \cite{aingworth,chan06j}). Whether one can solve Eccentricities faster than APSP was addressed by \cite{focsy} (for dense graphs) and by \cite{lincolnsoda} (for sparse graphs).
Vassilevska W. and Williams~\cite{focsy} showed that Eccentricities and APSP are equivalent under subcubic reductions, so that either both of them admit $O(n^{3-\eps})$ time algorithms for $\eps>0$, or neither of them do. Lincoln et al.~\cite{lincolnsoda} proved that under a popular conjecture about the complexity of weighted Clique, the $O(mn)$ runtime for Eccentricities cannot be beaten by any polynomial factor for any sparsity of the form $m=\Theta(n^{1+1/k})$ for integer $k$.

Due to the hardness of exact computation, efficient approximation algorithms are sought. A folklore $\tilde{O}(m+n)$ time algorithm achieves a $2$-approximation for Diameter in directed weighted graphs and a $3$-approximation for Eccentricities in undirected weighted graphs. Aingworth et al.~\cite{aingworth} presented an almost-$3/2$ approximation~\footnote{An almost-$c$ approximation of $X$ is an estimate $X'$ so that $X\leq X'\leq cX+O(1)$.} algorithm for Diameter running in $\tilde{O}(n^2+m\sqrt n)$ time. Roditty and Vassilevska W.~\cite{RV13} improved the result of~\cite{aingworth} with
an $\tilde{O}(m\sqrt n)$ expected time
almost-$3/2$ approximation algorithm.
Chechik et al.~\cite{ChechikLRSTW14} obtained a (genuine) $3/2$ approximation algorithm for Diameter (in directed graphs) and a (genuine) $5/3$-approximation algorithm for Eccentricities (in undirected graphs), running in $\tilde{O}(\min \{m^{3/2},mn^{2/3}\})$ time.
These are the only known non-trivial approximation algorithms for Diameter in directed graphs. So far, there are no known faster than $mn$ algorithms for approximating Eccentrities in directed graphs within any constant factor.

Cairo et al.~\cite{cairo} generalized the above results for {\em undirected} graphs and obtained a time-approximation tradeoff: for every $k\geq 1$ they obtained an $\tO(mn^{1/(k+1)})$ time algorithm that achieves an almost-$2-1/2^k$ approximation for Diameter and an almost $3-4/(2^k+1)$-approximation for Eccentricities.

\subsection{Our contributions.}
We address the following natural question:
\begin{center}{\bf Main Question:} {\em Are the known approximation algorithms for Diameter and Eccentricities optimal?}\end{center}

A partial answer is known. Under the Strong Exponential Time Hypothesis (SETH), every $3/2-\eps$ approximation algorithm (for $\eps>0$) for Diameter in undirected unweighted graphs with $O(n)$ vertices and edges must use $n^{2-o(1)}$ time~\cite{RV13}. Similarly, every $5/3-\eps$ approximation algorithm for the Eccentricities of undirected unweighted graphs with $O(n)$ vertices and edges must use $n^{2-o(1)}$ time~\cite{AbboudWW16}. This however does not answer the question of whether the runtimes of the known $3/2$ and $5/3$ approximation algorithms can be improved. It is completely plausible that there is a $3/2$-approximation algorithm for Diameter or a $5/3$-approximation for Eccentricities running in linear time.

We address our Main Question for both sparse and dense graphs. Our results are shown in Table~\ref{tab}.

\begin{table*}[h!]
    \setlength\extrarowheight{3pt}

{\small
    \begin{tabular}{|l|l|l|}
    \hline
     \bf Runtime & \bf Approximation & \bf Comments \\
      \hline
    \multicolumn{3}{|l|}{\bf Diameter Upper Bounds}\\
\hline

      $\tilde{O}(n^2)$ expected&$(3/2,5/3)$ & undirected unweighted\\
      \hline
      $O(n^{2.05})$&$(3/2,1/3)$& undirected unweighted\\
      \hline
      $O(m^2/n)$&  $<2$ for constant even Diameter& directed unweighted\\
      \hline
        \multicolumn{3}{|l|}{\bf Diameter Lower Bounds (under SETH)}\\
\hline
      $\Omega(n^{3/2-o(1)})$&$8/5-\eps$ & undirected unweighted, implies \cite{RV13,ChechikLRSTW14} alg is {\bf tight}\\
      \hline
      $\Omega(n^{3/2-o(1)})$&$5/3-\eps$& undirected weighted\\
      \hline
      $\Omega(n^{1+1/(k-1)-o(1)})$&$(5k-7)/(3k-4)-\eps$& directed unweighted, any $k\geq 3$\\
      \hline
              \multicolumn{3}{|l|}{\bf Eccentricities Upper Bounds}\\
              \hline
              $\tilde{O}(m\sqrt{n})$&2& directed weighted, approximation factor is {\bf tight}\\
              \hline
              $\tilde{O}(m/\delta)$&$2+\delta$& directed weighted, essentially {\bf tight} \\
              \hline
              $\tilde{O}(n^2)$&$(5/3,1)$& undirected unweighted\\
              \hline
                            $O(n^{2.05})$&$(5/3,1/5)$& undirected unweighted\\
\hline
                \multicolumn{3}{|l|}{\bf Eccentricities Lower Bounds (under SETH)}\\
              \hline
               $\Omega(n^{1+1/(k-1)-o(1)})$&$2-1/(2k-1)-\eps$&undirected unweighted, any $k\geq 2$, {\bf tight} for extremal $k$\\
               \hline
               $\Omega(n^{2-o(1)})$&$2-\eps$&directed unweighted, essentially {\bf tight} \\
               \hline
                \multicolumn{3}{|l|}{\bf $S$-$T$ Diameter Upper Bounds}\\
                \hline
                $O(m)$&3&{\bf tight}\\
                \hline
                $\tilde{O}(m\sqrt{n})$&$(2,3/2)$&{essentially \bf tight}\\
                \hline
                $\tilde{O}(n^2)$&$(2,7/2)$& \\
                \hline
                              $O(n^{2.05})$&$(2,3/2)$& \\
\hline
                 \multicolumn{3}{|l|}{\bf $S$-$T$ Diameter Lower Bounds (under SETH)}\\
                 \hline
                 $\Omega(n^{1+1/(k-1)-o(1)})$&$3-2/k-\eps$& any $k\geq 2$, {\bf tight} for extremal $k$\\

\hline
    \end{tabular}}
        \caption{Our results. An $(\alpha,\beta)$-approximation means that if $D$ is the true value and $D'$ is our estimate, then $D/\alpha-\beta\leq D'\leq D$. All of the lower bounds hold even for sparse graphs. $S$-$T$ Diameter is a variant of Diameter introduced later in this section.}
\label{tab}
\end{table*}


%
%
%

\paragraph{Sparse graphs.}
Our first result (restated as Theorem~\ref{diameter_5_8}) regards approximating Diameter in undirected unweighted sparse graphs.

\begin{theorem}[$3/2$-Diameter Approx. is Tight]\label{thm:32tight}
Under SETH, no $O(n^{3/2-\delta})$ time algorithm for $\delta>0$ can output a $8/5-\eps$ approximation for $\eps>0$ for the Diameter of an undirected unweighted sparse graph.
\end{theorem}

In particular, any $3/2$-approximation algorithm in sparse graphs must take $n^{3/2-o(1)}$ time. Hence the $\tilde{O}(m^{3/2})$ time $3/2$-approximation algorithm of \cite{RV13,ChechikLRSTW14} is optimal in two ways: improving the approximation ratio to $3/2-\eps$ causes a runtime blow-up to $n^{2-o(1)}$ (\cite{RV13}) and improving the runtime to $O(m^{3/2-\delta})$ causes an approximation ratio blow-up to $8/5$.

Our lower bound instance says that in $O(m^{3/2-\delta})$ time one cannot return $6$ when the Diameter is $8$.
One may be tempted to extend the above lower bound, by showing that, say, in $O(m^{4/3-\delta})$ time one cannot even return $5$ when the Diameter is $8$. This approach, however fails: in Theorem~\ref{thm:fastersparsediam} we give an $O(m^2/n)$ time algorithm that does return $5$ in this case, and in general when the Diameter is $2h$, it returns at least $h+1$. Notice that when the Diameter is $2h$, the folklore linear time algorithm returns an estimate of only $h$. Hence for sparse graphs, our algorithm runs in linear time and outperforms the folklore algorithm. Also, for constant even Diameter, it gives a better than $2$ approximation.



We obtain stronger Diameter hardness results for weighted graphs and for directed unweighted graphs. In particular, assuming SETH:
\begin{enumerate}
\item For weighted sparse graphs, no $O(n^{3/2-\delta})$ time algorithm for $\delta>0$ can output a
$5/3-\eps$ Diameter approximation (for $\eps>0$) (Theorem~\ref{diameter_6_10}).

\item For directed unweighted sparse graphs, using a general time-accuracy tradeoff lower bound (Theorem~\ref{diameter_k_simple}), we show that no near-linear time algorithm can achieve an approximation factor better than $5/3$.
\end{enumerate}

Figure~\ref{fig:diam1} shows our Diameter lower bounds.

\begin{figure*}[t!]
    \centering
    \begin{subfigure}[b]{0.35\textwidth}
        \includegraphics[width=\textwidth]{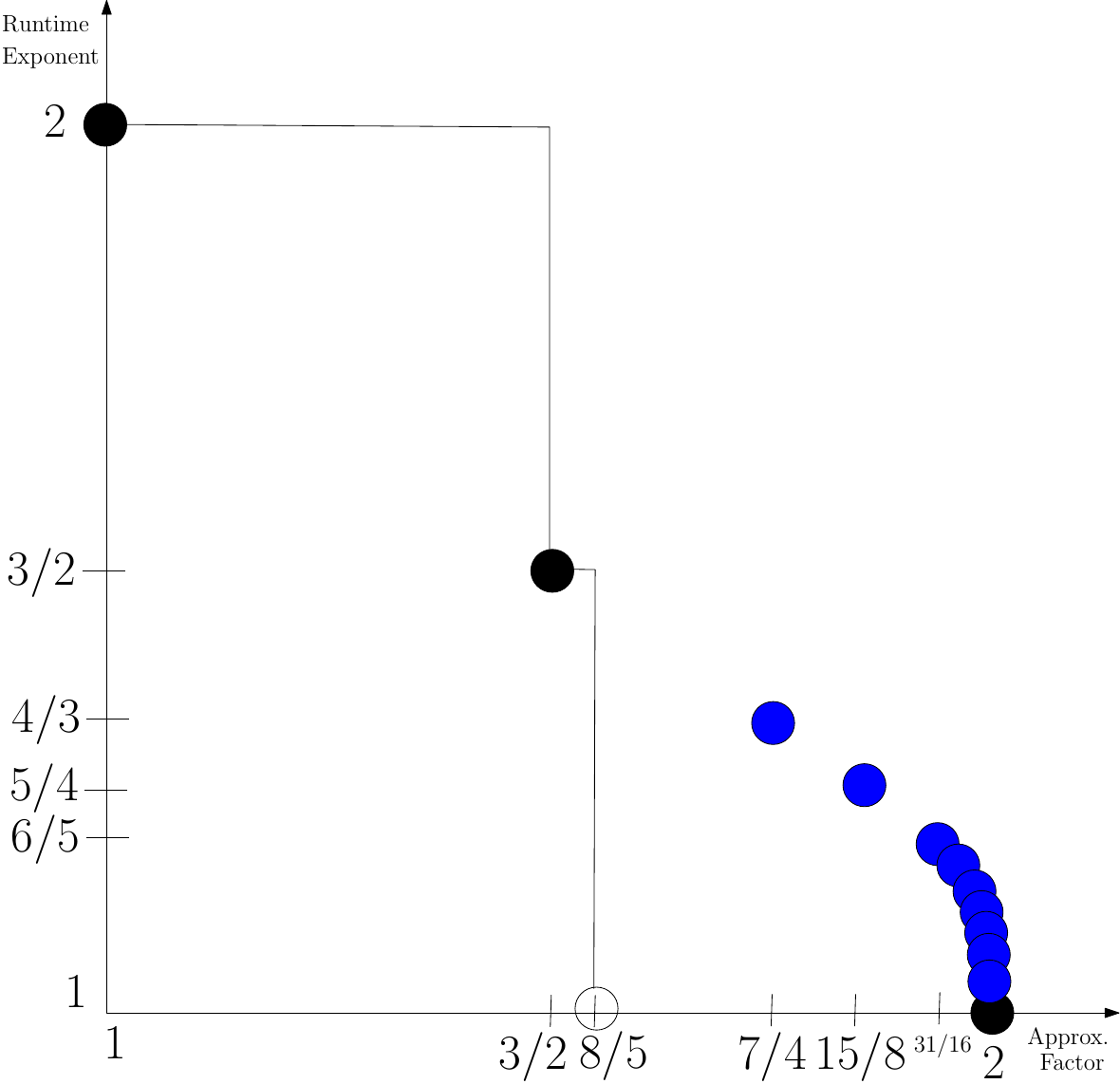}
		\caption{Undirected unweighted Diameter}
		\label{fig:unundiam}
    \end{subfigure}%
    \begin{subfigure}[b]{0.33\textwidth}
        \includegraphics[width=\textwidth]{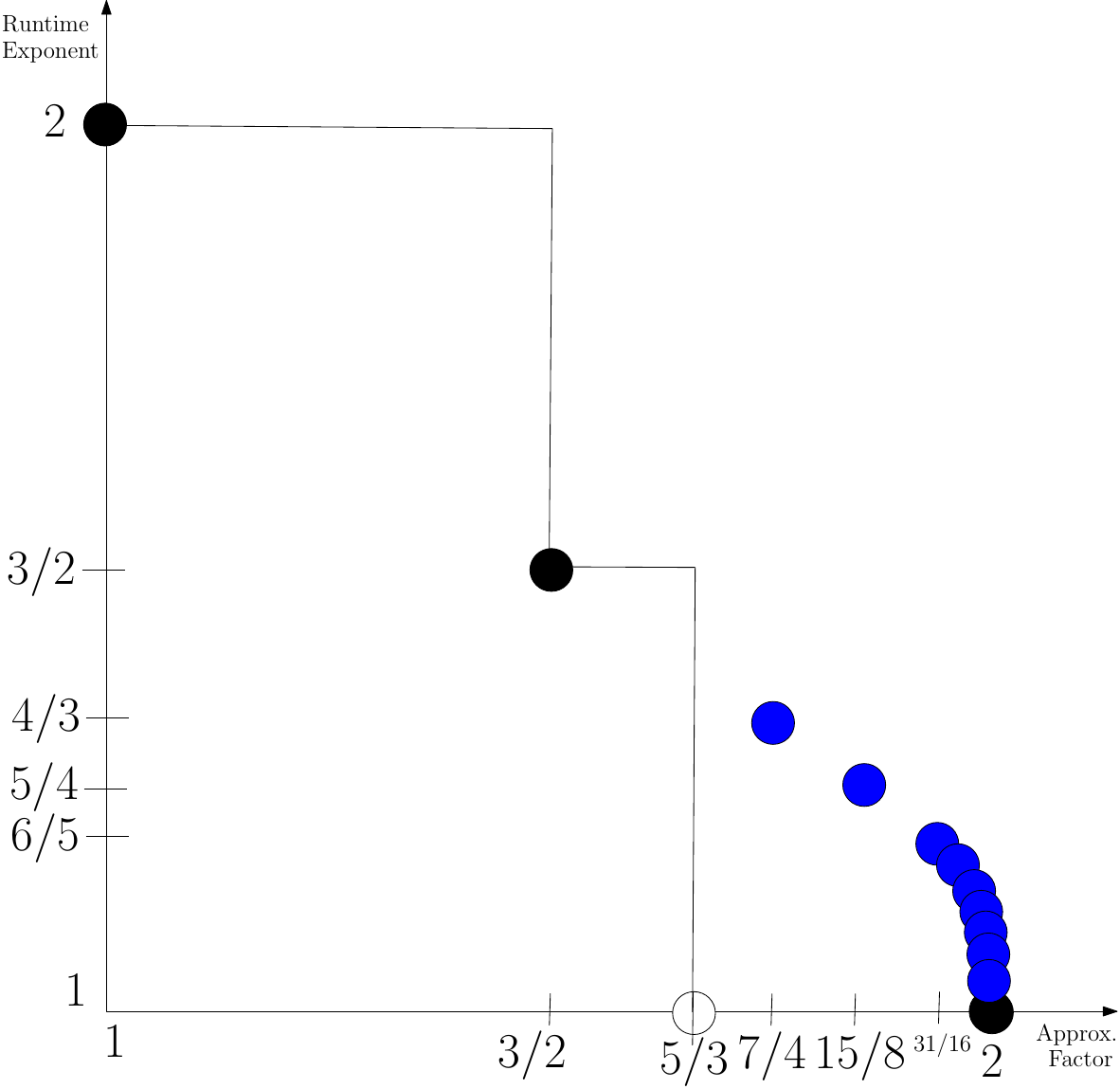}
		\caption{Undirected weighted Diameter}
		\label{fig:unwtddiam}
    \end{subfigure}
	~
	\begin{subfigure}[b]{0.32\textwidth}
        \includegraphics[width=\textwidth]{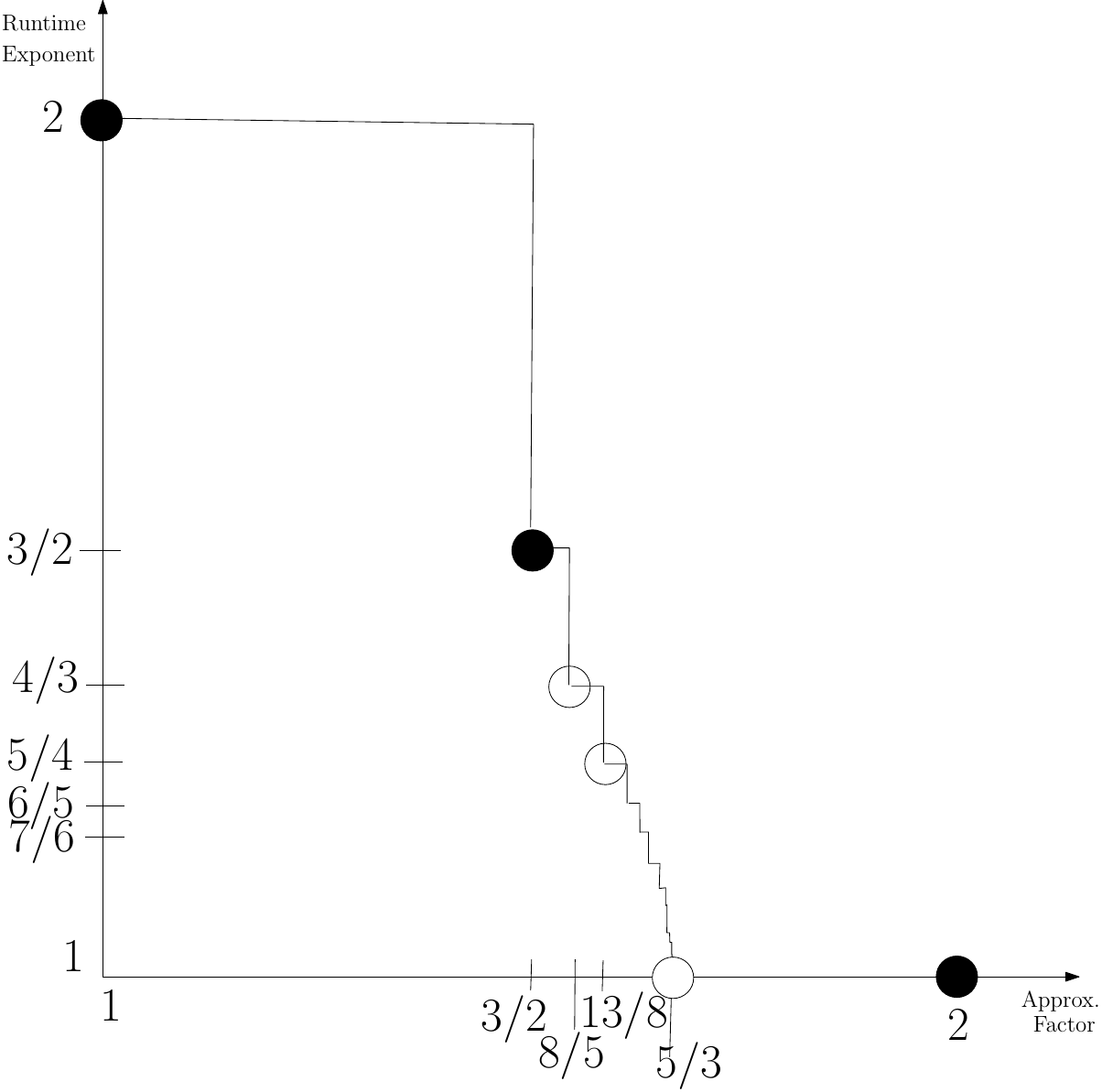}
		\caption{Directed unweighted Diameter}
		\label{fig:dirundiam}
    \end{subfigure}
    \caption{\small Our hardness results for Diameter. The $x$-axis is the approximation factor and the $y$-axis is the runtime exponent. Black lines represent lower bounds. Black dots represent existing algorithms. Blue dots represent existing algorithms whose approximation is potentially off by an additive term (the algorithms of~\cite{cairo}). Transparent dots represent algorithms that might exist and would be tight with our lower bounds. }\label{fig:diam1}
\end{figure*}

We address our Main Question for Eccentricities as well. Our main result for Eccentricities is Theorem~\ref{thm:undirecchard}. Its first consequence is as follows:

\begin{theorem}[$5/3$-Eccentricities Alg. is Tight]\label{thm:95eccintro}
Under SETH, no $O(n^{3/2-\delta})$ time algorithm for $\delta>0$ can output a 
$9/5-\eps$ approximation for $\eps>0$ for the Eccentricities of  an undirected unweighted sparse graph.
\end{theorem}

In other words, the $\tilde{O}(m^{3/2})$ time $5/3$-approximation algorithm of \cite{RV13,ChechikLRSTW14} is tight in two ways. Improving the approximation ratio to $5/3-\eps$ causes a runtime blow-up to $n^{2-o(1)}$ (\cite{AbboudWW16}) and improving the runtime to $O(m^{3/2-\delta})$ causes an approximation ratio blow-up to $9/5$.

More generally, we prove (in Theorem~\ref{thm:undirecchard}): for every $k\geq 2$, under SETH, distinguishing between Eccentricities $2k-1$ and $4k-3$ in unweighted undirected sparse graphs requires $n^{1+1/(k-1)-o(1)}$ time. Thus, no near-linear time algorithm can achieve a $2-\eps$-approximation for Eccentricities for $\eps>0$.

The best (folklore) near-linear time approximation algorithm for Eccentricities currently only achieves a $3$-approximation, and only in undirected graphs. There is no known constant factor approximation algorithm for directed graphs! Is our limitation result for linear time Eccentricity algorithms far from the truth?

We show that our lower bound result is essentially tight, for both directed and undirected graphs by producing the first non-trivial near-linear time approximation algorithm for the Eccentricities in weighted directed graphs (Theorem~\ref{alg2}).
\begin{theorem}[$2$-Approx. for Eccentricities in near-linear time.]
Under SETH, no $n^{1+o(1)}$ time algorithm can output a
$2-\eps$ approximation for $\eps>0$ for the Eccentricities of an undirected unweighted sparse graph.

For every $\delta>0$, there is a randomized $\Ot(m/\delta)$ time algorithm that with high probability produces a $(2+\delta)$-approximation for the Eccentricities of any directed weighted graph.\label{thm:2ecc}
\end{theorem}

The approximation hardness result is the first result within fine-grained complexity that gives tight hardness for {\em near linear time} algorithms.

The $2+\delta$ approximation ratio that our algorithm produces beats {\em all} approximation ratios for Eccentricities given by Cairo et al.~\cite{cairo}. It also constitutes the first known constant factor approximation algorithm for Eccentricities in directed graphs.

Our approximation algorithm also implies as a corollary an approximation algorithm for the Source Radius problem\footnote{The Source Radius problem is a natural extension of the undirected Radius definition. The goal is to return $\min_x\max_vd(x,v)$.} studied in~\cite{AbboudWW16} with the same runtime and approximation factor ($2+\delta$). Abboud et al.~\cite{AbboudWW16} showed that, under the Hitting Set Conjecture, any $(2-\eps)$-approximation algorithm (for $\eps>0$) for Source Radius requires $n^{2-o(1)}$ time, and hence our Source Radius algorithm is also essentially tight.

Our lower bound in Theorem~\ref{thm:2ecc} holds already for undirected unweighted graphs, and the upper bound works even for directed weighted graphs. The algorithm produces a $(2+\delta)$-approximation, which while close, is not quite a $2$-approximation.
We design (in Theorem~\ref{alg1}) a genuine $2$-approximation algorithm running in $\tilde{O}(m\sqrt n)$ time that also works for directed weighted graphs.
We then complement it (in Theorem~\ref{thm:direcchard}) with a tight lower bound under SETH: in sparse directed graphs, if you go below factor $2$ in the accuracy, the runtime blows up to quadratic.

\begin{theorem}[Tight $2$-Approx. for Eccentricities]\label{thm:otherecc}
Under SETH, no $n^{2-\delta}$ time algorithm for $\delta>0$ can output a $2-\eps$ approximation for the Eccentricities of a directed unweighted sparse graph.

There is a randomized $\Ot(m\sqrt n)$ time algorithm that with high probability produces a $2$-approximation for the Eccentricities of any directed weighted graph.
\end{theorem}

We thus give an essentially complete answer to our Main Question for Eccentricities. Our results are summarized in Figures~\ref{fig:ununecc} and \ref{fig:dirundecc}.

\begin{figure*}[t!]
    \centering
    \begin{subfigure}[b]{0.33\textwidth}
        \includegraphics[width=\textwidth]{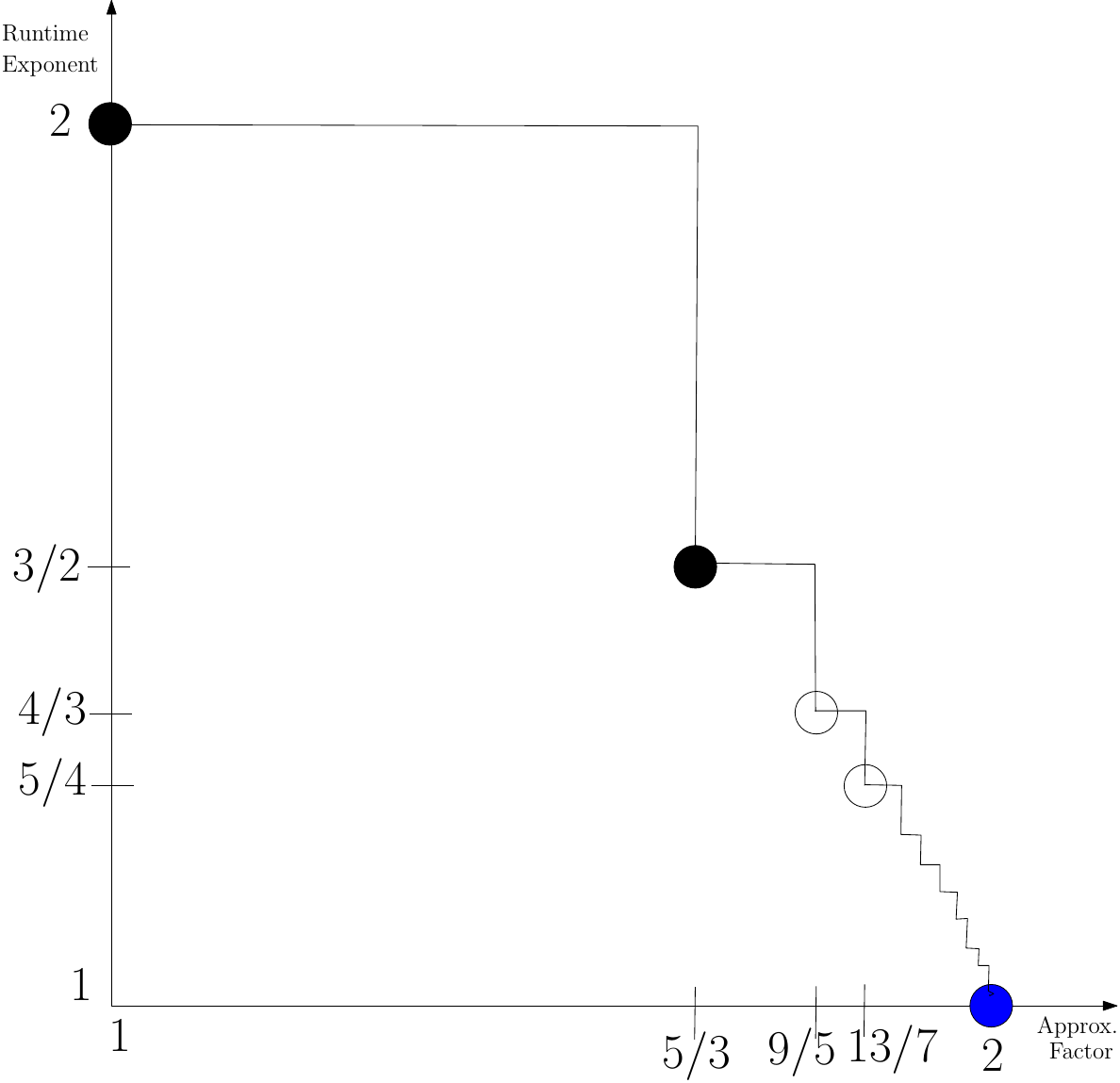}
		\caption{Undirected Eccentricities}
		\label{fig:ununecc}
    \end{subfigure}%
    \hspace{1em}
    \begin{subfigure}[b]{0.33\textwidth}
		\includegraphics[width=\textwidth]{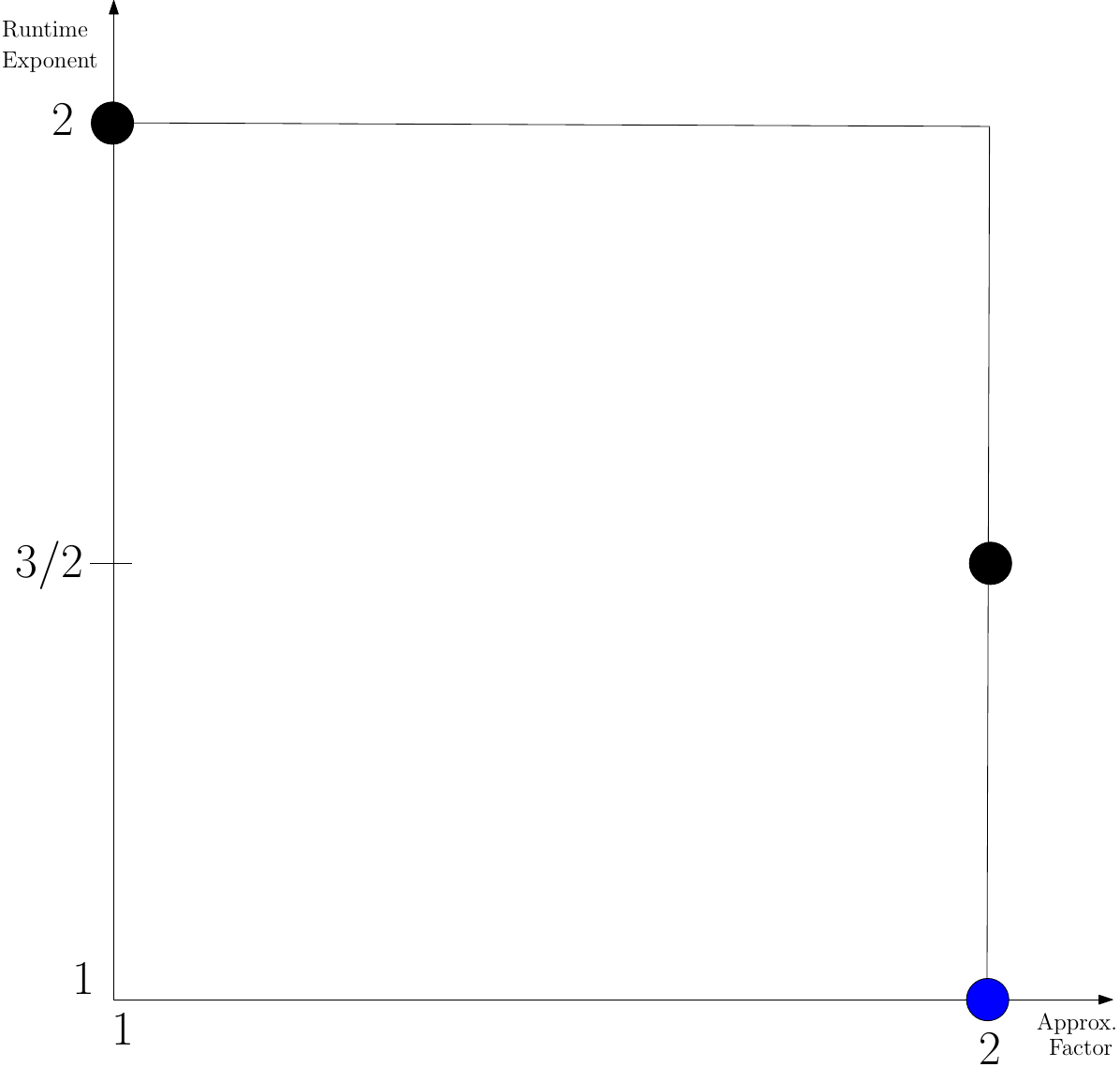}
		\caption{Directed Eccentricities}
		\label{fig:dirundecc}
    \end{subfigure}
    \caption{\small Our algorithms and hardness results for Eccentricities. The lower bounds are for unweighted graphs and the upper bounds are for weighted graphs. The $x$-axis is the approximation factor and the $y$-axis is the runtime exponent. Black lines represent lower bounds. Black dots represent existing algorithms (including our algorithm at $(2,3/2)$ in figure b). Blue dots represent existing algorithms whose position may not be exactly as it appears in the figure. Here, the blue dots represent our $(2+\delta)$-approximation algorithm running in $\tilde{O}(m/\delta)$ time. Transparent dots represent algorithms that might exist and would be tight with our lower bounds.}
\end{figure*}

Our conditional lower bounds for both Diameter and Eccentricities are all based on a common construction: a conditional lower bound for a problem called $S$-$T$ Diameter. In $S$-$T$ Diameter, the input is an undirected graph $G=(V,E)$ and two subsets $S,T\subseteq V$, not necessarily disjoint, and the output is $D_{S,T}:=\max_{s\in S,t\in T} d(s,t)$.

$S$-$T$ Diameter may be a problem of independent interest. It is related to the bichromatic furthest pair problem studied in geometry (e.g. as in~\cite{katoh}), but for graphs (if we set $T=V\setminus S$).

It is easy to see that if one can compute the $S$-$T$ Diameter, then one can also compute the Diameter in the same time: just set $S=T=V$. We show that actually, when it comes to exact computation, the $S$-$T$ Diameter and Diameter in weighted graphs are computationally equivalent (Theorem~\ref{thm:stdiamequiv}).

We show that $S$-$T$ Diameter also has similar approximation algorithms to Diameter. We give a $3$-approximation running in linear time (Claim~\ref{claim:3appxstdiam} based on the folklore Diameter $2$-approximation algorithm), and a $2$-approximation running in $\tilde{O}(m^{3/2})$ time (Theorem~\ref{thm:stdiam32} based on the $3/2$-approximation algorithm of \cite{RV13,ChechikLRSTW14}).

We prove the following lower bound for $S$-$T$ Diameter (restated as Theorem~\ref{thm:hardness}), the proof of which is
the starting point for all of our conditional lower bounds.

\begin{theorem}\label{thm:stintro}
Under SETH, for every $k\geq 2$, every algorithm that can distinguish between $S$-$T$ Diameter $k$ and $3k-2$ in undirected unweighted graphs requires $n^{1+1/(k-1)-o(1)}$ time.
\end{theorem}

Theorem~\ref{thm:stintro} implies that under SETH, our aforementioned $2$ and $3$-approximation algorithms are optimal.

For all of our lower bounds, we also address the question of whether they can be extended to higher values of Diameter and Eccentricities. All of our lower bounds, with the exception of directed Eccentricities, are of the form ``any algorithm that can distinguish between Diameter (or Eccentricity) $a$ and $b$ requires a certain amount of time" for small values of $a$ and $b$. This doesn't exclude the possibility of an algorithm that distinguishes between higher Diameters (or Eccentricities) of the same ratio i.e. between $a\ell$ and $b\ell$ for some $\ell$. For weighted Diameter, our lower bound easily extends to higher values of Diameter by simply scaling up the edge weights. For $S-T$ Diameter and undirected Eccentricities, our lower bounds easily extend to higher values of Diameter and Eccentricities by simply subdividing the edges. For unweighted directed Diameter, our lower bound extends to higher values of Diameter with a slight loss in approximation factor by subdividing some of the edges. For unweighted undirected Diameter, our lower bound does not seem to easily extend to higher values of Diameter.

\paragraph{Dense graphs.}
Can we address our Main Question for dense graphs as well? In particular, can we extend our runtime lower bounds of the form $n^{1+1/\ell-o(1)}$ to $mn^{1/\ell-o(1)}$, thus matching the known algorithms for larger values of $m$? We show that the answer is ``no''. For undirected unweighted graphs, we obtain $\tilde{O}(n^2)$ time algorithms for Diameter
achieving an almost $3/2$-approximation (Theorem~\ref{thm:densediam}), and for all Eccentricities achieving an almost $5/3$-approximation algorithm (Theorem~\ref{thm:fastdenseecc}). These algorithms run in near-linear time in dense graphs, improving the previous best runtime of $\tilde{O}(m\sqrt n)$ by Roditty and Vassilevska W.~\cite{RV13}, and subsuming (for dense unweighted graphs) the results of Cairo et al.~\cite{cairo}.

\begin{theorem}\label{thm:denseintro} There is an expected $O(n^2\log n)$ time algorithm that for any undirected unweighted graph with Diameter $D=3h+z$ for $h\geq 0,z\in\{0,1,2\}$, returns an extimate $D'$ such that $2h-1\leq D'\leq D$ if $z=0,1$ and $2h\leq D'\leq D$ if $z=2$.

There is an expected $O(n^2\log n)$ time algorithm that for any undirected unweighted graph returns estimates $\eps'(v)$ of the Eccentricities $\eps(v)$ of all vertices such that $3\eps(v)/5-1\leq \eps'(v)\leq \eps(v)$ for all $v$.
\end{theorem}

We also show (in Theorem~\ref{thm:205}) that one can improve the estimates slightly with an $O(n^{2.05})$ time algorithm. 

%% file: related.tex
\subsection{Related work}\label{sec:rel}
The fastest known algorithm for APSP in dense weighted graphs  is by R. Williams~\cite{ryanapsp} and runs in $O(n^3/2^{\Theta(\sqrt{\log n})})$ time.
For sparse undirected graphs, the fastest known APSP algorithm is by Pettie~\cite{Pettie04} running in $O(mn+n^2\log\log n)$ time.
The fastest APSP algorithm for sparse undirected weighted graphs is by Pettie and Ramachandran~\cite{PettieR05} and runs in $O(mn\log \alpha(m,n))$ time. For APSP on undirected unweighted graphs with $m>n\log\log n$, Chan~\cite{chan06j} presented an $O(mn \log\log n/\log n)$ time algorithm.
In graphs with small integer edge weights bounded in absolute
value by $M$, APSP can be computed in $\tilde{O}(Mn^\omega)$ time (by Shoshan and Zwick~\cite{sz99}
building upon Seidel~\cite{Seidel} and Alon, Galil and Margalit~\cite{AlonGM97}) in undirected graphs and in $\tilde{O}(M^{0.681}n^{2.5302})$ time
(by Zwick~\cite{zwickbridge}) in directed graphs. Zwick~\cite{zwickbridge} also showed that APSP in directed weighted graphs admits
an $(1+\eps)$-approximation algorithm for any $\eps>0$, running in time $\tilde{O}(n^\omega/\eps \log(M/\eps))$.
For Diameter in graphs with integer edge weights bounded by $M$, Cygan et al.~\cite{cyganbaur} obtained an algorithm running in time $\tilde{O}(Mn^\omega)$.

The pioneering work of Aingworth et al.~\cite{aingworth} on Diameter and shortest paths approximation was the root to many
subsequent works. Building upon Aingworth et al.~\cite{aingworth}, Dor, Halperin and Zwick~\cite{DorHZ00} presented additive approximation
algorithms for APSP in undirected unweighted graphs, achieving among other things, an additive $2$-approximation in $\tilde{O}(n^{7/3})$ time
(notably, the best known bound on $\omega$ is $>7/3$). They  also presented an $\tilde{O}(n^2)$ time additive $O(\log n)$-approximation algorithm.
These algorithms were generalized by Cohen and Zwick~\cite{CoZw01} who showed that in undirected weighted graphs APSP has a (multiplicative)
$3$-approximation in $\tilde{O}(n^2)$ time, a $7/3$-approximation in $\tilde{O}(n^{7/3})$ time,
and a $2$-approximation in $\tilde{O}(n \sqrt{mn})$ time.
Baswana and Kavitha~\cite{BaKa10} presented an $\tilde{O}(m\sqrt n+n^2)$ time multiplicative $2$-approximation algorithm and an $\tilde{O}(m^{2/3} n+n^2)$
time $7/3$-approximation algorithm for APSP in weighted undirected graphs.

Spanners are closely related to shortest paths approximation. A subgraph $H$ is an $(\alpha, \beta)$-spanner of $G=(V,E)$
if for every $u,v\in V$, $d_H(u,v) \leq \alpha\cdot d_G(u,v) + \beta$, where $d_{G'}(u,v)$ is the distance between $u$ and $v$ in $G'$.
Any weighted undirected graph has a $(2k-1,0)$-spanner with $O(n^{1+1/k})$ edges~\cite{AlthoferDDJS93}. Baswana and Sen~\cite{BaswanaS07} presented a randomized linear time algorithm for constructing a  $(2k-1,0)$-spanner with $O(kn^{1+1/k})$ edges.
Dor, Halperin and Zwick~\cite{DorHZ00} showed that a $(1,2)$-spanner with $O(n^{1.5})$ edges can be constructed in $\tilde{O}(n^2)$ time.
Elkin and Peleg~\cite{ElkinP04} showed that for every integer $k\ge 1$ and $\eps>0$ there is a $(1+\eps, \beta)$-spanner with $O(\beta n^{1+1/k})$ edges, where $\beta$ depends on $k$ and $\eps$ but is independent of $n$.
Baswana et. al~\cite{BaswanaKMP10} presented a $(1,6)$-spanner with $O(n^{4/3})$ edges. Woodruff~\cite{Woodruff} presented an $\tilde{O}(n^2)$ time algorithm that computes a $(1,6)$-spanner with $O(n^{4/3})$ edges. Chechik~\cite{Chechik13} presented a $(1,4)$-spanner with $O(n^{7/5})$ edges.
Recently, Abboud and Bodwin\cite{AbBo16} showed that there is no additive spanner with constant error and $O(n^{4/3-\eps})$ edges.

Thorup and Zwick~\cite{ThZw05} introduced the notion of distance oracles, a data structure that stores approximate distances for a weighted undirected graph. Thorup and Zwick designed a distance oracle that for any $k$ takes $O(mn^{1/k})$ time to construct and, is of size $O(kn^{1+1/k})$, and given a pair of vertices $u,v \in V$ it returns in $O(k)$ time a $(2k-1)$-approximation for $d(u,v)$. Baswana and Sen~\cite{BaswanaS06} showed that the construction time can be made $O(n^2)$ for
unweighted graphs. Baswana and Kavitha~\cite{BaKa10} extended the $O(n^2)$ construction time to
weighted graphs. Subsequently, Baswana, Gaur, Sen, and Upadhyay~\cite{BaswanaGSU08} obtained subquadratic construction
time in unweighted graphs, at the price of having additive constant error in addition to the $2k-1$ multiplicative error.

 Chechik~\cite{Chechik15} gave an oracle with space $O(n^{1+1/k})$ and $O(1)$ query time, which like previous work, returns a $(2k-1)$-approximation.
P\v{a}tra\c{s}cu and Roditty~\cite{patrascuroditty} obtained a distance oracle that
uses $\tilde{O}(n^{5/3})$ space, has $O(1)$ query time, and returns a $(2k+1)$-approximation. Sommer~\cite{Sommer16} presented an $\tilde{O}(n^2)$ time algorithm that constructs such a distance oracle. The construction time was recently improved to $O(n^2)$ by Knudsen~\cite{Knudsen17}.
P\v{a}tra\c{s}cu et. al~\cite{PatrascuRT12} presented infinitely many distance oracles with fractional approximation factors that for graphs with $m=\tilde{O}(n)$ converge exactly to the integral stretch factors and the corresponding space bound of Thorup and Zwick.
Thorup and Zwick~\cite{ThZw01} also extended their techniques from \cite{ThZw05} to compact routing schemes.

The lower bounds presented in this paper were inspired by a lower bound by P\v{a}tra\c{s}cu and Roditty~\cite{patrascuroditty}
who showed conditional hardness based on a conjecture on the hardness of a set intersection problem
for the space usage of any distance oracle that can distinguish between distances $3$ and $7$.

\paragraph{Subsequent work}

Our approximation algorithms for Eccentricities in Theorems~\ref{thm:2ecc} and~\ref{thm:otherecc} have recently been improved to a true 2-approximation in $\tilde{O}(m)$ time by Choudhary and Gold~\cite{choudhary2020extremal}. From the lower bounds side, there have been several very recent improvements for Diameter. Li~\cite{3vs5} improved our unweighted undirected construction for Diameter to match our weighted undirected construction. That is, he showed that under SETH, any $5/3-\varepsilon$ approximation algorithm for Diameter in undirected unweighted graphs requires $m^{3/2-o(1)}$ time. Then, Bonnet~\cite{4ov} surpassed this $5/3$ bound for \emph{directed weighted} graphs, by showing that under SETH, any $7/4-\varepsilon$ approximation algorithm for Diameter requires $m^{4/3-o(1)}$ time. Then, concurrent and independent work by Dalirrooyfard and Wein~\cite{dalirrooyfard2020tight}, and Li~\cite{3vs5} closed the gap for near-linear time algorithms for Diameter in \emph{directed unweighted} graphs, by showing that for any fixed integer $k\ge 2$, SETH implies that for all $\delta>0$, any $(\frac{2k-1}{k}-\delta)$-approximation algorithm for Diameter in an unweighted directed graph on $m$ edges requires $m^{\frac{k}{k-1}-o(1)}$ time. This implies that the folklore $\tilde{O}(m)$ time 2-approximation is tight. Li also showed that better SETH-based reductions are impossible, assuming other nondeterministic versions of SETH. Additionally, there have been two surveys with sections devoted to Diameter approximation~\cite{rubinstein2019seth,williams2018some}.

%% file: organization_journal.tex
\subsection{Organization}
In Section~\ref{sec:st_hard} we prove our lower bounds for $S$-$T$ Diameter which serve as a basis for the rest of our lower bounds. We also show equivalence between Diameter and $S$-$T$ Diameter. In Section~\ref{sec:ecchardness} we prove our lower bounds for Eccentricities: one for directed graphs and one for undirected graphs. In Section~\ref{sec:diam_lb} we prove our lower bounds for Diameter. This section is divided into four subsections, one for each of the results in Table~\ref{tab}. In Section~\ref{sec:eccalgs} we describe our algorithms for sparse graphs: our 2-approximation and $(2+\delta)$-approximation for Eccentricities, our 2-approximation and 3-approximation for $S$-$T$ Diameter, and our less than 2-approximation for Diameter. In Section~\ref{sec:diamnewalgs} we describe our algorithms for dense graphs: our nearly $3/2$-approximations for Diameter and our nearly $5/3$-approximations for Eccentricities.

%% file: preliminaries.tex
\section{Preliminaries}\label{sec:prelims}
%

Let $G=(V,E)$ be a weighted or unweighted, directed or undirected graph, where $|V|=n$ and $|E|=m$. For every $u,v\in V$ let $d_G(u,v)$ be the length of the shortest path from $u$ to $v$. When the graph $G$ is clear from the context we omit the subscript $G$.

The eccentricity $\eps(v)$ of a vertex $v$ is defined as $\max_{u\in V} d(v,u)$. The diameter $D$ of a graph is $\max_{v\in V}\eps(v)$.
In a directed graph we additionally let $\eps^{in}(v)=\max_{u\in V} d(u,v)$. In a directed graph, we sometimes use $\eps^{out}(v)$ to denote $\eps(v)$ to emphasize the distinction between $\eps^{in}(v)$ and $\eps^{out}(v)$. For all of these definitions, a distance $d(u,v)$ is considered to be $\infty$ if $v$ is not reachable from $u$.

Let $deg(v)$ be the degree of $v$ and let $N_s(u)$ be the set of the $s$ closest vertices of $v$, where ties are broken by taking the vertex with the smaller ID.
In  a directed graph let $deg^{out}(v)$  (resp., $deg^{in}(v)$) be the outgoing (incoming) degree of $v$.
Let $\Nout_s(v)$ (resp., $\Nin_s(v)$) be the set of the $s$ closest  outgoing (incoming) vertices of $v$, where ties are broken by taking the vertex with the smaller ID. For a subset $S\subseteq V$ of vertices and a vertex $v \in V$ we write $d(S,v):=\min_{s \in S}d(s,v)$ to denote the distance from the set $S$ to the vertex $v$.

Let $k\geq 2$. The $k$-Orthogonal Vectors Problem ($k$-OV) is as follows: given $k$  sets $S_1,\ldots,S_k$, where each $S_i$ contains $N$ vectors in $\{0,1\}^d$, determine whether there exist $v_1\in S_1,\ldots,v_k\in S_k$ so that their generalized inner product is $0$, i.e. $\sum_{i=1}^d \prod_{j=1}^k v_j[i]=0$.

R. Williams~\cite{TCS05} (see also \cite{ipecsurvey}) showed that if for some $\eps>0$ there is an $N^{k-\eps}\poly(d)$ time algorithm for $k$-OV, then CNF-SAT on formulas with $N'$ variables and $M$ clauses can be solved in $2^{N'(1-\eps/k)}\poly(M)$ time. In particular, such an algorithm would contradict
the Strong Exponential Time Hypothesis (SETH) of Impagliazzo, Paturi and Zane~\cite{ipz2} which states that for every $\eps>0$ there is a $K$ such that $K$-SAT on $N'$ variables cannot be solved in $2^{(1-\eps)N'}\poly N'$
time (say, on a word-RAM with $O(\log N')$ bit words). 

This also motivates the following $k$-OV Conjectures (implied by SETH) for all constants $k\geq 2$: $k$-OV requires $N^{k-o(1)}$ time on a word-RAM with $O(\log N)$ bit words. Most of our conditional lower bounds are based on the $k$-OV Conjecture for a particular constant $k$, and thus they also hold under SETH. 

A main motivation behind SETH is that despite decades of research, the best upper bounds for $K$-SAT on $N'$ variables and $M$ clauses remain of the form $2^{N'(1-c/K)}\poly{(M)}$ for constant $c$ (see e.g. \cite{hirschsat,PPSZ05,Scho99sat}). The best algorithms for the $k$-OV problem for any constant $k\geq 2$ on $N$ vectors and dimension $c\log N$ run in time $N^{k-1/O(\log c)}$ (Abboud, Williams and Yu~\cite{AbboudWY15} and Chan and Williams~\cite{ChanW16}).

%
%

%% file: ST-Diameter-LBs.tex
\section{$S$-$T$ Diameter hardness}\label{sec:st_hard}
In this section we will prove that under SETH, our $S$-$T$ Diameter algorithms are essentially optimal. Our $S$-$T$ Diameter construction serves as the basis for all of our conditional lower bounds.
We prove the following theorem, which implies Theorem~\ref{thm:stintro}.

\begin{theorem}
Let $k\geq 2$ be an integer. There is an $O(k N^{k-1} d^{k-1})$ time reduction that transforms any instance of $k$-OV on sets of $N$ $d$-dimensional vectors into an unweighted undirected graph on $O(N^{k-1}+k N^{k-2} d^{k-1})$ vertices and $O(k N^{k-1} d^{k-1})$ edges and two disjoint sets $S$ and $T$ on $N^{k-1}$ vertices each, so that if the $k$-OV instance has a solution, then $D_{S,T}\geq 3k-2$, and if it does not, $D_{S,T}\leq k$.\label{thm:hardness}
\end{theorem}
%

From Theorem~\ref{thm:hardness} we get that if there is some $k\geq 2$, $\eps>0$ and $\delta>0$ so that there is an $O(m^{1+1/(k-1)-\eps})$ time $(3-2/k-\delta)$-approximation algorithm for $S$-$T$ Diameter in $m$-edge graphs, then $k$-OV has an $n^{k-\gamma}\poly(d)$-time algorithm for some $\gamma>0$ and SETH is false.

We obtain an immediate corollary.
\begin{corollary} 
For $S$-$T$ Diameter,
under SETH, there is 
\begin{itemize}
\item no $O(m^{2-\delta})$ time $(2-\eps)$-approximation algorithm for any $\eps>0,\delta>0$,
\item no $O(m^{3/2-\delta})$ time $(7/3-\eps)$-approximation algorithm for any $\eps>0,\delta>0$,
\item no $m^{1+o(1)}$ time, $(3-\eps)$-approximation algorithm for any $\eps>0$. 
\end{itemize}
\end{corollary}

\subsection{Warm-up: Construction for $k=3$}
We first review the construction of~\cite{RV13} for $S$-$T$ Diameter for $k=2$. Their construction is explicitly written as a construction for the standard Diameter problem, but it implicitly gives a construction for $S$-$T$ Diameter. Then, we will describe our construction for $k=3$. This will provide some intuition for the general construction.

\paragraph{Review of the $k=2$ case.} In the $k=2$ case, we are given an OV instance consisting of sets $W_0,W_1\subseteq \{0,1\}^d$, each of size $N$. Our goal is to construct a graph on $\tilde{O}(N)$ vertices and edges so that if the OV instance is a NO instance then the $S$-$T$ Diameter is 2, and if the OV instance is a YES instance then the $S$-$T$ Diameter is at least 4. 

We construct a layered graph $G$ on three layers $L_0, L_1, L_2$  where edges only go between adjacent layers. We set $S=L_0$ and $T=L_2$ for the $S$-$T$ Diameter instance. $L_0$ consists of one vertex for each vector $a\in W_0$, and $L_2$ consists of one vertex for each vector $b\in W_1$. $L_1$ consists of one vertex for each coordinate in $[d]$. 

There is an edge between $a\in L_0$ and $i\in L_1$ if and only if $a$ is 1 in coordinate $i$. There is an edge between $b\in L_2$ and $i\in L_1$ if and only if $b$ is 1 in coordinate $i$. This completes the description of the construction.

If the OV instance is a NO instance, then by definition, for every pair $a\in W_0$, $b\in W_1$, there exists a coordinate $x$ that is 1 for both $a$ and $b$. Thus, there is a path of length 2 in $G$ from $a\in L_0$ to  $b\in L_2$ through $x\in L_1$. On the other hand, if the OV instance is a YES instance with orthogonal pair $a\in W_0$, $b\in W_1$, then by definition there is no coordinate such that $a$ and $b$ are both 1. Therefore, the distance between $a\in L_0$ and $b\in L_2$ is more than 2, and it must be at least 4 due to the layered structure of the graph.

\paragraph{The $k=3$ case.} We are given a $3$-OV instance consisting of sets $W_0, W_1, W_2\subseteq \{0,1\}^d$, each of size $N$. Our goal is to construct a graph on $\tilde{O}(N^2)$ vertices and edges so that if the $3$-OV instances is a NO instance, the $S$-$T$ Diameter is 3, and if the $3$-OV instance is a YES instance, the $S$-$T$ Diameter is at least 7. Our construction is shown in Figure~\ref{fig:k3}.

We construct a layered graph $G$ on four layers $L_0, L_1, L_2, L_3$ where the edges go only between adjacent layers. We will set $S=L_0$ and $T=L_3$ for the $S-T$ Diameter instance. $L_0$ consists of one vertex for each pair of vectors $a_0\in W_0$, $a_1\in W_1$, and $L_3$ consists of one vertex for each pair of vectors $b_1\in W_1$, $b_2\in W_2$. 

Now, we would like to define $L_1$, $L_2$, and the edges so that the $S$-$T$ Diameter is 3 if and only if the $3$-OV instance is a NO instance. To provide some intuition, we fix a pair of vertices $(a_0,a_1)\in S$, $(b_1,b_2)\in T$ and ask the question: what can we say about the vectors $a_0$, $a_1$, $b_1$, and $b_2$ in a NO instance? By definition, in a NO instance the vectors $a_0$, $a_1$, and $b_2$ are all 1 at some coordinate $x_0$. Similarly, the vectors $a_0$, $b_1$, and $b_2$ are all 1 at some coordinate $x_1$. Because the $S$ side of the graph concerns the vectors $a_0$ and $a_1$ and the $T$ side of the graph concerns the vectors $b_1$ and $b_2$, we separate the conditions on $a_0$, $a_1$, $b_1$, and $b_2$ according to each side of the graph. For the $S$ side, we have that $a_0[x_0]=a_1[x_0]=a_0[x_1]=1$. For the $T$ side, we have that $b_1[x_1]=b_2[x_1]=b_2[x_0]=1$.

This motivates a first attempt for how to define the rest of the graph. Suppose $L_1$ and $L_2$ both consist of one vertex for every pair of coordinates $x_0,x_1\in [d]$. We add an edge from $(a_0,a_1)\in L_0$ to $(x_0,x_1)\in L_1$ if $a_0[x_0]=a_1[x_0]=a_0[x_1]=1$. We add an edge from $(b_1,b_2)\in L_3$ to $(x_0,x_1)\in L_2$ if $b_1[x_1]=b_2[x_1]=b_2[x_0]=1$. Finally, we add an edge from $(x_0,x_1)\in L_1$ to $(x_0',x_1')\in L_2$ if $x_0=x_0'$ and $x_1=x_1'$. While this construction has $S$-$T$ Diameter 3 for a NO instance of $3$-OV, it does not have $S$-$T$ Diameter 7 for a YES instance, as we would like. In particular, suppose $a_0\in W_0$, $a_1\in W_1$, $a_2\in W_2$ is an orthogonal triple. We would like the distance between $(a_0,a_1)\in L_0$ and $(a_1,a_2)\in L_3$ to be at least 7, however, with the current construction, there could be a path of length 5 from $(a_0,a_1)\in L_0$ to some $(x_0,x_1)\in L_1$, to some $(a_0',a_1)\in L_0$, to some $(x_0',x_1')\in L_1$, to $(x_0',x_1')\in L_2$, to $(a_1,a_2)\in L_3$. The issue is that from $(a_0,a_1)\in L_0$ we can reach $(a_0',a_1)\in L_0$ with a path of length 2 for some convenient choice of $a_0'$. 

To overcome this issue, we also include the vector $a_0$ in the representation of vertices in $L_1$; that is, $L_1$ consists of one vertex for every triple $(a_0\in W_0,x_0\in[d],x_1\in[d])$. There is an edge from $(a_0,a_1)\in L_0$ to $(a_0',x_0,x_1)\in L_1$ if and only if $a=a'$ and $a_0[x_0]=a_1[x_0]=a_0[x_1]=1$. Symmetrically, $L_2$ consists of one vertex for every triple $(b_2\in W_2,x_0\in[d],x_1\in[d])$ and there is an edge from $(b_1,b_2)\in L_3$ to $(b_2',x_0,x_1)\in L_2$ if and only if $b_2=b_2'$ and $b_1[x_1]=b_2[x_1]=b_2[x_0]=1$. Lastly, there is an edge between $(a_0,x_0,x_1)\in L_1$ and $(b_2,x_0',x_1')\in L_2$ if and only if $x_0=x_0'$ and $x_1=x_1'$. This completes the description of the construction. One can verify that this construction indeed satisfies the property that the $S$-$T$ Diameter is 3 for a NO instance of $3$-OV, and the $S$-$T$ Diameter is 7 for a YES instance of $3$-OV.

\begin{figure*}[h]
  \centering
    \includegraphics[width=0.575\textwidth]{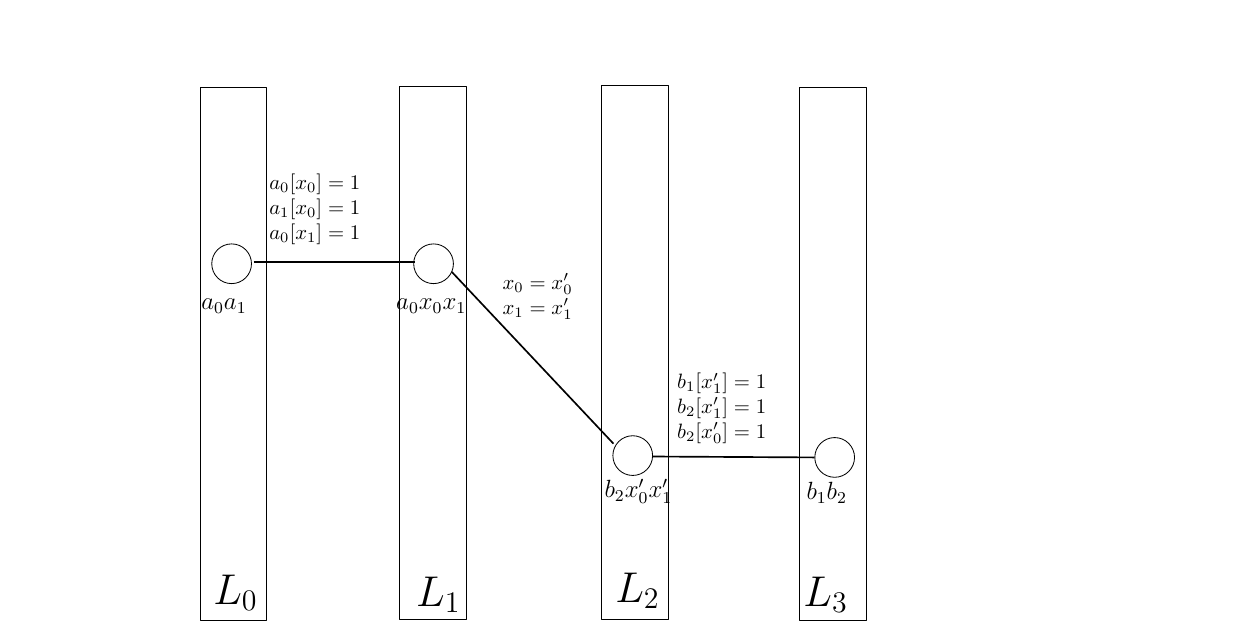}
    \caption{The construction when $k=3$.}\label{fig:k3}
\end{figure*}

\subsection{Construction for general $k$}
We will prove the following more detailed theorem, which will be useful for our Diameter lower bounds.

\begin{theorem} \label{kOV} Let $k\geq 2$.
	Given a $k$-OV instance consisting of sets $W_0,W_1,\dots,W_{k-1} \subseteq \{0,1\}^d$, each of size $N$, we can in $O(kN^{k-1}d^{k-1})$ time construct an unweighted, undirected graph with \newline $O(N^{k-1}+k N^{k-2} d^{k-1})$ vertices and $O(k N^{k-1} d^{k-1})$ edges that satisfies the following properties.
	\begin{enumerate}[itemsep=0mm]
		\item The graph consists of $k+1$ layers of vertices $S=L_0,L_1,L_2,\dots,\allowbreak L_k=T$. The number of vertices in the sets is $|S|=|T|=N^{k-1}$ and $|L_1|,|L_2|,\dots,|L_{k-1}|\leq N^{k-2}d^{k-1}$.
		\item $S$ consists of all tuples $(a_0,a_1,\ldots, a_{k-2})$ where for each $i$, $a_i\in W_i$. Similarly, $T$ consists of all tuples $(b_1,b_2,\ldots, b_{k-1})$ where for each $i$, $b_i\in W_i$.
		\item If the $k$-OV instance has no solution, then $d(u,v)=k$ for all $u \in S$ and $v \in T$.
		\item If the $k$-OV instance has a solution $a_0, a_1,\dots,a_{k-1}$ where for each $i$, $a_i\in W_i$ then if $\alpha=(a_0,\dots a_{k-2})\in S$ and $\beta= (a_1,\dots,a_{k-1}) \in T$, then $d(\alpha,\beta)\geq 3k-2$. 
		\item Suppose the $k$-OV instance has a solution $a_0, a_1,\dots,a_{k-1}$ where for each $i$, $a_i\in W_i$. Let $t=k-2$.
		Let $s$ be such that $0\leq s\leq t$.

Let $b_{t-s+j}\in W_{t-s+j}$ for all $j\in [1,\ldots,s]$ be some other vectors, potentially different from $a_{t-s+j}$.  
Consider $\alpha=(a_0,a_1,\ldots,a_{t-s},\allowbreak b_{t-s+1},\ldots,b_{t})\in L_0$ and $\beta=(a_{1},\ldots,a_{t+1})\in L_{t+2}$.
Then the distance between $\alpha$ and $\beta$ is at least $3t-2s+4$.

Symmetrically, let $c_{j}\in W_{j}$ for all $j\in [1,\ldots,s]$ be some other vectors, potentially different from $a_{j}$.  
Consider $\alpha=(a_0,a_1,\ldots,a_{t})\in L_0$ and $\beta=(c_{1},\ldots,c_s,a_{s+1},\dots, a_{t+1})\in L_{t+2}$.
Then the distance between $\alpha$ and $\beta$ is at least $3t-2s+4$.
		\item For all $i$ from 1 to $k-1$, for all $v \in L_i$ there exists a vertex in $L_{i-1}$ adjacent to $v$ and a vertex in $L_{i+1}$ adjacent to $v$. We can assume that this property holds because we can remove all vertices that do not satisfy this property from the graph and the resulting graph will still satisfy the other properties.
	\end{enumerate}
\end{theorem}

\paragraph{Proof of Theorem~\ref{kOV}}
We will prove the theorem for $k=t+2$ for any $t\geq 0$.

We will create a layered graph $G$ on $t+3$ layers, $L_0,\ldots,L_{t+2}$, where the edges go only between adjacent layers $L_i,L_{i+1}$. We will set $S=L_0$ and $T=L_{t+2}$ for the $S$-$T$ Diameter instance. In particular, $D_{S,T}\geq t+2$ because of the layering.

Let us describe the vertices of $G$. $L_0$ consists of $N^{t+1}$ vertices, each corresponding to a $t+1$-tuple $(a_0,a_1,\ldots, a_t)$ where for each $i$, $a_i\in W_i$. Similarly, $L_{t+2}$ consists of $N^{t+1}$ vertices, each corresponding to a $t+1$-tuple $(b_1,b_2,\ldots, b_{t+1})$ where for each $i$, $b_i\in W_i$. Layer $L_1$ consists of $N^t d^{t+1}$ vertices, each corresponding to a tuple $(a_0,\ldots,a_{t-1}, \bar{x})$ where for each $i$, $a_i\in W_i$ and $\bar{x}=(x_0,\ldots,x_t)$ is a $(t+1)$-tuple of coordinates in $[d]$. Similarly, $L_{t+1}$ consists of $N^t d^{t+1}$ vertices, each corresponding to a tuple $(b_2,\ldots,b_{t+1}, \bar{x})$ where for each $i$, $b_i\in W_i$ and $\bar{x}$ is a $(t+1)$-tuple of coordinates. For every $j\in \{2,\ldots,t\}$, $L_j$ consists of $N^t d^{t+1}$ vertices  
$(a_0, \ldots,  a_{t-j},\allowbreak b_{t+3-j},\ldots,  b_{t+1},\allowbreak  \bar{x})$, where for each $i$, $a_i\in W_i$, $b_i\in W_i$ and $\bar{x}=(x_0,\ldots,x_t)$ is a $(t+1)$-tuple of coordinates in $[d]$. In other words, there is a vector from $W_i$ for every $i\notin \{t-j+1,t-j+2\}$.

Now let us define the edges. Consider a vertex $(a_0,\ldots,a_t)\in L_0$. For every $\bar{x}=(x_0,\ldots,x_t)$, connect $(a_0,\ldots,a_t)$ to  $(a_0,\ldots,a_{t-1},\bar{x})\in L_1$ if and only if for every $j\in \{0,\ldots,t\}$, $a_j$ is $1$ in coordinates $x_0,\ldots,x_{t-j}$. 
For any $i\in \{1,\ldots,t\}$ let's define the edges between $L_i$ and $L_{i+1}$. For \newline $(a_0,\ldots,a_{t-i},b_{t+3-i},\ldots,b_{t+1}, \bar{x})\in L_i$~\footnote{Here if $i=1$, there are no $b$'s in the tuple.} and for any $c_{t+2-i}\in W_{t+2-i}$, add an edge to \newline $(a_0,\ldots,a_{t-i-1},c_{t+2-i},b_{t+3-i},\ldots,b_{t+1}, \bar{x})\in L_{i+1}$. Here we ``forget'' vector $a_{t-i}$ and replace it with $c_{t+2-i}$, leaving everything else the same. 

Finally, the edges between $L_{t+1}$ and $L_{t+2}$ are as follows. Consider some $(b_1,\ldots,b_{t+1})\in L_{t+2}$. For every $\bar{x}=(x_0,\ldots,x_t)$, connect $(b_1,\ldots,b_{t+1})$ to $(b_2,\ldots,b_{t+1},\bar{x})\in L_{t+1}$ if and only if for every $j\in \{1,\ldots,{t+1}\}$, $b_j$ is $1$ in coordinates $x_{t+1-j},\ldots,x_t$. 
Figure~\ref{fig:example2} shows the construction of the graph for $t=2$.

\begin{figure*}
  \centering
    \includegraphics[width=0.7\textwidth]{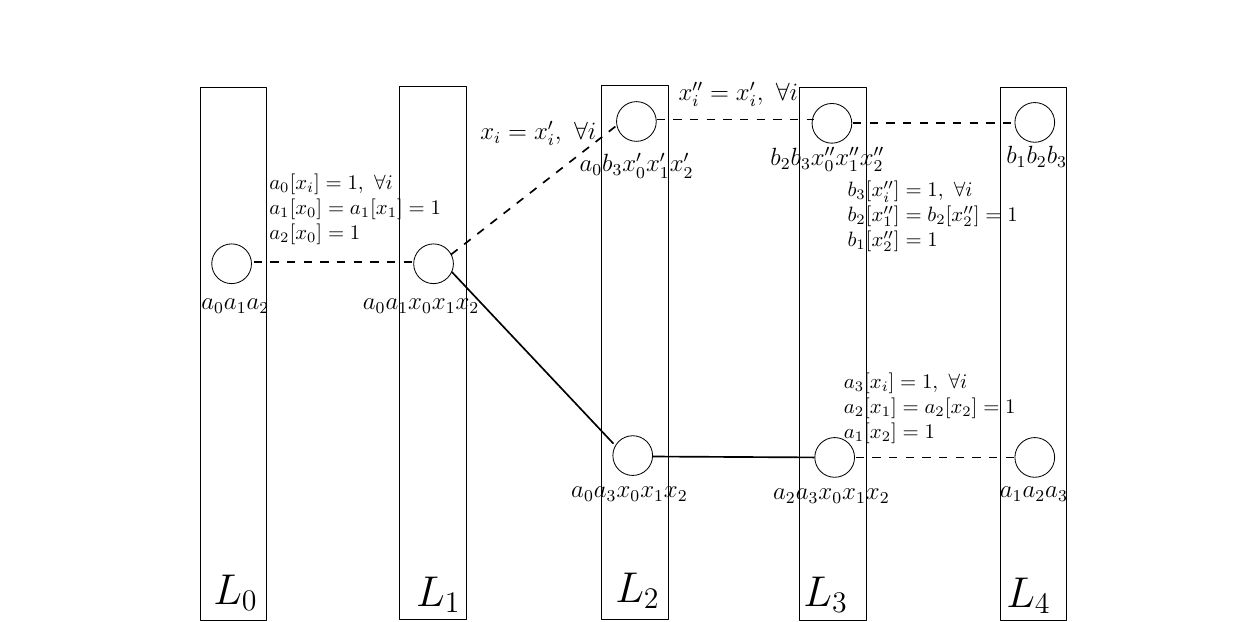}
    \caption{The reduction graph from $(t+2)$-OV for $t=2$. The figure depicts when a path of length $t+2$ exists between arbitrary $a_0a_1a_2\in L_0$ and $b_1b_2b_3\in L_{t+2}$. It also shows that when there is a path of length $t+2$ between $a_0a_1a_2\in L_0$ and $a_1a_2a_3\in L_{t+2}$, $a_0,a_1,a_2,a_3$ cannot be an orthogonal $4$-tuple.}\label{fig:example2}
\end{figure*}

An important claim is as follows:
\begin{claim}
For every $\bar{x}$, each $(a_0,\ldots,a_{t-1},\bar{x})\in L_1$ is at distance $t$ to every $(b_2,\ldots,b_t,\bar{x})\in L_{t+1}$.\label{claim:important}
\end{claim}

\begin{proof}
Consider the path starting from
$(a_0,\ldots,a_{t-1},\bar{x})$,
and then for each $i\geq 1$ following the edges $(a_0,\ldots,a_{t-i},b_{t+3-i},\ldots,\allowbreak b_{t+1},\allowbreak \bar{x})\in L_i$ to $(a_0,\ldots,a_{t-1-i},\allowbreak b_{t+2-i},\ldots,b_{t+1},\bar{x})\in L_{i+1}$, until we reach $(b_2,\ldots,b_{t+1},\bar{x})\in L_{t+1}$. This path exists by construction and has length $t$.
\end{proof}

Now we proceed to prove the bounds on the $S$-$T$ Diameter.

\begin{lemma}[Property 3 of Theorem~\ref{kOV}]If the $(t+2)$-OV instance has no solution, then $D_{S,T}=t+2$.
\end{lemma}

\begin{proof}
If the $(t+2)$-OV instance has no solution, then for every $c_0\in W_0,c_1\in W_1,\ldots,c_{t+1}\in W_{t+1}$, there is some coordinate $x$ such that $c_0[x]=c_1[x]=\ldots=c_{t+1}[x]=1$.

Now consider the graph and any $(a_0,\ldots,a_t)\in L_0$, $(b_1,\ldots,b_{t+1})\in L_{t+2}$.
For every $j\in \{0,\ldots,t\}$, let $x_j$ be a coordinate so that $a_0,\ldots,a_{t-j},\allowbreak b_{t-j+1},\ldots,b_{t+1}$ are all $1$ in $x_j$. Let $\bar{x}=(x_0,\ldots,x_{t})$.

By construction, $(a_0,\ldots,a_t)$ has an edge to $(a_0,\ldots,a_{t-1},\bar{x})$ and $(b_2,\ldots,b_{t+1},\bar{x})$ has an edge to $(b_1,\ldots,b_{t+1})$. Also, by Claim~\ref{claim:important}, $(a_0,\ldots,a_{t-1},\bar{x})$ has a path of length $t$ to $(b_2,\ldots,b_{t+1},\bar{x})$.
 
This shows that $D_{S,T}\leq t+2$; equality follows because the graph is layered.
\end{proof}

Now we prove the guarantee for the case when an orthogonal tuple exists.
\begin{lemma}[Property 4 of Theorem~\ref{kOV}]\label{lem:above}
If there exist $a_0\in W_0,\ldots,a_{t+1}\in W_{t+1}$ that are orthogonal, then $D_{S,T}\geq 3t+4$.
\end{lemma}

To prove the lemma, we will actually prove a more general claim: Property 5 of Theorem~\ref{kOV}.

\begin{claim}[Property 5 of Theorem~\ref{kOV}]
Suppose that $a_0\in W_0,\ldots,a_{t+1}\in W_{t+1}$ are orthogonal. Let $s$ be such that $0\leq s\leq t$.

Let $b_{t-s+j}\in W_{t-s+j}$ for all $j\in [1,\ldots,s]$ be some other vectors, potentially different from $a_{t-s+j}$.  
Consider $\alpha=(a_0,a_1,\ldots,a_{t-s},\allowbreak b_{t-s+1},\ldots,b_{t})\in L_0$ and $\beta=(a_{1},\ldots,a_{t+1})\in L_{t+2}$.
Then the distance between $\alpha$ and $\beta$ is at least $3t-2s+4$.

Symmetrically, let $c_{j}\in W_{j}$ for all $j\in [1,\ldots,s]$ be some other vectors, potentially different from $a_{j}$.  
Consider $\alpha=(a_0,a_1,\ldots,a_{t})\in L_0$ and $\beta=(c_{1},\ldots,c_s,a_{s+1},\dots, a_{t+1})\in L_{t+2}$.
Then the distance between $\alpha$ and $\beta$ is at least $3t-2s+4$.

%
\label{claim:veryimportant}
\end{claim}

If the claim is true, then using $s=0$ we get that the Diameter is at least $3t+4$ so Lemma~\ref{lem:above} is true. The claim for $s>0$ is useful for the rest of our constructions.

\begin{proof}
We will show that the distance between
$\alpha=(a_0,a_1,\ldots,\allowbreak a_{t-s},\allowbreak b_{t-s+1},\ldots,b_{t})\in L_0$ and $\beta=(a_{1},\ldots,a_{t+1})\in L_{t+2}$ is strictly more than $3t+2-2s$. Because the graph is layered and hence bipartite and $t+2\equiv 3t+2\mod 2$, the distance must be at least $3t-2s+4$.

Let's assume for contradiction that the shortest path $P$ between $\alpha$ and $\beta$ is of length $\leq 3t+2-2s$. First let's look at any subpath $P'$ of $P$ strictly within $M=L_1\cup\ldots\cup L_{t+1}$.
All vertices on $P'$ must share the same $\bar{x}$. Furthermore, if $P'$ starts with a vertex of $L_1$ and ends with a vertex of $L_{t+1}$, then by Claim~\ref{claim:important}, $P'$ must be of length exactly $t$. Next, notice that 
$P$ cannot go from $L_0$ to $L_{t+2}$ and then back to $L_0$. This is because it needs to end up in $L_{t+2}$ and any time it crosses over $M$, it would need to pay a distance of $t+2$, so $P$ would have to have length at least $3t+6>3t+2-2s$. Hence, $P$ must be of the following form: a path from $\alpha$ through $L_0\cup M$ back to $L_0$ (possibly containing only $\alpha$), followed by a path crossing $M$ to reach $L_{t+2}$, followed by a path through $L_{t+2}\cup M$ to $L_{t+2}$ (possibly empty).

We will show that if $P$ has length $\leq 3t+2-2s$ then $P$ must contain a length $t+2$ subpath $Q$ between a vertex $(a_0,\ldots,a_{q},w_{q+1},\ldots, w_t)\in L_0$, for some choices of the $w$'s and some $q\leq t-s$, and a vertex $(v_1,\ldots,v_q,a_{q+1},\ldots,a_{t+1})\in L_{t+2}$, for some choices of $v$'s. That is, this path traverses $M$ without weaving, by following $(a_0,\ldots,a_{q},w_{q+1},\ldots, w_{t-1},\bar{x})\in L_1$, $(a_0, \ldots, a_{q},\allowbreak w_{q+1},\ldots, w_{t-2},\allowbreak a_{t+1},\bar{x})\in L_2$, $\ldots$, $(v_2, \ldots, v_q, \allowbreak a_{q+1}, \ldots, a_{t-s}, \ldots, a_{t+1}, \bar{x})\in L_{t+1}$. Suppose we show that such a subpath exists. Then by the construction of our graph we have that for every $i\in \{0,\ldots,q\}$, $a_i[x_j]=1$ for all $j\in \{0,\ldots,t-i\}$, and that for all $i\in \{q+1,\ldots,t+1\}$, $a_i[x_j]=1$ for all $j\in\{t+1-i,\ldots,t\}$. That is, for all $i$, $a_i[x_{t-q}]=1$, and we get a contradiction since the $a_i$ were supposed to be orthogonal.


Now let $\alpha^*$ be the last vertex from $L_0$ on $P$ and let $\beta^*$ be the first vertex of $L_{t+2}$ of $P$. Let $a^*\in L_1$ be the vertex right after $\alpha^*$ and let $b^*\in L_{t+1}$ be the vertex right before $\beta^*$. See Figure~\ref{fig:case3}.

\begin{figure*}
  \centering
    \includegraphics[width=0.6\textwidth]{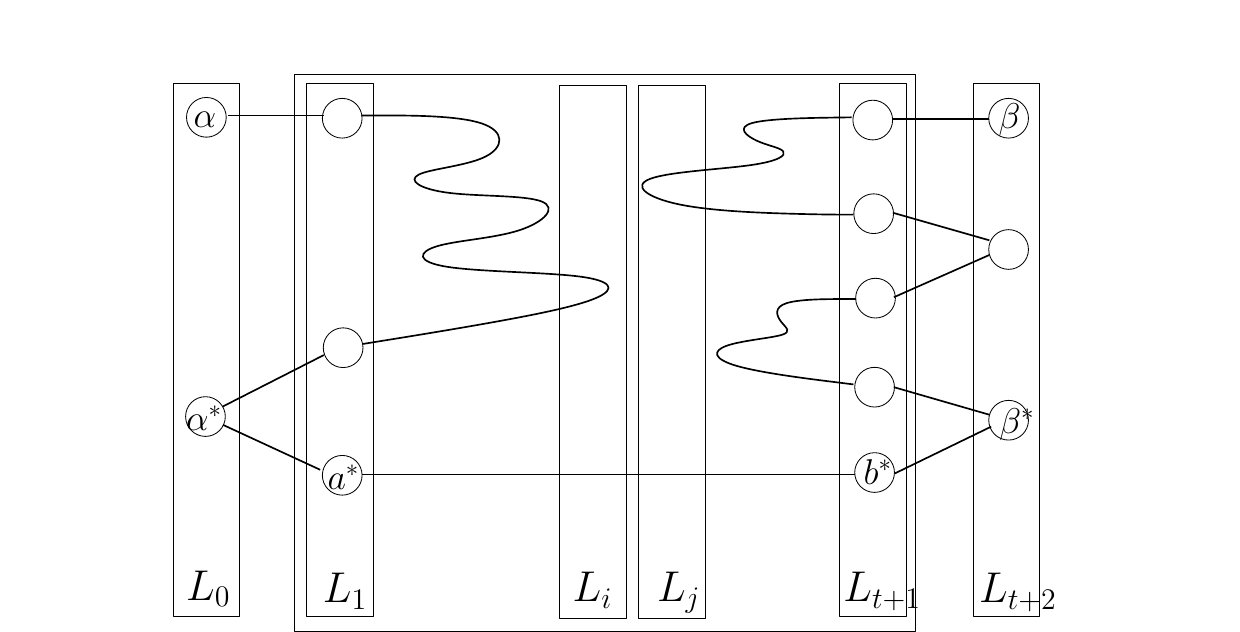}
    \caption{Here $P$ contains at least $2$ vertices in $L_0$ and at least $2$ in $L_{t+2}$, and $s=0$.}
	\label{fig:case3}
\end{figure*}

Since the subpath of $P$ between $a^*$ and $b^*$ is within $M$, it must share the same $\bar{x}$, and it must have length exactly $t$ by Claim~\ref{claim:important}. We will show that the subpath $Q$ that we are looking for is the subpath of $P$ between $\alpha^*$ and $\beta^*$. Its length is exactly what we want: $t+2$. It remains to show that for some $q\leq t-s$ and some choices of $w$'s and $v$'s, $\alpha^*=(a_0,\ldots,a_{q},w_{q+1},\ldots, w_t)$  and $\beta^*=(v_1,\ldots,v_q,a_{q+1},\ldots,a_{t+1})$. 

Consider the path $P_1$ between $\alpha=(a_0,a_1,\ldots,a_{t-s},b_{t-s+1},\ldots,b_{t})$ and $\alpha^*$ and the path $P_2$ between $\beta=(a_1,\ldots,a_{t+1})$ and $\beta^*$.
Let $L_i$ be the layer in $M$ with largest $i$ that $P_1$ hits and let $L_j$ be the layer in $M$ with smallest $j$ that $P_2$ hits. For convenience, let us define $j'=t+2-j$.
The length of $P_1$ is then at least $2i$ and the length of $P_2$ is at least $2j'$. 
The length $|P|$ of $P$ equals $t+2+|P_1|+|P_2|\geq t+2+2i+2j'=t+2+2(i+j')$. Since we have assumed that $|P|\leq 3t+2-2s$, we must have that $t+2+2(i+j')\leq 3t+2-2s$ and hence $i+j'\leq t-s$.
Now, since $P_1$ goes at most to $L_i$, then from getting from $\alpha$ to $\alpha^*$, at most the last $i$ elements of $(a_0,a_1,\ldots,a_{t-s},b_{t-s+1},\ldots,b_{t})$ can have been ``forgotten''. Hence, $\alpha^*=(a_0,\ldots,a_{t-\max\{s,i\}},\allowbreak b_{t-\max\{s,i\}+1},\ldots,b_{t-i},\allowbreak w_{t-i+1},\allowbreak \ldots,w_{t})$ for some $w$'s. (If $i\geq s$, the $b$'s do not appear.)

Similarly, between $\beta$ and $\beta^*$, at most the first $j'$ elements of $\beta$ can have been forgotten.
Thus, we have that $\beta^*=(v_1,\ldots,v_{j'},\allowbreak a_{j'+1},\ldots,a_{t+1})$ for some $v$'s. Now, since $i+j'\leq t-s$, we must have that $j'\leq t-s-i\leq t-\max\{s,i\}$, and hence the path between $\alpha^*$ and $\beta^*$ is the path $Q$ we are searching for. 
\end{proof}

%% file: ST-Diam-Equivalence.tex
\subsection{Equivalence between Diameter and $S$-$T$ Diameter}\label{subsec:diameq}
Here we will prove that when it comes to exact computation, $S$-$T$-Diameter and Diameter in weighted graphs are equivalent. The proof for directed graphs is much simpler, so we focus on the equivalence for undirected graphs. 
Also, it is clear that if one can solve $S$-$T$ Diameter, one can also solve Diameter in the same running time since one can simply set $S=T=V$. 
We prove:

\begin{theorem}\label{thm:stdiamequiv}
Suppose that there is a $T(n,m)$ time algorithm that can compute the Diameter of an $n$ vertex, $m$ edge graph with nonnegative integer edge weights. Then, the $S$-$T$ Diameter of any $n$ vertex $m$ edge graph with nonnegative integer edge weights can be computed in $T(O(n),O(m))$ time.
\end{theorem}

\begin{proof}
Let $G=(V,E),S,T$ be the $S$-$T$ Diameter instance; let $w:E\rightarrow \{0,\ldots,M\}$ be the edge weights.
First, we can always assume that $M$ is even: if it is not, multiply all edge weights by $2$; all distances (and hence also the $S$-$T$ Diameter) double. Let $S=\{s_1,\ldots,s_k\}$ and $T=\{t_1,\ldots,t_{\ell}\}$.

Now, let $W=Mn$. First add $|S|=k$ new vertices $S'=\{v_1,\ldots,v_k\}$. For each $i\in \{1,\ldots,k\}$ add a new edge $(v_i,s_i)$ of weight $W$. Let $G_S$ be this new graph. Let's consider the Diameter of $G_S$. For every pair of vertices $u,v\notin S'$, the distance is the same as in $G$. For $v_i\in S'$ and $x\notin S'$, the distance is $W+d_G(s_i,x)\leq W+M(n-1)<2W$.
For $v_i,v_j\in S'$, the distance is $2W+d(s_i,s_j)$. Hence the Diameter of $G_S$ is $2W+\max_{s_i,s_j\in S} d(s_i,s_j)$. 
 Hence by computing the Diameter of $G_S$, we can compute $D_S=\max_{s_i,s_j\in S} d(s_i,s_j)$.

We can create a similar graph $G_T$ whose Diameter will allow us to compute $D_T=\max_{t_i,t_j\in S} d(t_i,t_j)$.

After this, let's create a graph $G_{S,T}$ as follows (note that $G_{S,T}$ is not yet the final construction).
Add new vertices $S'=\{v_1,\ldots,v_k\}$ and $T'=\{u_1,\ldots,u_\ell\}$. For each $i\in \{1,\ldots,k\}$ add a new edge $(v_i,s_i)$ of weight $W$. For each $j\in \{1,\ldots,\ell\}$ add a new edge $(u_j,t_j)$ of weight $W$.
With a similar argument as above, the Diameter $D'$ of $G_{S,T}$ is $D'=2W+\max_{u,v\in S\cup T} d_G(u,v)$.

Let's assume without loss of generality that $D_S\geq D_T$.
If $D'> D_S$, then $D'=2W+\max_{u\in S,v\in T} d_G(u,v)$, and we can compute the $S$-$T$ Diameter of $G$ by subtracting $2W$.

Now suppose that we get $D'\leq D_S$; we must have then actually gotten $D'=D_S$. The $S$-$T$ Diameter of $G$ might be strictly smaller than $D_S$.
We add two new vertices $x$ and $y$ to $G_{S,T}$. We add an edge $(x,y)$ of weight $2W$, edges $(x,v_i)$ for every $v_i\in S'$ of weight $D_S/2$ and (symmetrically) edges $(y,u_j)$ for every $u_j\in T'$ of weight $D_S/2$.

Let $G'$ be the resulting graph. See Figure~\ref{fig:equiv}.

\begin{figure*}[h]
\centering
\includegraphics[width=0.7\textwidth]{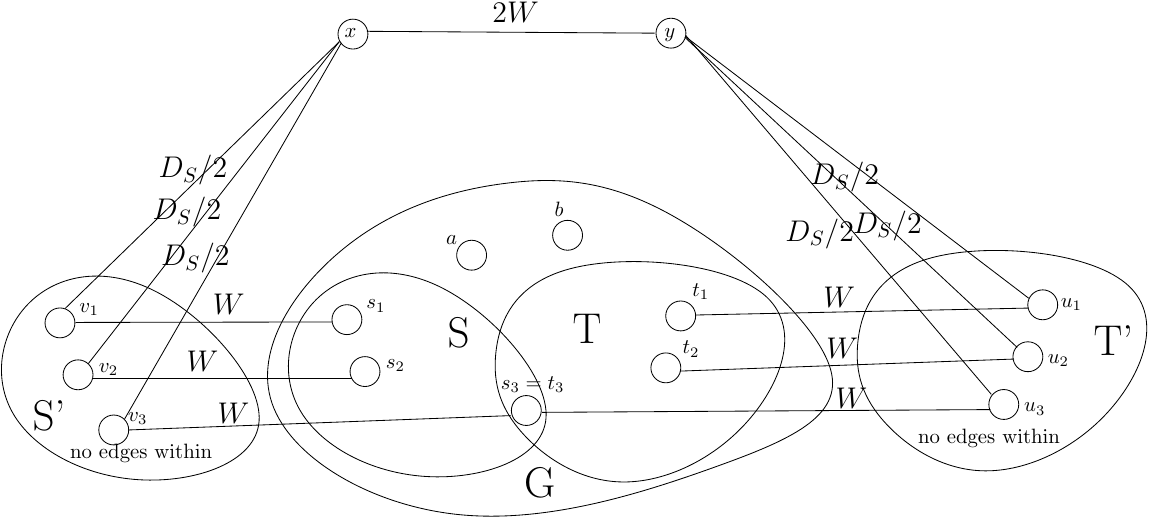}
\caption{A depiction of the construction of $G'$.}
\label{fig:equiv}
\end{figure*}

Let us consider the distances in this new graph $G'$.
\begin{enumerate}

\item For every $a,b\notin S'\cup T'\cup \{x,y\}$, $d(a,b)=d_G(a,b)<W$ as any path not in $G$ would have to use an edge of weight $W>d(a,b)$.

\item For every $b\notin S'\cup T'\cup \{x,y\}$, $d(x,b)=W+D_S/2 +\min_{a\in S} d_G(a,b)\leq 2W+D_S/2$. Similarly, $d(y,b)=W+D_S/2 +\min_{a\in T} d_G(a,b)\leq 2W+D_S/2$.

\item For every $v_i\in S'$, $d(x,v_i)=D_S/2$, and $d(y,v_i)=2W+D_S/2$. For every $u_i\in T'$, $d(y,u_i)=D_S/2$, and $d(x,u_i)=2W+D_S/2$.


\item For every $v_i,v_j\in S'$, $d(v_i,v_j)=D_S$. For every $u_i,u_j\in T'$, $d(u_i,u_j)=D_S$.


\item For every $v_i\in S'$ and $b\notin S'\cup T'\cup \{x,y\}$, $d(v_i,b)\leq W+d(s_i,b)\leq 2W.$ For every $u_i\in T'$ and $b\notin S'\cup T'\cup \{x,y\}$, $d(u_i,b)\leq W+d(t_i,b)\leq 2W.$

\item For every $v_i\in S'$ and $u_j\in T'$, $d(v_i,u_j)$ is the minimum of $D_S+2W, D_S+2W+\min_{s\in S} d_G(s,t_j), D_S+2W+\min_{t\in t} d_G(t,s_i)$ and $2W+d_G(s_i,t_j)$. The middle two terms are $\geq D_S+2W$, and hence $d(v_i,u_j)=2W+\min\{D_S,d_G(s_i,t_j)\}$.
\end{enumerate}

Consider $s_i\in S, t_j\in T$ that are the end points of the $S$-$T$ Diameter $D$ in $G$. Then $D=d_G(s_i,t_j)$.
Now, we have from before that $D\leq D_S$, as otherwise we have computed $D$ already.
Hence in $G'$, the distance $d(u_i,v_j)$ equals $2W+\min\{D_S,D\}=2W+D$ 

We note that for any $s_i,s_j\in S$, and any $t\in T$, $d_G(s_i,s_j)\leq d_G(s_i,t)+d_G(t,s_j)\leq 2\max_{s\in S,t\in T} d_G(s,t) = 2D$. Thus, $D\geq D_S/2$.
The distances in cases (1) to (5) are all $\leq 2W+D_S/2\leq 2W+D$. 
Hence the Diameter of $G'$ is actually exactly $2W+D$.
\end{proof}

%% file: Eccentricity-Hardness.tex
\section{Lower bounds for Eccentricities}\label{sec:ecchardness}

\subsection{Undirected graphs}

In this section we will prove the following theorem, which implies Theorem~\ref{thm:95eccintro} and the lower bound part of Theorem~\ref{thm:2ecc}.
\begin{theorem}\label{thm:undirecchard}
Let $k\geq 2$. Under the $k$-OV conjecture, for any $\delta>0$, any $(\frac{4k-3}{2k-1}-\delta)$-approximation algorithm for all Eccentricities in an unweighted undirected graph with $n$ vertices and $O(n)$ edges, requires at least $n^{1+1/(k-1)-o(1)}$ time on a $O(\log n)$-bit word-RAM.
\end{theorem}

\begin{proof}
Let's start with the $S$-$T$-diameter construction for $k$ obtained from a given $k$-OV instance.
We have a graph on $O(N^{k-1}d^{k-2})$ vertices and edges with the following properties: 

(1) Suppose that the $k$-OV instance has no $k$-OV solution. Then for every $s\in S,t\in T$, $d(s,t)=k$. Also, for every $s\in S$ and $u\notin S\cup T$, $d(s,u)\leq (k-1)+k=2k-1$ since we can take a $\leq (k-1)$ length path from $u$ to some vertex $t\in T$ and since $d(s,t)=k$.

(2) If there is a $k$-OV solution, there are two vertices $s\in S,t\in T$ with $d(s,t)\geq 3k-2$.

We modify the construction as follows. For every $s\in S$, we create an undirected path on $k-2$ new vertices $s_1\rightarrow s_2\rightarrow\ldots\rightarrow s_{k-2}$ and add an edge $(s,s_1)$; let's call $s$ by $s_0$. Now, the distance between $s_0$ and $s_i$ is $i$.
Add a new vertex $y$ and create edges $(s_{k-2},y)$ for every $s\in S$. Now, $d(y,s_0)=k-1$ for every $s\in S$, and also for every $s,s'\in S$ and all $i,j\in \{0,\ldots,k-2\}$, we have that $d(s_i,s'_j)\leq 2k-2$.

Now, we also attach paths to the vertices in $T$. In particular, for each $t$, add an undirected path $t\rightarrow t_1\rightarrow\ldots \rightarrow t_{k-1}$. See Figure~\ref{fig:undirecclb}.

\begin{figure*}[h]
  \centering
    \includegraphics[width=0.7\textwidth]{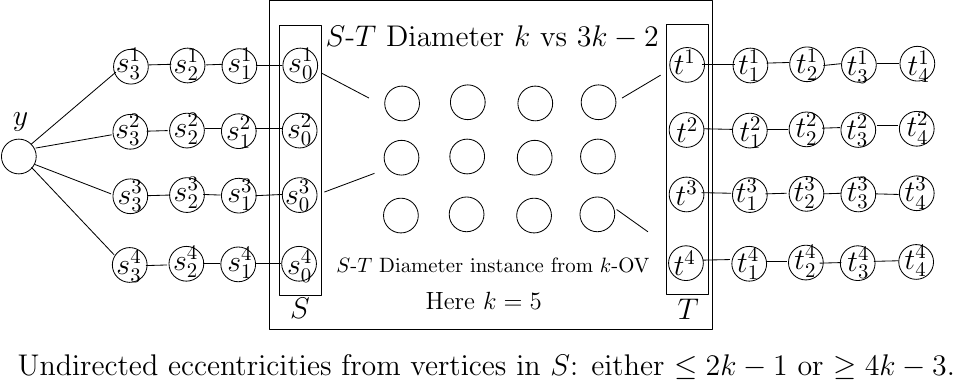}
    \caption{The undirected Eccentricities lower bound for $k=5$.}
	\label{fig:undirecclb}
\end{figure*}

The distance between any $s\in S$ and any $t_i$ is $i+d(s,t)$. Hence when there is no $k$-OV solution, the Eccentricities of all $s_0$ for $s\in S$ are $\leq k+(k-1)=2k-1$. 

For every $s\in S$, $t\in T$, there is now potentially a path between them through $y$ that was not present in the original $S$-$T$ Diameter construction. This path goes from $s$ to $y$ in $k-1$ steps, then to some other $s'$ in $k-1$ steps and then to $t$ using $\geq k$ steps. The length is $\geq 2(k-1)+k=3k-2$. Thus, when there is a $k$-OV solution, there is still a pair $s,t$ at distance at least $3k-2$. Then, due to the paths attached to $T$, we have $d(s_0,t_{k-1})\geq (3k-2)+(k-1)=4k-3$.\end{proof}

%% file: eccentricity-hardness-directed.tex
\subsection{Directed graphs}

In this section we prove the following theorem, which implies the lower bound part of Theorem~\ref{thm:otherecc}.

\begin{theorem}\label{thm:direcchard}
Under the $2$-OV conjecture, for any $\delta>0$, any $(2-\delta)$-approximation algorithm for all Eccentricities in an $n$ vertex, $O(n)$-edge directed unweighted graph, requires $n^{2-o(1)}$ time on a $O(\log n)$-bit word-RAM.
\end{theorem}

\begin{proof}
Suppose we are given an instance of $2$-OV: two sets of vectors $U,V$ over $\{0,1\}^d$, each of size $N$, and we want to know whether there are $u\in U,v\in V$ with $u\cdot v=0$. 

Let $L\geq 1$ be any integer.
Let us create a directed unweighted graph $G$; an illustration can be found in Figure~\ref{fig:direcclb}. $G$ will have a vertex $u$ for every $u\in U$ and a vertex $c$ for every $c\in [d]$. Every $v\in V$ will be represented by a directed path $v_0\rightarrow v_1\rightarrow \ldots\rightarrow v_L$.

\begin{figure*}[h]
  \centering
    \includegraphics[width=0.6\textwidth]{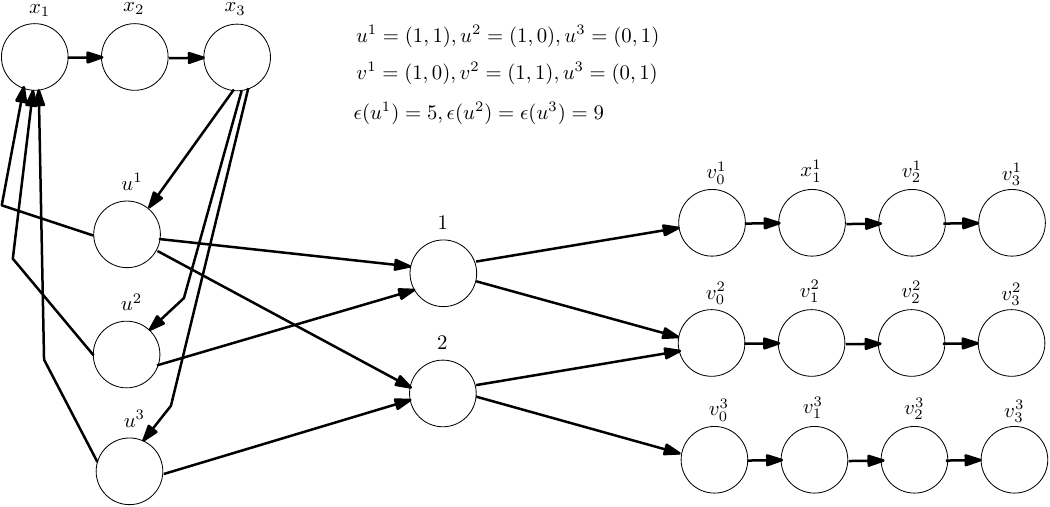}
    \caption{The directed Eccentricities lower bound.}\label{fig:direcclb}
\end{figure*}

In addition, there is a directed path $P_x$ on $L$ extra vertices, $x_1\rightarrow \ldots \rightarrow x_L$ so that every $u\in U$ has  directed edges $(u,x_1)$ and $(x_L,u)$. For every $u\in U$ and every $c$ for which $u[c]=1$, we add a directed edge $(u,c)$. For every $v$ and every $c$ for which $v[c]=1$, we add a directed edge $(c,v_0)$. Remove any $c$ that does not have at least one edge coming from $U$.

Let us consider the eccentricity of any vertex $u\in U$. First, for all $u'\in U$, $d(u,u')\leq L+1$ since one can go through the path $P_x$. For every $c\in [d]$, there is at least one edge coming from some $u'\in U$, and so one can reach $c$ from $u$ by first taking $P_x$ to $u'$ and then using the edge $(u',c)$. Hence, $d(u,c)\leq L+2$ for all $c\in [d]$. 

For any $v\in V$ and $i\in \{1,\ldots, L\}$, the distance $d(u,v_i)=i+d(u,v_0)$, and so we consider $d(u,v_0)$.
If there is a $c$ for which $u[c]=v[c]=1$, then $d(u,v_0)=2$, and hence for all $i$, $d(u,v_i)\leq L+2$.
If no such $c$ exists and so if $u$ and $v$ are orthogonal, the only way to reach $v_0$ is potentially via $P_x$ to some other $u'\in U$ which is at distance $2$ to $v_0$. Hence if $u$ and $v$ are orthogonal, $d(u,v_0)=L+3$, and hence $d(u,v_L)=2L+3$. Thus, we have that the eccentricity of $u$ is $L+2$ if it is not orthogonal to any vectors in $V$ and it is $\geq 2L+3$ if there is some $v$ that is orthogonal to $u$.

The number of vertices in the graph is $O(NL+d)$ and the number of edges is $O(NL+Nd)$. Suppose that there is a $(2-\eps)$-approximation algorithm for all Eccentricities in graphs with $O(m)$ vertices and edges running in $O(m^{2-\delta})$ time for some $\eps,\delta>0$. Then, we construct the above instance for $L=\lceil 1/\eps\rceil$ and run the algorithm on it. The approximation returned is at least as good as a $(2-1/L)$-approximation. Hence if the diameter is at least $2L+3$, the algorithm will return an estimate that is at least $L(2L+3)/(2L-1) > L+2$. Thus the algorithm can solve $2$-OV in time $O((NL+Nd)^{2-\delta}) = O(N^{2-\delta}d^{2-\delta})$, contradicting the $2$-OV conjecture.
\end{proof}

%% file: diameter_lb_5_8.tex
\section{Diameter lower bounds}\label{sec:diam_lb}
For all of our constructions we begin with the $S$-$T$ diameter lower bound construction from Theorem~\ref{kOV}. Here, if the $k$-OV instance has no solution, $D_{S,T}\leq k$ and if the instance has a solution $D_{S,T}\geq 3k-2$. To adapt this construction to Diameter, we need to ensure that if the OV instance has no solution then \emph{all} pairs of vertices have small enough distance. We begin by augmenting the $S$-$T$ Diameter construction by adding a matching between $S$ and a new set $S'$ as well as a matching between $T$ and a new set $T'$. Without any further modifications, pairs of vertices $u,v\in S\cup S'$ (or $u,v\in T\cup T'$) could be far from one another. 
The challenge is to add extra gadgetry to make these pairs close for ``no" instances while maintaining that in ``yes" instances the distance between the diameter endpoints $s'\in S',t'\in T'$ is large. That is, for ``yes" instances, we want a shortest path between the diameter endpoints $s'$ and $t'$ to contain the vertex $s\in S$ matched to $s'$ and the vertex $t\in T$ matched to $t'$ so that we can use use the fact that $d(s,t)\geq 3k-2$. In other words, we do not want there to be a shortcut from $s'$ to some vertex in $S$ that allows us to use a path of length $k$ from $S$ to $T$. For example, we cannot simply create a vertex $x$ and connect it to all vertices in $S\cup S'$ because this would introduce shortcuts from $S'$ to $S$.

We will describe some intuition for the augmentations to the graph regarding 3-OV for simplicity. Recall that $s'\in S', t'\in T'$ are the endpoints of the diameter and let $t$ be the vertex matched to $t'$. To solve the problem outlined in the above paragraph, we observe that in the ``yes" case there are three types of vertices $s\in S$. (1) close: $d(s,t)=3$, (2) far: $d(s,t)\geq 7$ (property 4 of Theorem~\ref{kOV}), and (3) intermediate: $d(s,t)\geq 5$ (property 5 of Theorem~\ref{kOV}). For close $s$, we need $d(s',s)$ to be large so that there is no shortcut from $s'$ to $t'$ through $s$. For far $s$, it is acceptable if $d(s',s)$ is small because $d(s,t)$ is large enough to ensure that paths from $s'$ to $t'$ through $s$ are still long enough. For intermediate $s$, $d(s',s)$ cannot be small, but it also need not be large. To fulfill these specifications, we add a small clique (the graph is still sparse) and connect each of its vertices to only {\em some} of the vertices in $S$ and/or $S'$ according to the implications of property 5 of Theorem~\ref{kOV}. When $s$ is close, we ensure that $d(s',s)$ is large by requiring that a shortest path from $s'$ to $s$ goes from $s'$ to the clique, uses an edge inside of the clique, and then goes from the clique to $s$. When $s$ is intermediate, we ensure that $d(s',s)$ is not too small by requiring that a shortest path from $s'$ to $s$ goes from $s'$ to the clique and then from the clique to $s$ (without using an edge inside of the clique). These intermediate $s$ are important as they allow every vertex in the clique to have an edge to some vertex in $S$ and thus be close enough to the $T$ side of the graph in the ``no" case.

\subsection{$5$ vs $8$ unweighted undirected construction}

In this section we show that under the $3$-OV Hypothesis, any algorithm that can distinguish between Diameter $5$ and $8$ in sparse undirected unweighted graphs, requires $\Omega(n^{3/2-o(1)})$ time.

Theorem~\ref{kOV} gives us the following theorem.
\begin{theorem} \label{3OV_construction}
	Given a $3$-OV instance consisting of three sets $A,B,C \subseteq \{0,1\}^d$, $|A|=|B|=|C|=N$, we can in $O(N^2d^2)$ time construct an unweighted, undirected graph with $O(N^2+Nd^2)$ vertices and $O(N^2d^2)$ edges that satisfies the following properties.
	\begin{enumerate}
		\item The graph consists of $4$ layers of vertices $S,L_1,L_2,T$. The number of vertices in the sets is $|S|=|T|=N^2$ and $|L_1|,|L_2|\leq Nd^2$.
		\item $S$ consists of all tuples $(a,b)$ of vertices $a \in A$ and $b \in B$. Similarly, $T$ consists of all tuples $(b,c)$ of vertices $b \in B$ and $c \in C$.
		\item If the $3$-OV instance has no solution, then $d(u,v)=3$ for all $u \in S$ and $v \in T$.
		\item If the $3$-OV instance has a solution $a \in A, b \in B, c \in C$ with $a,b,c$ orthogonal, then $d((a,b)\in S,(b,c) \in T)\geq 7$. 
		\item If the $3$-OV instance has a solution $a \in A, b \in B, c \in C$ with $a,b,c$ orthogonal, then by setting $k=3,s=1$ in Property 5 of Theorem~\ref{kOV} we have: for any $b' \in B$ we have $d((a,b) \in S,(b',c) \in T)\geq 5$ and $d((a,b') \in S,(b,c) \in T)\geq 5$.
		\item For any vertex $u \in L_1$ there exists a vertex $s \in S$ that is adjacent to $u$. Similarly, for any vertex $v \in L_2$ there exists a vertex $t \in T$ that is adjacent to $v$. We can assume that this property holds because we can remove all vertices that do not satisfy this property from the graph and the resulting graph will still satisfy the other properties.
	\end{enumerate}
\end{theorem}

In the rest of the section we use Theorem~\ref{3OV_construction} to prove the following result, which implies Theorem~\ref{thm:32tight}.
\begin{theorem} \label{diameter_5_8}
	Given a $3$-OV instance, we can in $O(N^2d^2)$ time construct an unweighted, undirected graph with $O(N^2+Nd^2)$ vertices and $O(N^2d^2)$ edges that satisfies the following two properties.
	\begin{enumerate}
		\item If the $3$-OV instance has no solution, then for \emph{all} pairs of vertices $u$ and $v$ we have $d(u,v)\leq 5$.
		\item If the $3$-OV instance has a solution, then there exists a pair of vertices $u$ and $v$ such that $d(u,v)\geq 8$.
	\end{enumerate}
\end{theorem}

\paragraph{Construction of the graph}
We construct a graph with the required properties by starting with the graph from Thereom~\ref{3OV_construction} and adding more vertices and edges. 
Figure~\ref{fig:figure_5_8} illustrates the construction of the graph.
\begin{figure*}
  \centering
    \includegraphics[width=0.7\textwidth]{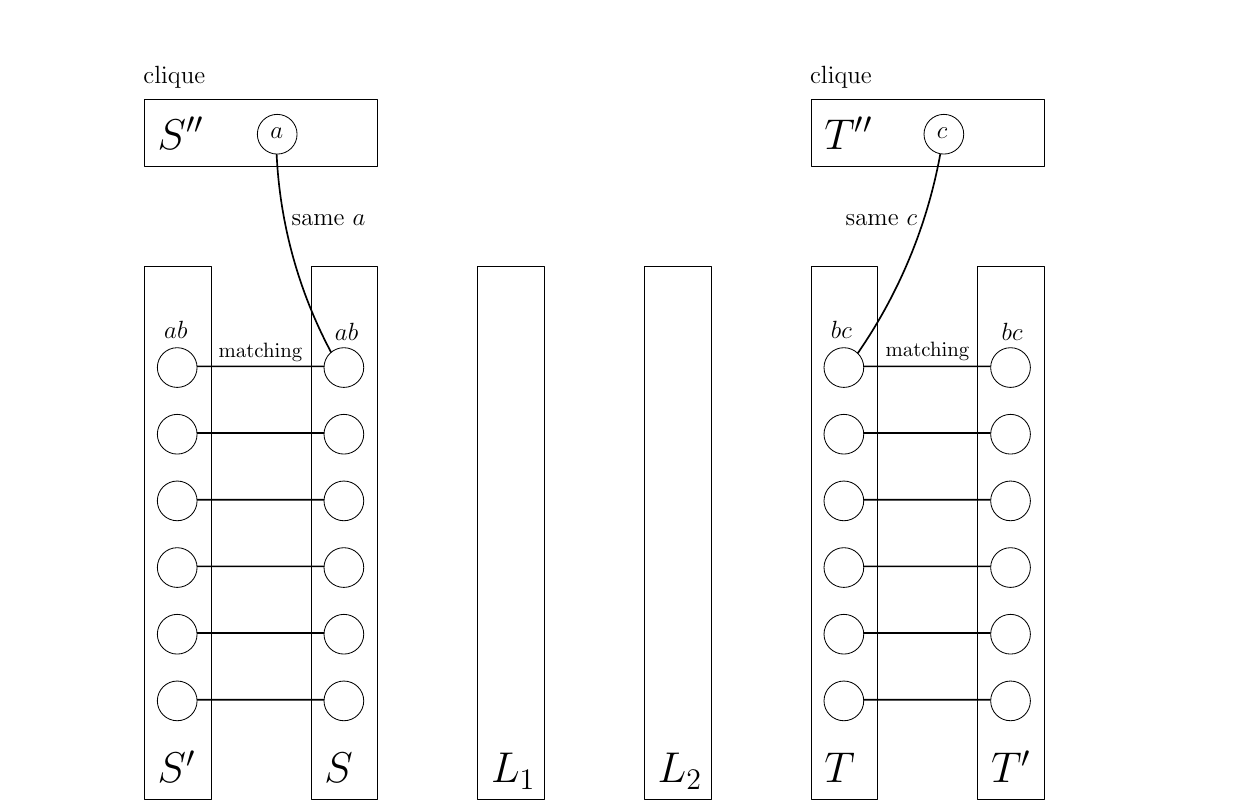}
    \caption{The illustration for the $5$ vs $8$ construction. The edges between sets $S,L_1,L_2$ and $T$ are not depicted. The edges between vertices in $S'$ and $S$ ($T$ and $T'$) form a matching. Vertices in $S''$ ($T''$) form a clique.}
	\label{fig:figure_5_8}
\end{figure*}
We start by adding a set $S'$ of $N^2$ vertices. $S'$ consists of all tuples $(a,b)$ of vertices $a \in A$ and $b \in B$. We connect every $(a,b) \in S'$ to its counterpart $(a,b) \in S$. Thus, there is a matching between the sets of vertices $S$ and $S'$. We also add another set $S''$ of $N$ vertices. $S''$ contains one vertex $a$ for every $a \in A$. For every pair of vertices from $S''$ we add an edge between the vertices. Thus, the $N$ vertices form a clique.
Furthermore, for every vertex $a \in S''$ we add an edge to $(a,b) \in S$ for all $b \in B$.
In total we added $N^2+N=O(N^2)$ vertices and $\binom{N}{2}+2N^2=O(N^2)$ edges. We do a similar construction for the set $T$ of vertices. We add a set $T'$ of $N^2$ vertices - one vertex for every tuple $(b,c)$ of vertices $b \in B$ and $c \in C$. We connect every $(b,c)\in T'$ to $(b,c) \in T$. Finally, we add a set $T''$ of $N$ vertices. $T''$ contains one vertex for every vector $c \in C$. For every pair of vertices from $T''$ we add an edge between the vertices. 
We connect every $c \in T''$ to $(b,c) \in T$ for all $b \in B$.
This finishes the construction of the graph. 
In the rest of the section we show that the construction satisfies the promised two properties.

\paragraph{Correctness of the construction} We need to consider two cases.
\paragraph{Case 1: the $3$-OV instance has no solution} In this case we want to show that for all pairs of vertices $u$ and $v$ we have $d(u,v)\leq 5$. We consider three subcases.
\paragraph{Case 1.1: $u \in S \cup S' \cup S'' \cup L_1$ and $v \in T \cup T' \cup T'' \cup L_2$} We observe that there exists $s \in S$ with $d(u,s)\leq 1$. Indeed, if $u \in S$, then $s=u$ works. If $u \in S' \cup S''$, then we are done by the construction. On the other hand, if $u \in L_1$, then there exists such an $s \in S$ by property 6 from Theorem~\ref{3OV_construction}. Similarly we can show that there exists $t \in T$ such that $d(v,t)\leq 1$. Finally, by property $3$ we have that $d(s,t)=3$. Thus, we can upper bound the distance between $u$ and $v$ by $d(u,v)\leq d(u,s)+d(s,t)+d(t,v)\leq 1+3+1=5$ as required.
\paragraph{Case 1.2: $u,v \in S \cup S' \cup S'' \cup L_1$} From the previous case we know that there are two vertices $s_1,s_2 \in S$ such that $d(u,s_1)\leq 1$ and $d(s_2,v)\leq 1$. To show that $d(u,v)\leq 5$ it is sufficient to show that $d(s_1,s_2)\leq 3$. This is indeed true since both vertices $s_1$ and $s_2$ are connected to some two vertices in $S''$ and every two vertices in $S''$ are at distance at most $1$ from each other.
\paragraph{Case 1.3: $u,v \in T \cup T' \cup T'' \cup L_2$} The case is analogous to the previous case.
\paragraph{Case 2: the $3$-OV instance has a solution} In this case we want to show that there is a pair of vertices $u,v$ with $d(u,v)\geq 8$. Let $a \in A, b \in B, c \in C$ be a solution to the $3$-OV instance. We claim that $d((a,b) \in S',(b,c) \in T')\geq 8$. Let $P$ be an optimal path between $u=((a,b) \in S')$ and $v=((b,c) \in T')$ that achieves the smallest distance. We want to show that $P$ uses at least $8$ edges. Let $t \in T$ be the first vertex from the set $T$ that is on path $P$. Let $s \in S$ be the last vertex on path $P$ that belongs to $S$ and precedes $t$ in $P$. We can easily check that, if $s \neq ((a,b) \in S)$, then $d(u,s)\geq 3$ and, similarly, if $t \neq ((b,c) \in T)$, then $d(t,v)\geq 3$. We consider three subcases.
\paragraph{Case 2.1: $s \neq ((a,b) \in S)$ and $t \neq ((b,c) \in T)$} Since $s$ and $t$ are separated by two layers of vertices, we must have $d(s,t)\geq 3$. Thus we get lower bound $d(u,v)\geq d(u,s)+d(s,t)+d(t,v)\geq 3+3+3=9>8$ as required.
\paragraph{Case 2.2: $s = ((a,b) \in S)$ and $t = ((b,c) \in T)$} In this case we use property 4 and conclude $d(u,v)\geq d(u,s)+d(s,t)+d(t,v)=1+d((a,b) \in S,(b,c) \in T)+1\geq 1+7+1=9>8$ as required.
\paragraph{Case 2.3: either $s = ((a,b) \in S)$ or $t = ((b,c) \in T)$ holds but not both}
W.l.o.g.\ $s \neq ((a,b) \in S)$ and $t = ((b,c) \in T)$. If the path uses an edge in the clique on $S''$ before arriving at $s$, then $d(u,s)\geq 4$ and we get that $d(u,v)\geq d(u,s)+d(s,t)+d(t,v)\geq 4+3+1=8$. On the other hand, if the path does not use any edge of the clique, then $s=((a,b') \in S)$ for some $b' \in B$. By property 5 we have $d(s,t)=d((a,b')\in S,(b,c)\in T)\geq 5$. We conclude that $d(u,v)\geq d(u,s)+d(s,t)+d(t,v)\geq 3+5+1=9>8$ as required.

%% file: diameter_lb_6_10.tex
\subsection{$6$ vs $10$ weighted undirected construction}

In this section we change the construction from Theorem~\ref{diameter_5_8} to show that under the $3$-OV Hypothesis, any algorithm that can distinguish between diameter $6$ and $10$ in sparse undirected weighted graphs requires $\Omega(n^{3/2-o(1)})$ time.

We get the following theorem.
\begin{theorem} \label{diameter_6_10}
	Given a $3$-OV instance, we can in $O(N^2d^2)$ time construct a weighted, undirected graph with $O(nN^2+nNd^2)$ vertices and $O(nN^2d^2)$ edges that satisfies the following two properties.
	\begin{enumerate}
		\item If the $3$-OV instance has no solution, then for \emph{all} pairs of vertices $u$ and $v$ we have $d(u,v)\leq 6$.
		\item If the $3$-OV instance has a solution, then there exists a pair of vertices $u$ and $v$ such that $d(u,v)\geq 10$.
	\end{enumerate}
	Each edge of the graph has weight either $1$ or $2$.
\end{theorem}

\paragraph{Construction of the graph}
The construction of the graph is the same as in Theorem~\ref{diameter_6_10} except all edges connecting vertices between sets $L_1$ and $L_2$ have weight $2$ and all edges inside the cliques on vertices $S''$ and $T''$ have weight $2$. All the remaining edges have weight $1$.

\paragraph{Correctness of the construction} The correctness proof is essentially the same as for Theorem~\ref{diameter_5_8}. As before we consider two cases.

\paragraph{Case 1: the $3$-OV instance has no solution} In this case we want to show that for all pairs of vertices $u$ and $v$ we have $d(u,v)\leq 6$. In the analysis of Case 1 in Theorem~\ref{diameter_5_8} we show a path between $u$ and $v$ such that the path involves at most one edge from the cliques or between sets $L_1$ and $L_2$. Since we added weight $2$ to the latter edges, the length of the path increased by at most $1$ as a result. So we have upper bound $d(u,v)\leq 6$ for all pairs $u$ and $v$ of vertices.

\paragraph{Case 2: the $3$-OV instance has a solution} In this case we want to show that there is a pair of vertices $u,v$ with $d(u,v)\geq 10$. Similarly to Theorem~\ref{diameter_5_8} we will show that $d((a,b) \in S',(b,c) \in T')\geq 10$, where $a \in A, b \in B, c \in C$ is a solution to the $3$-OV instance. The analysis of the subcases is essentially the same as in Theorem~\ref{diameter_5_8}. For cases 2.1 and 2.2 in the proof of Theorem~\ref{diameter_5_8} we had $d((a,b) \in S',(b,c) \in T')\geq 9$. Since we increased edge weights between $L_1$ and $L_2$ to $2$ and every path from $(a,b) \in S'$ to $(b,c) \in T'$ must cross the layer between $L_1$ and $L_2$, we also increased the lower bound of the length of the path from $9$ to $10$ for cases 2.1 and 2.2. It remains to consider Case 2.3.
As in the proof of Theorem~\ref{diameter_5_8}, w.l.o.g.\ $s \neq ((a,b) \in S)$ and $t = ((b,c) \in T)$. If the path uses an edge in the clique on $S''$ before arriving at $s$, then $d(u,s)\geq 5$ and we get lower bound $d(u,v)\geq d(u,s)+d(s,t)+d(t,v)\geq 5+4+1=10$. On the other hand, if the path does not use any edge of the clique, then $s=((a,b') \in S)$ for some $b' \in B$. By property 5 and because we increased edge weights between $L_1$ and $L_2$ to $2$, we have $d(s,t)=d((a,b')\in S,(b,c)\in T)\geq 6$. We conclude that $d(u,v)\geq d(u,s)+d(s,t)+d(t,v)\geq 3+6+1=10$ as required.

%% file: diameter_lb_k_simple.tex
\subsection{$3k-4$ vs $5k-7$ unweighted directed construction}
In this section, we show that under SETH, for every $k\geq 3$, every algorithm that can distinguish between Diameter $3k-4$ and $5k-7$ in directed unweighted graphs requires $\Omega(n^{1+1/(k-1)-o(1)})$ time.

Theorem~\ref{kOV} gives us the following theorem.
\begin{theorem} \label{kOV_construction_simple}

	Given a $k$-OV instance consisting of $k\geq 2$ sets $W_0,W_1,\dots,W_{k-1} \subseteq \{0,1\}^d$, each of size $N$, we can in $O(kN^{k-1}d^{k-1})$ time construct an unweighted, undirected graph with $O(N^{k-1}+k N^{k-2} d^{k-1})$ vertices and $O(k N^{k-1} d^{k-1})$ edges that satisfies the following properties.
	\begin{enumerate}
		\item The graph consists of $k+1$ layers of vertices $S=L_0,L_1,L_2,\dots,L_k=T$. The number of vertices in the sets is $|S|=|T|=N^{k-1}$ and $|L_1|,|L_2|,\dots,|L_{k-1}|\leq N^{k-2}d^{k-1}$.
		\item $S$ consists of all tuples $(a_0,a_1,\ldots, a_{k-2})$ where for each $i$, $a_i\in W_i$. Similarly, $T$ consists of all tuples $(b_1,b_2,\ldots, b_{k-1})$ where for each $i$, $b_i\in W_i$.
		\item If the $k$-OV instance has no solution, then $d(u,v)=k$ for all $u \in S$ and $v \in T$.
		\item If the $k$-OV instance has a solution $a_0, a_1,\dots,a_{k-1}$ where for each $i$, $a_i\in W_i$ then if $\alpha=(a_0,\dots a_{k-2})\in S$ and $\beta= (a_1,\dots,a_{k-1}) \in T$, then $d(\alpha,\beta)\geq 3k-2$. 
		\item Setting $s=k-2$ in Property 5 of Theorem~\ref{kOV}: If the $k$-OV instance has a solution $a_0, a_1,\dots,a_{k-1}$ where for each $i$, $a_i\in W_i$ then for any tuple $(b_1,\dots,b_{k-2})$, if $\alpha=(a_0,b_1,\dots,b_{k-2})\in S$ and $\beta=(a_1,\dots,a_{k-1})\in T$, then $d(\alpha,\beta)\geq k+2$. Symmetrically, if $\alpha=(a_0,a_1,\dots,a_{k-2})\in S$ and $\beta=(b_1,\dots,b_{k-2},a_{k-1})\in T$, then $d(\alpha,\beta)\geq k+2$.

		\item For all $i$ from 1 to $k-1$, for all $v \in L_i$ there exists a vertex in $L_{i-1}$ adjacent to $v$ and a vertex in $L_{i+1}$ adjacent to $v$. 
		We can assume that this property holds because we can remove all vertices that do not satisfy this property from the graph and the resulting graph will still satisfy the previous three properties.
	\end{enumerate}
\end{theorem}

In the rest of the section we use Theorem~\ref{kOV_construction_simple} to prove the following result.

\begin{theorem} \label{diameter_k_simple}
	Given a $k$-OV instance, we can in $O(kN^{k-1}d^{k-1})$ time construct an unweighted, directed graph with $O(k N^{k-1}+k N^{k-2} d^{k-1})$ vertices and $O(kN^{k-1}d^{k-1})$ edges that satisfies the following two properties.
	\begin{enumerate}
		\item If the $k$-OV instance has no solution, then for \emph{all} pairs of vertices $u$ and $v$ we have $d(u,v)\leq 3k-4$.
		\item If the $k$-OV instance has a solution, then there exists a pair of vertices $u$ and $v$ such that $d(u,v)\geq 5k-7$.
	\end{enumerate}
\end{theorem}

\paragraph{Construction of the graph}
We construct a graph with the required properties by starting with the graph from Thereom~\ref{kOV_construction_simple} and adding more vertices and edges. First we will construct a weighted graph and then we will make it unweighted.
Figure~\ref{fig:figure_8_13} illustrates the construction of the graph for the special case $k=4$.

\begin{figure*}[h]
  \centering
    \includegraphics[width=0.7\textwidth]{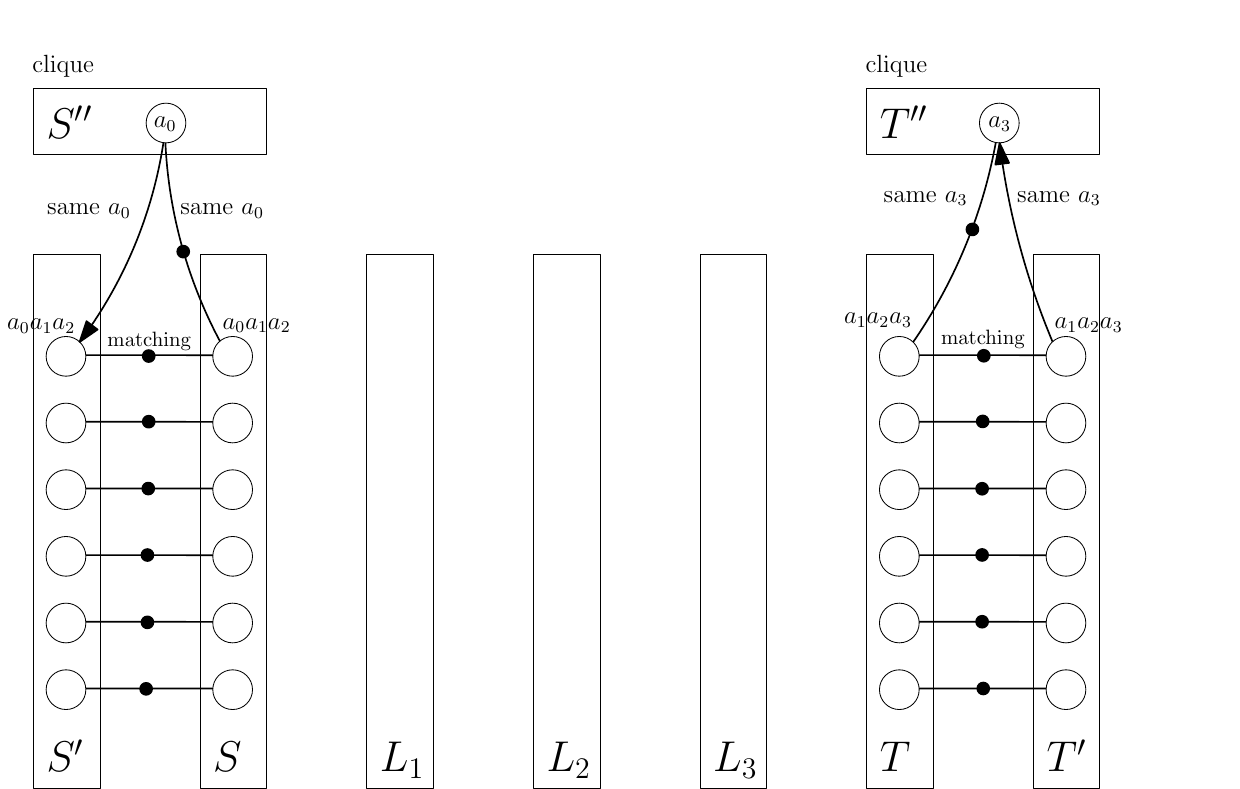}
    \caption{The $3k-4$ vs $5k-7$ construction for the special case $k=4$. The edges between sets $S,L_1,L_2,L_3$ and $T$ are not depicted. The matching between sets $S$ and $S'$ consists of unweighted paths of length $k-2=2$. The edges between sets $S$ and $S''$ consists of unweighted paths of length $k-2=2$. Similarly for the right side.}
	\label{fig:figure_8_13}
\end{figure*}
We start by adding a set $S'$ of $N^{k-1}$ vertices. $S'$ consists of all tuples $(a_0,a_1,\ldots, a_{k-2})$ where for each $i$, $a_i\in W_i$. We connect every $(a_0,a_1,\ldots, a_{k-2}) \in S'$ to its counterpart $(a_0,a_1,\ldots, a_{k-2}) \in S$ with an undirected edge of weight $k-2$ to form a matching. We also add another set $S''$ of $N$ vertices. $S''$ contains one vertex $a_0$ for every $a_0 \in W_0$. For every pair of vertices in $S''$ we add an undirected edge of weight $1$ between the vertices. Thus, the $N$ vertices form a clique. 
Furthermore, for every vertex $a_0 \in S''$ we add an undirected edge of weight $k-2$ to $(a_0,b_1,\ldots, b_{k-2}) \in S$ for all $b_1,\dots, b_{k-2}$.
Finally for every vertex $a_0 \in S''$ we add a \emph{directed} edge of weight $1$ towards $(a_0,b_1,\ldots, b_{k-2}) \in S$ for all $b_1,\dots, b_{k-2}$.
Some of the edges that we added have weight $k-2$. We make those unweighted by subdividing them into edges of weight $1$. Let $S'''$ be the set of newly added vertices.
In total we added $O(kN^{k-1})$ vertices and $O(kN^{k-1})$ edges.

We do a similar construction for the set $T$ of vertices. We add a set $T'$ of $N^{k-1}$ vertices --- one vertex for every tuple $(a_1,\dots,a_{k-1})$ where for each $i$, $a_i\in W_i$. We connect every $(a_1,\dots,a_{k-1})\in T'$ to $(a_1,\dots,a_{k-1}) \in T$ by an undirected edge of weight $k-2$. Finally, we add a set $T''$ of $n$ vertices. $T''$ contains one vertex for every vector $a_{k-1} \in W_{k-1}$.
We connect every pair of vertices in $T''$ by an undirected edge of weight $1$.
We connect every vertex $a_{k-1} \in T''$ to $(b_1,\dots,b_{k-2},a_{k-1}) \in T$ by an undirected edge of weight $k-2$ for all $b_1,\dots,b_{k-2}$.
Also, for every vertex $a_{k-1} \in T''$ we add a \emph{directed} edge of weight $1$ from $(b_1,\dots,b_{k-2},a_{k-1}) \in T'$  to $a_{k-1}$ for all  $b_1,\dots,b_{k-2}$.
Some of the edges that we just added have weight $k-2$. We make those unweighted by subdividing them into edges of weight $1$. Let $T'''$ be the set of newly added vertices.
This finishes the construction of the graph. In the rest of the section we show that the construction satisfies the promised two properties stated in Theorem~\ref{diameter_k_simple}.

\paragraph{Correctness of the construction} We need to consider two cases.
\paragraph{Case 1: the $k$-OV instance has no solution} 
In this case we want to show that for all pairs of vertices $u$ and $v$ we have $d(u,v)\leq 3k-4$. We consider subcases.
\paragraph{Case 1.1: $u \in S \cup S' \cup S'' \cup S''' \cup L_i$ for $1\leq i\leq k-2$ and $v \in T \cup T' \cup T'' \cup T''' \cup L_j$ for $2\leq j\leq k-1$}   We observe that there exists $s \in S$ that has $d(u,s)\leq k-2$. Similarly, there exists $t \in T$ with $d(t,v)\leq k-2$. By property 3 from Theorem~\ref{kOV_construction_simple} we have that $d(s,t)\leq k$. This gives us upper bound $d(u,v)\leq d(u,s)+d(s,t)+d(t,v)\leq (k-2)+k+(k-2)=3k-4$ as required. The proof when the sets for $u$ and $v$ are swapped is identical since we only use paths on unweighted edges.
\paragraph{Case 1.2: $u,v \in S \cup S' \cup S'' \cup S''' \cup L_1$} We note that there is some vertex $s\in S''$ with $d(u,s)\leq 2(k-2)$ (via undirected edges). Also, there is some vertex $s'\in S''$ with $d(s',v)\leq k-1$ (possibly using directed edges). $S''$ is a clique so $d(s,s')\leq 1$. Thus, $d(u,v)\leq d(u,s)+d(s,s')+d(s',v)\leq 2(k-2)+1+(k-1)=3k-4$. 
\paragraph{Case 1.3: $u,v \in T \cup T' \cup T'' \cup T''' \cup L_{k-1}$} This case is similar to the previous case. We note that there is some vertex $t\in T''$ with $d(t,v)\leq 2(k-2)$ (via undirected edges). Also, there is some vertex $t'\in T''$ with $d(u,t')\leq k-1$ (possibly using directed edges). $S''$ is a clique so $d(t',t)\leq 1$. Thus, $d(u,v)\leq d(u,t')+d(t',t)+d(t,v)\leq (k-1)+1+2(k-2)=3k-4$. 
\paragraph{Case 2: the $k$-OV instance has a solution} In this case we want to show that there is a pair of vertices $u,v$ with $d(u,v)\geq 5k-7$. Let $(a_0,a_1,\ldots, a_{k-1})$  be a solution to the $k$-OV instance where for each $i$, $a_i\in W_i$. We claim that $d((a_0,\dots,a_{k-2}) \in S',(a_1,\dots,a_{k-1}) \in T')\geq 5k-7$. Let $P$ be an shortest path between $u=((a_0,\dots,a_{k-2}) \in S')$ and $v=((a_1,\dots,a_{k-1}) \in T')$. We want to show that $P$ uses at least $5k-7$ edges.
Let $s \in S$ be the first vertex on path $P$ that belongs to $S$ and let $t \in T$ be the last vertex from the set $T$ that is on path $P$. We observe that due to the directionality of the edges, $s$ and $t$ must be the counterparts of $u$ and $v$ respectively; that is, $s=((a_0,\dots,a_{k-2}) \in S)$ and $t=((a_1,\dots,a_{k-1}) \in T)$. Note that these definitions of $s$ and $t$ differ from the definitions of $s$ and $t$ in previous proofs. We consider three subcases.
\paragraph{Case 2.1: A vertex in $S'\cup S''\cup S'''$ appears after $s$ on the path $P$} We observe that if $s_1,s_2\in S$ is a pair of vertices on the path $P$ such that no vertex in $S$ appears between them on $P$, then the portion of $P$ between $s_1$ and $s_2$ either contains only vertices in $S'\cup S''\cup S'''$ or contains no vertices in $S'\cup S''\cup S'''$. Let $s_1,s_2\in S$ be such that the portion of $P$ between them contains only vertices in $S'\cup S''\cup S'''$. Such $s_1,s_2$ exist by the specification of this case. If $s_1=s_2$ then $P$ is not a shortest path. Otherwise, the portion of $P$ between $s_1$ and $s_2$ must include a vertex in $S''$. Thus, $d(s_1,s_2)\geq 2(k-2)$. We consider three subcases.
\begin{itemize}
\item $s_1\not=s$. The distance between any pair of vertices in $S$ is at least 2 so $d(s,s_1)\geq 2$. Then, $d(u,v)\geq d(u,s)+d(s,s_1)+d(s_1,s_2)+d(s_2,t)+d(t,v)\geq (k-2)+2+2(k-2)+k+(k-2)=5k-6$.
\item $s_1=s$ and $s_2=((a_0,b_1,\dots,b_{k-2}) \in S)$ for some $b_1,\dots,b_{k-2}$. In this case, by property 5 we have $d(s_2,t)\geq k+2$. Thus, $d(u,v)\geq d(u,s_1)+d(s_1,s_2)+d(s_2,t)+d(t,v)\geq (k-2)+2(k-2)+(k+2)+(k-2)=5k-6$.
\item $s_1=s$ and $s_2=((b_0,\dots,b_{k-2} \in S)$ for some with $b_0\not=a_0$. In this case, the path from $s_1$ to $s_2$ must include an edge in the clique $S''$ since these are the only edges among vertices in $S'\cup S''\cup S'''$ for which adjacent tuples can differ with respect to their first element. Thus, $d(s_1,s_2)\geq 2(k-2)+1\geq 2k-3$. Therefore, $d(u,v)\geq d(u,s_1)+d(s_1,s_2)+d(s_2,t)+d(t,v)\geq (k-2)+(2k-3)+k+(k-2)=5k-7$.
\end{itemize}
\paragraph{Case 2.2: A vertex in $T'\cup T''\cup T'''$ appears before $t$ on the path $P$} This case is analogous to the previous case.
\paragraph{Case 2.3: The portion of the path $P$ between $s$ and $t$ contains no vertices in $S'\cup S''\cup S'''\cup T'\cup T''\cup T'''$} By property 4, $d(s,t)\geq 3k-2$. Thus, $d(u,v)\geq d(u,s)+d(s,t)+d(t,v)\geq (k-2)+(3k-2)+(k-2)=5k-6$.\\

We note that a slight modification of this construction gives a lower bound for higher values of Diameter. For any L, we can get an $L(3k-4)$ vs $L(5k-8)+1$ construction by subdividing all of the edges in the construction (into paths of length $L$) except for the directed edges and the edges within the cliques. 

%% file: algorithm.tex
\section{Algorithms for sparse graphs}\label{sec:eccalgs}

\subsection{2-Approximation for Eccentricities in $\tO(m\sqrt n)$ time}

In this section we prove the following theorem, which implies the upper bound part of Theorem~\ref{thm:otherecc}.

\begin{theorem} \label{alg1}
	Given a weighted, directed $m$ edge $n$ vertex graph, there is an $\tO(m\sqrt n)$ time randomized algorithm that outputs for each $v\in V$ a quantity $\epsilon'(v)$  such that for all $v \in V$ we have $\epsilon(v)/2\leq \epsilon'(v)\leq \epsilon(v)$.
\end{theorem}
\begin{proof}
	The algorithm is inspired by the 2-approximation algorithm for directed radius of Abboud, Vassilevska W., and Wang~\cite{AbboudWW16}. We claim that the following algorithm achieves the above guarantees.
	\begin{enumerate}
		\item Sample a random subset $S \subset V$ of size $|S|=\Theta(\sqrt n \log n)$. With high probability for every $u \in V$ we have $\Nin_{\sqrt n}(u)\cap S \neq \emptyset$.
		\item Let $w$ be a vertex that maximizes $d(S,w)$, which we find using Dijkstra's algorithm. Let $S':=\Nin_{\sqrt n}(w)$.
		\item For every vertex $v \in S'$ we output $\epsilon'(v)=\epsilon(v)$ by running Dijkstra's algorithm and following the outgoing edges.
		\item For every vertex $v \not\in S'$ we output the estimate $\epsilon'(v)=\max_{s \in S\cup \{w\}}d(v,s)$. We can determine all these quantities by running Dijkstra's algorithm out of all vertices in $S \cup \{w\}$ and following the incoming edges.
	\end{enumerate}
	
	\paragraph{Correctness} Consider an arbitrary vertex $v\not\in S'$ (if $v \in S'$, then we are done by the third step). If there exists $s \in S$ such that $d(v,s)\geq \epsilon(v)/2$, then we are done since $\epsilon'(v)\geq d(v,s)\geq \epsilon(v)/2$. Otherwise, we have $d(v,s)<\epsilon(v)/2$ for all $s \in S$. Let $v'$ be a vertex that achieves $d(v,v')=\epsilon(v)$. By the triangle inequality we have $d(s,v')>\epsilon(v)/2$ for all $s \in S$. Equivalently, $d(S,v')>\epsilon(v)/2$. This implies that $d(S,w)>\epsilon(v)/2$ by our choice of $w$. Since $d(S,w)>\epsilon(v)/2$ and $S'=\Nin_{\sqrt n}(w)$ intersects $S$, we must have that $S'$ contains all vertices $u$ with $d(u,w)\leq \epsilon(v)/2$. Since $v \not \in S'$, we must have $d(v,w)>\epsilon(v)/2$ and we are done since $\epsilon'(v)\geq d(v,w)>\epsilon(v)/2$.
\end{proof}

%% file: algorithm_linear.tex
\subsection{Almost 2-Approximation for Eccentricities in almost linear time}

In contrast to our $\tilde{O}(m\sqrt n)$ time algorithm from the previous section, our near-linear time $(2+\delta)$-approximation algorithm is very different from all previously known algorithms. 
Our algorithm proceeds in iterations and maintains a set $S$ of vertices for which we still do not have a good eccentricity estimate. In each iteration either we get a good estimate for many new vertices and hence remove them from $S$, or we remove all vertices from $S$ that have large eccentricities, and for the remaining vertices in $S$ we have a better upper bound on their eccentricities. After a small number of iterations we have a good estimate for all vertices of the graph. 

In this section we prove the following theorem, which implies the upper bound part of Theorem~\ref{thm:2ecc}.

\begin{theorem} 
Suppose that we are given a weighted, directed $m$ edge $n$ vertex graph. The weights of all edges are non-negative and in the range $[1/n^c,n^c]$ for some constant $c$. For any $1>\tau>0$, there is a randomized $\tO(m/\tau)$ time algorithm that with high probability outputs for each $v\in V$ a quantity $\eps'(v)$ such that for all $v \in V$ we have ${1-\tau\over 2}\eps(v)\leq \eps'(v)\leq \eps(v)$. \label{alg2}
\end{theorem}
\begin{proof}
	We begin by computing the strongly connected components of the graph, which allows us to determine which vertices have infinite eccentricity. We maintain a subset $S\subseteq V$ of vertices $v$ for which we still do not have an estimate $\eps'(v)$. Initially $S=V$ and we will end with $|S|= O(\log n)$. When $|S|= O(\log n)$ we can evaluate $\eps(v)$ for all $v \in S$ in the total time of $\tilde{O}(m)$. Also we maintain a value $D$ that upper bounds the largest eccentricity of a vertex in $S$. That is, $\eps(v)\leq D$ for all $v \in S$. We begin by computing the strongly connected components of the graph, which allows us to determine which vertices have infinite eccentricity. Then we set $D=n^C$ for some large enough constant $C>0$. The algorithm proceeds in phases. Each phase takes $\tO(m)$ time and either $|S|$ decreases by a factor of at least $2$ or $D$ decreases by a factor of at least $1/(1-\tau)$. After $O(\log(n)/\tau)$ phases either $|S|= O(\log n)$ or $D<1/n^c$. 

	For a subset $S \subseteq V$ of vertices and a vertex $x \in V$ we define a set $S_x\subseteq S$ to contain those $|S_x|=|S|/2$ vertices from $S$ that are closest to $x$ (according to distance $d(\cdot,x)$). The ties are broken by taking the vertex with the smaller id.
	Given a subset $S \subseteq V$ of vertices and a threshold $D$, a phase proceeds as follows.
	\begin{itemize}[itemsep=0mm]
		\item We sample a set $A \subseteq S$ of $O(\log n)$ random vertices from the set $S$. With high probability for all $x \in V$ we have $A \cap S_x \neq \emptyset$.
		\item Let $w \in V$ be a vertex that maximizes $d(A,w)$. We can find it using Dijkstra's algorithm.
		\item We consider two cases.
\paragraph{Case $d(S\setminus S_w,w)\geq {1-\tau \over 2}D$.} For all $x \in S\setminus S_w$ we have ${1-\tau \over 2} D \leq \eps(x)\leq D$ and we assign the estimate $\eps'(x)={1-\tau \over 2}D$. This gives us that ${1-\tau \over 2} \eps(x)\leq {1-\tau \over 2}D=\eps'(x)\leq \eps(x)$ for all $x\in S\setminus S_w$. We update $S$ to be $S_w$. This decreases the size of $S$ by a factor of $2$ as required.
\paragraph{Case $d(S\setminus S_w,w)<{1-\tau \over 2}D$.} Set $S'=S$. For every vertex $v \in S$ evaluate $r_v:=\max_{x \in A}d(v,x)$. We can evaluate these quantities by running Dijkstra's algorithm from every vertex in $A$ and following the incoming edges. If $r_v\geq {1-\tau \over 2}D$, then assign the estimate $\eps'(v)={1-\tau \over 2}D$ and remove $v$ from $S'$. Similarly as in the previous case we have ${1-\tau \over 2} \eps(v)\leq\eps'(v)\leq \eps(v)$ for all $v \in S\setminus S'$. Below we will show that for every $v \in S'$ we have $\eps(v)\leq (1-\tau)D$. Thus we can update $S=S'$ and decrease the threshold $D$ to $(1-\tau)D$ as required.
	\end{itemize}
	
	\paragraph{Correctness} We have to show that, if there exists $v \in S'$ such that $\eps(v)>(1-\tau)D$, then we will end up in the first case (this is the contrapositive of the claim in the second case). Since $v \in S'$ we must have that $d(v,x)\leq {1-\tau \over 2}D$ for all $x \in A$. Since $\eps(v)>(1-\tau)D$, we must have that there exists $v'$ such that $d(v,v')>(1-\tau)D$. By the triangle inequality we get that $d(x,v')>{1 - \tau \over 2}D$ for every $x \in A$. By choice of $w$ we have $d(A,w)>{1-\tau \over 2}D$. Since $A \cap S_w \neq \emptyset$, we have $d(S\setminus S_w,w)\geq {1-\tau \over 2}D$ and we will end up in the first case.
	
	The guarantee on the approximation factor follows from the description.
\end{proof}



%% file: corollary.tex
As a corollary, we get an algorithm for Source Radius with the same runtime and approximation ratio as Theorem~\ref{alg2}. First, run the Eccentricities algorithm and let $v$ be a vertex with minimum estimated eccentricity $\epsilon'(v)$. Then run Dijkstra's algorithm from $v$ and report $\epsilon(v)$ as the Radius estimate $R'$. Let $R$ be the true radius of the graph and let $x$ be a vertex with minimum Eccentricity i.e. $\epsilon(x)=R$. If $\alpha$ is the approximation ratio for the Eccentricities algorithm then $\epsilon(v)\leq \alpha\epsilon'(v)\leq \alpha\epsilon(v)$ and $\epsilon(x)\leq \alpha\epsilon'(x)\leq \alpha\epsilon(x)$. By choice of $v$, $\epsilon'(v)\leq \epsilon'(x)$. Thus, $\alpha R=\alpha\epsilon(x)\geq \alpha\epsilon'(x)\geq \alpha\epsilon'(v)\geq \epsilon(v)=R'$. Clearly $R'\geq R$, so $R\leq R'\leq \alpha R$.

%% file: ST-Diameter-Algs.tex
\subsection{$S$-$T$ Diameter algorithms}\label{sec:stdiamalgs}
Recall that the $S$-$T$ diameter problem is as follows: Given an undirected graph $G=(V,E)$ and two sets $S\subseteq V, T\subseteq V$, determine $D_{S,T}=\max_{s\in S, t\in T} d(s,t)$. Here we will outline two algorithms for the problem.

Let us first consider a fast $3$-approximation algorithm.
\begin{claim}\label{claim:3appxstdiam}
There is an $O(m+n)$ time deterministic algorithm that for any $n$ vertex $m$ edge unweighted graph $G=(V,E)$ and $S\subseteq V, T\subseteq V$, computes an estimate $D'$ such that $D_{S,T}/3\leq D'\leq D_{S,T}$ and two vertices $s\in S$, $t\in T$ such that $d(s,t)=D'$.
In graphs with nonnegative weights, the same estimate can be achieved in $O(m+n\log n)$ time.
\end{claim}

\begin{proof}
The algorithm is extremely simple: pick arbitrary vertices $s\in S$ and $t\in T$, compute BFS($s$) and BFS($t$) and return $\max\{\max_{t'\in T} d(s,t'),\max_{s'\in S} d(s',t)\}$ (also returning the two vertices achieving the maximum). For weighted graphs, run Dijkstra's algorithm instead of BFS.

Let's see why this algorithm provides the promised guarantee. Suppose that for every $t'\in T$, $d(s,t')<D_{S,T}/3$ (otherwise we are done). Then for every $t',t''\in T$, $d(t',t'')\leq d(t',s)+d(s,t'')<2D_{S,T}/3$. In particular, for all $t'\in T$, $d(t,t')<2D_{S,T}/3$. If we also had that for every $s'\in S$, $d(t,s')<D_{S,T}/3$, then we'd get that for all $s'\in S, t'\in T$, $d(s',t')\leq d(s',t)+d(t,t')<D_{S,T}$, contradicting the definition of $D_{S,T}$. Thus, $\max\{\max_{t'\in T} d(s,t'),\max_{s'\in S} d(s',t)\}\geq D_{S,T}/3$. 
\end{proof}

We will now show an analogue to the $\tilde{O}(m\sqrt n)$ time almost-$3/2$-approximation diameter algorithm of Roditty and Vassilevska W.~\cite{RV13} for $S$-$T$ Diameter giving a $2$-approximation. Using a trick from Chechik et al.~\cite{ChechikLRSTW14} we also obtain a true $2$ approximation algorithm running in $\tilde{O}(m^{3/2})$.


\begin{algorithm}
\caption{$2$-Approximation for $S$-$T$ Diameter}\label{euclid}
\begin{algorithmic}[1]
\Procedure{$2$-Approx}{}
\State $X$ - random sample of vertices, $|X|=\Theta(\sqrt n \log n)$
\State $D_1:=0$
\For{\texttt{every $x\in X$}}
      \State Run BFS($x$)
			\State Let $t_x$ be the closest vertex to $x$ in $T$
			\State Run BFS($t_x$)
			\State $D_1 = \max\{D_1,\max_{s\in S} d(s,t_x)\}$
\EndFor
\State \multiline{Let $\bar{t}$ be the furthest vertex of $T$ from $X$ (computed above)}
\State Run BFS($\bar{t}$)
\State $D_2=\max_{s\in S} d(s,\bar{t})$.
\State Let $Y$ be the closest $\sqrt n$ vertices to $\bar{t}$.
\For{\texttt{every $y\in Y$}}
      \State Run BFS($y$)
			\State Let $s_y$ be the closest vertex to $y$ in $S$
			\State Run BFS($s_y$)
			\State $D_2 = \max\{D_2,\max_{t\in T} d(s_y,t)\}$
\EndFor
\Return $\max\{D_1,D_2\}$
\EndProcedure
\end{algorithmic}
\end{algorithm}

We use Algorithm 1 to prove:

\begin{theorem}\label{thm:stdiam32}
There is an $\tilde{O}(m\sqrt n)$ time randomized algorithm that with high probability outputs an estimate $D'$ for the $S$-$T$ diameter $D$ of an $m$ edge $n$ vertex unweighted undirected graph such that $2\lfloor D/4\rfloor \leq D'\leq D$. 

In $\tilde{O}(m^{3/2})$ time one can obtain an estimate $D''$ such that $D/2\leq D''\leq D$.
\end{theorem}

\begin{proof}
First we analyze Algorithm 1.
Let $s^*\in S$ and $t^*\in T$ be a pair of vertices with $d(s^*,t^*)=D$.
Let $d=\lfloor D/4\rfloor$.

Suppose first that for some $x\in X$, $d(x,t^*)\leq d$. Then, $d(x,t_x)\leq d(x,t^*)\leq d$ and hence $d(t_x,t^*)\leq d(t_x,x)+d(x,t^*)\leq 2d$. However, then $d(t_x,s^*)\geq d(t^*,s^*)-d(t^*,t_x)\geq D-2d\geq D/2$. In this case, $D_1\geq D/2$ and we are done.

Thus, if $D_1<D/2$, it must be that for every $x\in X$, $d(x,t^*)\geq d+1$.
Hence, for every $x \in X$, $d(x,\bar{t})\geq d(x,t^*)\geq d+1$ by the definition of $\bar{t}$. 
If $d(\bar{t},s^*)\geq D/2$, then $D_2\geq D/2$ and we are done, so let us assume that $d(\bar{t},s^*)\leq D/2$.

Now, as $X$ is random of size $c\sqrt n\log n$ for large enough $c$, with high probability, $X$ hits the $\sqrt n$-neighborhoods of all vertices. In particular, $X\cap Y\neq \emptyset$. 
However, since $d(x,\bar{t})\geq d+1$ for every $x\in X$, it must be that $Y$ contains all vertices at distance $d$ from $\bar{t}$ as it contains all vertices closer to $\bar{t}$ than $x\in Y\cap X$.

If $s^*\in Y$, then we would have run BFS from $s^*$ and returned $D$. Hence $d(\bar{t},s^*)>d$. Let $a$ be the vertex on the shortest path between $\bar{t}$ and $s^*$ with $d(\bar{t},a)=d$. We thus have that $a\in Y$. 
Also, since $d(\bar{t},s^*)\leq D/2$, $d(a,s^*)\leq D/2-d$ and hence $d(a,s_a)\leq D/2-d$, so that $d(s_a,t^*)\geq D-2(D/2-d)\geq 2d$. This finishes the argument that $2$-\textsc{Approx} returns an estimate $D'$ with $2\lfloor D/4\rfloor\leq D'\leq D$.

It is not too hard to see that the only time that we might get an estimate that is less than $D/2$ is in the last part of the argument and only if the diameter is of the form $4d+3$. (We will prove the algorithm guarantees formally soon.) The analysis fails to work in that case because $Y$ is guaranteed to contain only the vertices at distance $d$ from $\bar{t}$.


In particular, if $Y$ contains all vertices at distance $d+1$ from $\bar{t}$ instead of just those at distance at most $d$, we could consider $a$ to be the vertex on the shortest path between $\bar{t}$ and $s^*$ with $d(\bar{t},a)=d+1$, and $a\in Y$. 
Now since $d(\bar{t},s^*)\leq 2d+1$ (as otherwise we'd be done), $d(a,s_a)\leq d(a,s^*)\leq 2d+1-d-1 =d$, so that $d(s_a,t^*)\geq 2d+3$. Hence everything would work out.

We handle this issue with a trick from Chechik et al.~\cite{ChechikLRSTW14}. First, we make graph have constant degree by blowing up the number of vertices and adding 0 weight edges as follows. Let $v$ be an original vertex and suppose it has degree $d(v)$. Replace $v$ with a $d(v)$-cycle of $0$ weight edges so that each of the cycle vertices is connected to a one of the neighbors of $v$, where each neighbor has a cycle vertex corresponding to it. This makes every vertex have degree $3$ and increases the number of vertices to $O(m)$. 

Now, we run algorithm $2$-Approx with two changes. The first is that instead of BFS we use Dijkstra's algorithm\footnote{We can also use Thorup's algorithm~\cite{thorup1999undirected}, which runs in linear time and is stated for positive weight edges but can also handle  zero weight edges.} because the edges now have weights. The second change is that we redefine $Y$ as follows. Let $Z$ be the closest $\sqrt m$ vertices to $\bar{t}$. Define $Y$ to be $Z$, together with all vertices that have a non-zero weight edge to some vertex of $Z$.


Since every vertex has degree $3$, the number of vertices in $Y$ is $\leq 4|Z|\leq O(\sqrt m)$ and hence we can afford to run Dijkstra from each of them and complete the algorithm in $\tilde{O}(m^{3/2})$ time.

Let us now formally analyze the guarantees of the algorithm. 
If some vertex $x\in X$ has $d(x,t^*)\leq D/4$, we get that $d(t_x,s^*)\geq D-2(D/4)=D/2$. If we are not done, all vertices of $X$ have $d(x,\bar{t})\geq d(x,t^*)> D/4$ and $Z$ contains all vertices at distance $\leq D/4$ from $\bar{t}$. If $s^*\in Z$, we are done so we must have $d(s^*,\bar{t})> D/4$. Consider the last vertex $a'$ on the $\bar{t}$ to $s^*$ shortest path (in the direction towards $s^*$) for which $d(\bar{t},a')\leq D/4$. We have that $a'\in Z$. Also, the vertex $a$ after $a'$ on the $\bar{t}$ to $s^*$ shortest path must be in $Y$ by definition.

 If $d(\bar{t},s^*)\geq D/2$, we are done. If we are not done, then we get that $d(a,s^*)<D/4$ since $d(\bar{t},a)>D/4$. Hence, $d(a,s_a)<D/4$, so $d(s_a,t^*)> D-2(D/4)=D/2$. 
\end{proof}

It is quite straightforward to extend the the $S$-$T$ Diameter algorithms to work for weighted undirected graphs as well: 

\begin{theorem}
In $\tilde{O}(m\sqrt n)$ time one can obtain an estimate $D'$ to the $S$-$T$ diameter $D$ of an $m$ edge $n$ vertex undirected graph with nonnegative edge weights such that $D/2 - 2w(a,a') \leq D'\leq D$ for some edge $(a,a')$. 

In $\tilde{O}(m^{3/2})$ time one can obtain an estimate $D''$ such that $D/2\leq D''\leq D$.
\end{theorem}

\begin{proof}
The $\tilde{O}(m\sqrt n)$ time algorithm is identical to Algorithm 1, but with BFS replaced by Dijkstra's algorithm. The proof is very similar to that of Theorem~\ref{thm:stdiam32}. The main difference concerns the definition of the vertex $a$, which is the vertex on the shortest path between $\bar{t}$ and $s^*$ with $d(\bar{t},a)=d$. Such a vertex $a$ may not exist here since the graph is weighted. Instead, we let $a'$ be the last vertex on the $\bar{t}$ to $s^*$ shortest path that is at distance $\leq D/4$ from $\bar{t}$, and let $a$ be the vertex after $a'$. 

We include the full analysis of correctness here for completeness. Let $s^*\in S$ and $t^*\in T$ be the end points of the $S$-$T$ Diameter path so that $d(s^*,t^*)=D$.

If some vertex $x\in X$ has $d(x,t^*)\leq D/4$, we get that $d(t_x,s^*)\geq D-2(D/4)=D/2$. If we are not done, all vertices of $X$ have $d(x,\bar{t})\geq d(x,t^*)> D/4$ and $Y$ contains all vertices at distance $\leq D/4$ from $\bar{t}$. If $s^*\in Y$, we are done so we must have $d(s^*,\bar{t})> D/4$. 

Recall that $a'$ is the last vertex on the $\bar{t}$ to $s^*$ shortest path that is at distance $\leq D/4$ from $\bar{t}$, and that $a$ is the vertex after $a'$. We have that $a'\in Y$. If $d(\bar{t},s^*)\geq D/2$, we are done. If we are not done, then we get that $d(a,s^*)<D/4$ since $d(\bar{t},a)>D/4$. Thus, $d(a',s^*)< D/4+w(a,a')$. Therefore, $d(a',s_{a'})< D/4+w(a,a')$, so $d(s_{a'},t^*)< D-2(D/4+w(a,a'))= D/2-2w(a,a')$. This completes the analysis of the $\tilde{O}(m\sqrt n)$ time algorithm. 

For the $\tilde{O}(m^{3/2})$ time algorithm, we apply precisely the same trick from~\cite{ChechikLRSTW14} as the proof of Theorem~\ref{thm:stdiam32}, with identical analysis.
\end{proof}

%% file: diameter_alg_sparse.tex
\subsection{Linear time less than 2-approximation for Diameter }

It is an easy exercise to see that when $D=2h+1$ then the value $\max \{ \epsilon^{in}(v), \epsilon^{out}(v) \}$ of an arbitrary vertex $v \in V$ is an estimation to the diameter which is at least $h+1$ and at most $D$.
In this section we present a deterministic algorithm that gets a directed unweighted graph $G$ with $D=2h$ and computes in $O(m^2/n)$ time an estimation $\hat {D}$ such that $h+1\leq \hat{D} \leq D$.

The algorithm works as follows. A variable $\hat{D}$ is set to zero. The algorithm searches for a vertex $v$ of minimum total degree (where total degree is the sum of in-degree and out-degree). Then the algorithm computes the in and out eccentricity of $v$ and every vertex that has an edge with $v$ (incoming or outgoing). The algorithm outputs the maximum of all the in and out eccentricities that were computed. See Algorithm 2.

\begin{algorithm}
\caption{Fast approximation of the diameter}\label{A-linear}
\begin{algorithmic}[1]
\Procedure{Diam-Approx}{$G$}
\State $\hat{D} =  0$
\State $v = \arg \min_{x\in V} deg^{in}(x)+deg^{out}(x)$

\For{\texttt{every $w \in \Nin(v) \cup \Nout(v) \cup \{v \} $}}
    \State compute $\epsilon^{in}(w)$ and $\epsilon^{out}(w)$
    \State $\hat{D}= \max \{ \hat{D},\epsilon^{in}(w),\epsilon^{out}(w)\}$
\EndFor
\State \Return $\hat{D}$

\EndProcedure
\end{algorithmic}
\end{algorithm}

\begin{theorem}\label{thm:fastersparsediam}
Let $G=(V,E)$ be an unweighted directed graph with diameter $D=2h$ where $h$ is a positive integer. Algorithm ~\ref{A-linear} returns in $O(m^2/n)$ time an estimate $\hat{D}$ such that $h+1\leq \hat{D} \leq D$.
\end{theorem}
\begin{proof}
We start with the running time analysis. Consider the graph $G$ and ignore the edge directions. For every $u\in V$ let $deg(u) = deg^{in}(u)+deg^{out}(u)$. Recall that $v$ is a vertex of minimum degree. Since $m = \frac{1}{2}\sum_{u\in V} deg(u)$, we have $deg(v)\leq 2m/n$. Therefore, the cost of computing in and out eccentricities for all vertices in the set $N(v) \cup \{ v \}$ is $O(\frac{m}{n} \times m)$.

We now turn to bound $\hat{D}$.
Let $a, b\in V$ and let $d(a,b)=2h$. If $d(a,v)\leq h-1$ then $\epsilon^{out}(v)\geq h+1$. Similarly, if $d(v,b)\leq h-1$ then $\epsilon^{in}(v)\geq h+1$. The remaining case is that $d(a,v)=h$ and $d(v,b)=h$. In this case, $v$ is on some shortest path $P(a,b)$ from $a$ to $b$. 

Let $u\in P(a,b)$ be the vertex that precedes $v$ on the $P(a,b)$. Since $u$ has an incoming edge to $v$ it follows that $u\in N(v)$ and  $\epsilon^{out}(u)$ and  $\epsilon^{in}(u)$ are computed. Since $d(a,v)=h$ it follows that $d(a,u)=h-1$, $\epsilon^{out}(u)\geq h+1$ and $\hat{D}$ is at least $h+1$.
\end{proof}

%% file: diameter-alg.tex
\section{Algorithms for dense graphs}\label{sec:diamnewalgs}
In this section we prove the following theorem, which is a restatement of Theorem~\ref{thm:denseintro}.

\begin{theorem} There is an expected $O(n^2\log n)$ time algorithm that for any undirected unweighted graph with Diameter $D=3h+z$ for $h\geq 0,z\in\{0,1,2\}$, returns an extimate $D'$ such that $2h-1\leq D'\leq D$ if $z=0,1$ and $2h\leq D'\leq D$ if $z=2$.

There is an expected $O(n^2\log n)$ time algorithm that for any undirected unweighted graph returns estimates $\eps'(v)$ of the Eccentricities $\eps(v)$ of all vertices such that $3\eps(v)/5-1\leq \eps'(v)\leq \eps(v)$ for all $v$.
\end{theorem}

\subsection{Algorithm overview}
 
Recall that the Diameter approximation algorithm  of Aingworth et al.~\cite{aingworth} runs in $\tilde{O}(n^2+m\sqrt n)$ time. Roditty and Vassilevska W.~\cite{RV13} removed the $\tilde{O}(n^2)$ term to obtain an $\tilde{O}(m\sqrt n)$ expected time almost-$3/2$ approximation algorithm. For every graph with $\Omega(n^{1.5})$ edges the running time of the latter algorithm is not better than the running time of the former algorithm. Therefore, even for not so dense graphs,  it is interesting to consider the opposite question to the one considered by \cite{RV13}. Can the $\tilde{O}(m\sqrt n)$ term be removed? 

For an unweighted undirected  graph of Diameter $D=3h+z$, where $z\in [0,1,2]$, we first show  that using existing techniques it is relatively straightforward to obtain an $\tilde{O}(n^{2})$ time algorithm that returns an estimation $\hat{D}$ such that $2h-2+z\leq \hat{D} \leq D$, when $z \in [0,1]$ and $2h-1\leq \hat{D} \leq D$, when $z=2$.

We then show that using a new implementation of a technique that was introduced by Thorup and Zwick~\cite{ThZw01} in the context of compact routing schemes  we can  return an estimation $\hat{D}$ such that $2h-1\leq \hat{D} \leq D$, when $z \in [0,1]$ and $2h\leq \hat{D} \leq D$, when $z=2$.

The improvement in the estimation might look negligible. To understand the importance of this improvement consider  the case of directed graphs.  The algorithm of Roditty and Vassilevska W.~\cite{RV13} runs in $\tilde{O}(n^{2.5})$ expected time. For $D=3h+z$, where $z\in[0,1,2]$, the estimation $\hat{D}$ satisfies  $2h+z\leq \hat{D} \leq D$, for $z\in[0,1]$ and  $2h+1\leq \hat{D} \leq D$, for $z=2$.  The algorithm of Chechik et al.~\cite{ChechikLRSTW14} 
 runs in $\tilde{O}(n^{8/3})$ expected time and returns an estimation  $\hat{D}$ that satisfies  $\lceil 2/3D \rceil \leq \hat{D} \leq D$.  If we consider for example a graph with $D=5$   we can get an estimation of at least $3$ in $\tilde{O}(n^{2.5})$ time and at least $4$ in  in $\tilde{O}(n^{8/3})$ time. 

In the case of undirected graphs and $D=5$ with the straightforward approach we can only get an estimation of $3$ while using  our more complicated algorithm we can get an estimation of $4$. 
As we showed in earlier sections of this paper every small difference in the approximation might indicate that a conditional lower bound exists, therefore, every   improvement in the quality of the upper bound is crucial for our understanding of the problem.

As we mentioned above, our algorithm is obtained by using ideas developed originally for distance oracles and compact routing schemes. 
Let $a,b\in V$ and let $d(a,b)=D$, both \cite{aingworth} and \cite{RV13} used the following idea. Sample a set $A\subseteq V$ and compute full shortest paths trees for all vertices of $A$. If a vertex that is close to $a$ or $b$ is in $A$ we have a good approximation, if not then all sampled vertices are far from both $a$ and $b$ so pick that farthest one and compute for it and for its $\sqrt n$ closest vertices full shortest paths trees. 
Our algorithm uses a different approach.  As we are allowed to use quadratic time, we try to estimate the distance between every pair of vertices. 
To enable this approach we can no longer sample $A$ naively. Instead, we adapt a recursive sampling algorithm to compute $A$, that was introduced by Thorup and Zwick~\cite{ThZw01} in the context of compact routing schemes. The expected running time of their algorithm is $\tilde{O}(mn/|A|)$. We provide a new implementation of their algorithm that runs in expected $\tilde{O}(n(n/|A|)^2)$ time.

The set $A$ has the following important property,  for every vertex $w\in V$, its {\em cluster} (see \cite{ThZw05})
 $\{u \mid d(u,w)<d(u,A) \}$ is of size $O(n/ |A|)$.  Consider now a pair of vertices $u$ and $v$ that are in the cluster of $w$. For any such pair we can efficiently compute their exact distance. Moreover, we show that for all pairs $u,v$ that are not in the same cluster of any vertex, we can bound $d(u,v)$ from below with $d(u,A)+d(v,A)-1$.
This, combined with some other ideas, gives our approximation guarantees.
We extend our approach to also provide an almost $5/3$-approximation for all Eccentricities. 
The idea of using the bounded clusters of Thorup and Zwick~\cite{ThZw01} has been used in prior work
 to obtain improved distance oracles~\cite{patrascuroditty,AG13}, approximate shortest paths~\cite{BaKa10} and compact routing schemes~\cite{AbrahamG11}.
 
\subsection{A simple approach with additive error}

In this section we present a simple approach for the problem of approximating the Diameter, Eccentricities or $S$-$T$ Diameter, that is based on running existing algorithms on an additive 2 spanner that is precomputed for the input graph.
This simple approach runs in $\tilde{O}(n^2)$ time and gets an estimation that is worse by an additive term of $2$.  We first present this approach for the Diameter. It is simple to adapt it to Eccentricities or $S$-$T$ Diameter.

Suppose that we have an algorithm {\em ALG} that can compute in $\tilde{O}(m\sqrt n)$ time, for any graph $G'$, an estimate $D'$ of its Diameter $D$ such that $p\cdot D - q\leq D'\leq D$. 
Now, Dor, Halperin and Zwick~\cite{DorHZ00} showed that in $\tilde{O}(n^2)$ time one can compute for any $n$ vertex $G$, an additive $2$ spanner $H$ on $\tilde{O}(n^{1.5})$ edges. In fact Knudsen~\cite{Knudsen17} recently showed that in $O(n^2)$ time one can get $H$ on $O(n^{1.5})$ edges (i.e. he removed all logs!).

Let's compute $H$ for our given graph and run {\em ALG} on $H$. The runtime is $\tilde{O}(n^{1.5}\cdot \sqrt n)\leq \tilde{O}(n^2)$ since $H$ has $\leq O(n^{1.5})$ edges.

Let $D'_H$ be the estimates that we obtain  for the Diameter $D_H$ of $H$. 
Notice that $pD-q\leq p\cdot D_H - q\leq D'_H\leq D_H\leq D+2$ and so $pD-2-q \leq D'_H-2\leq D$.
Thus, in $\tilde{O}(n^2)$ time we get almost the same guarantees as in the $\tilde{O}(m\sqrt n)$ time algorithm, except for an extra additive loss of $2$ in the quality.

For the case of Eccentricities and $S$-$T$ Diameter the same approach works without a change. 
Suppose that we have an algorithm {\em ALG} that can compute in $\tilde{O}(m\sqrt n)$ time estimates $e(v)$ of $\epsilon(v)$ for all $v$ so that $r\epsilon(v)-s\leq e(v)\leq \epsilon(v)$, and an estimate $D''$ of the $S$-$T$ Diameter $D_{S,T}$ so that $t\cdot D_{S,T}-u\leq D''\leq D_{S,T}$.
Let $e_H(\cdot),D''_H$ be the estimates that we obtain respectively for  the Eccentricities $\epsilon_H(\cdot)$ of $H$ and the $S,T$ Diameter $D^H_{S,T}$. Let's return $e_H(\cdot)-2,D''_H-2$ as our estimates for the  Eccentricities and $S$-$T$ Diameter of $G$.

Since $r\epsilon(v)-s\leq r\epsilon_H(v)-s\leq e_H(v)\leq \epsilon_H(v)\leq \epsilon(v)+2$, we get 
$r\epsilon(v)-s-2\leq e_H(v)-2\leq \epsilon(v)$.

Finally since $t\cdot D_{S,T}-u\leq t\cdot D^H_{S,T}-u\leq D''_H\leq D^H_{S,T}\leq D_{S,T}+2$, we get $t\cdot D_{S,T}-u-2\leq  D''_H-2\leq D_{S,T}$.

Below we show how to make the additive loss in quality smaller for Diameter and Eccentricities. This is especially important when these parameters are constant, which is one of the hard cases of the problems.\footnote{If, for example, the diameter is polynomial in $n$, say $n^{\varepsilon}$, then we can approximate the diameter to an arbitrary precision of $1+\delta$ in $\tilde{O}(mn^{1-\varepsilon}/\delta)$ time by sampling a vertex on the true diameter path of distance at most $\frac{\delta}{1+\delta}n^\varepsilon$ from one of the true diameter endpoints.}

\subsection{Near linear almost 3/2-approximation for Diameter}

%
%
%
%

Thorup and Zwick~\cite{ThZw05} introduced distance oracles, a succinct data structure for answering approximate distance queries efficiently.
Among the tools they use are clusters and bunches. Let $A\subseteq V$, let $p_A(u)$ be the closest vertex to $u$ from $A$, where ties are broken in favor of the vertex with a smaller identifier and let $d(u,A)=d(u,p_A(u))$.
For every $v\in V$, let $B_A(u)=\{ v\in V\mid d(u,v) <d(u,A) \}$ be the \emph{bunch} of $u$. For every $w\in V\setminus A$, let $C_A(w)= \{ v \mid w\in B_A(v)\}$ be the \emph{cluster} of $w$.

Thorup and Zwick~\cite{ThZw05} showed that if a set $A$ is formed by adding every vertex of $V$ to $A$ with probability $p$ then the expected size of $B_A(v)$ is $O(1/p)$, for every $v\in V$.
They also showed, in the context of compact routing schemes~\cite{ThZw01}, that if the set $A$ is constructed by a recursive sampling algorithm then it is possible to bound the maximum size of a cluster as well.
They also showed, in the context of compact routing schemes~\cite{ThZw01}, that if the set $A$ is constructed by a recursive sampling algorithm then it is possible to bound the maximum size of a cluster as well.
Their algorithm works as follows. It sets $A$ to the empty set and $W$ to $V$. Next, as long as the set $W$ is not empty the algorithm samples from $W$ vertices with probability $p$ and adds the sampled vertices to $A$.
The algorithm computes $C_A(w)$ for every $w\in W$ and removes from $W$ all the vertices whose cluster has at most $4/p$ vertices with respect to the updated $A$. The pseudo-code is given in Algorithm~\ref{A-center-TZ}.

\begin{algorithm}
\caption{Thorup and Zwick center algorithm}\label{A-center-TZ}
\begin{algorithmic}[1]
\Procedure{center}{$G,p$}
\State $A= \emptyset$
\State $W =V$

\While{$W\neq \emptyset$}
    \State $X$ - random sample of vertices from $W$, $|X|=|W|p$
    \State $A= A \cup X$
    \State  $W = \{ w\in V \mid |C_A(w)|> 4/p \}$

\EndWhile

\State \Return $A$

\EndProcedure
\end{algorithmic}
\end{algorithm}

Thorup and Zwick proved the following Theorem:
\begin{theorem}[Theorem 3.1 from ~\cite{ThZw01}]
The expected size of the set $A$ returned by Algorithm~\ref{A-center-TZ}
is at most $2np \log n$. For every $w \in V$ we then
have $|C_A(w)|\leq 4/p$.
\end{theorem}

Thorup and Zwick claimed that the expected running time of Algorithm~\ref{A-center-TZ} is $O(mnp \log n)$.
They did not provide the details and refer the reader to~\cite{ThZw05}.
However, an educated guess is that they compute clusters for the vertices currently in $W$ in each iteration of the while loop, which results in the claimed running time.

The starting point of the Diameter and Eccentricities algorithms presented in this section is an $O(n/p^2 \log n)$ expected time implementation of Algorithm~\ref{A-center-TZ}.

The first idea behind our implementation is that, as opposed to what Thorup and Zwick did, we will compute the bunches and use them to compute the clusters and the set $W$.
This can be done as follows. Once we have computed $B_A(v)$ for every $v\in V$, we can scan $B_A(v)$, and for every $w\in B_A(v)$ we can add $v$ to $C_A(w)$.
The cost of this process is $O(\sum_{v\in V} |B_A(v)|)$ and since the clusters are by definition the inverse of the bunches, at the end of this process we have $C_A(w)$ and $|C_A(w)|$, for every $w\in V$ and we can compute $W$ (as needed in Algorithm~\ref{A-center-TZ}).

However, in the current implementation only the expected size of a bunch is bounded, and since the Thorup-Zwick bound on the number of iterations is $O(\log n)$ in expectation as well, we cannot apply this idea directly to deduce a good expected running time. To this end, more ideas are needed.

The following simple observation helps us to achieve our goal.

\begin{observation}\label{O-bunch}
Let $A_i$ be the set $A$ after updating it in the beginning of the $i$-th iteration of the while loop in Algorithm~\ref{A-center-TZ}. Let $A^*$ be a set such that $A^*\subseteq A_i$, for every $i\geq 1$.
For every $v\in V$ it holds that $B_{A_i}(v) \subseteq B_{A^*}(v)$.
\end{observation}

It follows from this observation that we only need to pick the first set $A_1$ such that $|B_{A_1}(v)|\leq O(1/p)$ for every $v\in V$.

It is folklore that the $s$ closest vertices $N_{s}(v)$ to a vertex $v$ can be computed in $O(s^2)$ time~\cite{DorHZ00}.
This implies that we can compute $N_{1/p}(v)$ for every $v\in V$ in $O(n/p^2)$ time.
It is not hard to see   that, given the sets $N_{1/p}(v)$ of all $v\in V$, one can (deterministically) compute a ``hitting'' set $A$ of size $O(np \log n)$ in $O(n+n/p)$ worst case time, so that $N_{1/p}(v)\cap A\neq \emptyset $ for every $v\in V$ (a greedy algorithm works; e.g. see~\cite{ThZw05}).

The second idea behind our implementation is that we first compute the sets $N_{1/p}(v)$ for every $v\in V$ and the hitting set $A$,
as described above. Then, using these sets, we initialize Algorithm~\ref{A-center-TZ} with a set $A$ such that $|B_A(v)|=O(1/p)$, for every $v\in V$.

%


In more detail, our algorithm works as follows. For every $v\in V$ it computes the set $N_{1/p}(v)$ in $O(n/p^2)$ time. Then it finds a  set $A$ such that  $N_{1/p}(v)\cap A\neq \emptyset $ for every $v\in V$. Given the hitting set $A$, it computes $d(v,A)$ and $p_A(v)$ for every $v\in V$. Using $d(v,p_A(v))$ and
$N_{1/p}(v)$ it computes for every $v\in V$ the bunch $B_A(v)$. Finally, it computes the clusters and $W$ using the bunches as we described above. The rest of the algorithm is almost identical to Algorithm~\ref{A-center-TZ}.
The only difference is that we compute the bunches and use them to compute the clusters and the set $W$.
The pseudo-code is given in Algorithm~\ref{A-center}.

\begin{algorithm}
\caption{New implementation of Thorup and Zwick center algorithm}\label{A-center}
\begin{algorithmic}[1]
\Procedure{center}{$G,p$}
\State compute $N_{1/p}(v)$ for every $v\in V$.
\State $A=$ hitting set of the sets $N_{1/p}(v)$, where $v\in V$.
\State compute $d(v,A)$ and $p_A(v)$ for every $v\in V$.
\State compute $B_A(v)$ using $N_{1/p}(v)$ and $d(v,p_A(v))$.
\For{\texttt{every $u\in V$}}
    \State compute $C_A(u)$ using $B_A(\cdot)$
\EndFor
  \State  $W = \{ w\in V \mid |C_A(w)|> 4/p \}$

\While{$W\neq \emptyset$}
    \State $X$ - random sample of vertices from $W$, $E[|X|]=np$
    \State $A= A \cup X$
    \For{\texttt{every $v\in V$}}
    \State compute $B_A(v)$
    \EndFor

\For{\texttt{every $u\in V$}}
    \State compute $C_A(u)$ using $B_A(\cdot)$
\EndFor

  \State  $W = \{ w\in V \mid |C_A(w)|> 4/p \}$

\EndWhile

\State \Return $A$

\EndProcedure
\end{algorithmic}
\end{algorithm}

We show:

\begin{lemma}\label{L-A-center}
Algorithm~\ref{A-center} computes in $O(n/p^2 \log n)$ expected time a set $A$ of expected size $O((pn)\cdot  \log n)$
that guarantees for every vertex $w\in V \setminus A$ that $|C_A(w)|=O(1/p)$, and for every $v\in V$ that $|B_A(v)|=O(1/p)$.
\end{lemma}
\begin{proof}

The cost of computing $N_{1/p}(v)$ for every $v\in V$ is $O(n (1/p)^2)$~\cite{DorHZ00}. The cost of computing $A$ is $O(np)$ time~\cite{ThZw05}. Computing  $d(v,A)$ and $p_A(v)$ for every $v\in V$ in $O(m)$ time is straightforward by running shortest paths tree computation from a
dummy vertex that is connected to the set $A$. To compute $B_A(v)$ using $N_{1/p}(v)$ we only scan $N_{1/p}(v)$, thus, the total cost is $O(n(1/p))$. As we explained earlier the cost of computing clusters using bunches is
$O(\sum_{v\in V}|B_A(v)|)$. Since for every $v\in V$ we have $B_A(v)\subseteq N_{1/p}(v)$ the total cost is $O(n(1/p))$.

This completes the analysis of the part that precedes the while loop.
Next, we analyze the cost of the while loop.

Let $A^*$ be the set $A$ that was computed before the while loop and let $A_i$ be the set $A$ after updating it in the beginning of the $i$-th iteration of the while loop.
From Observation~\ref{O-bunch} it follows that $B_{A_i}(v) \subseteq B_{A^*}(v)$ and therefore in every iteration the cost of computing bunches from scratch
is at most $O(n (1/p)^2)$ as $|B_{A^*}(v)|=O(1/p)$, for every $v\in V$. One can also compute $B_{A_{i+1}}(v)$ from $B_{A_i}(v)$ by first computing $d(v,A_{i+1})$ and if $d(v,A_{i+1})<d(v,A_i)$ to prune $B_{A_{i+1}}(v)$
accordingly at a smaller cost of $O(n(1/p) + m)$, however this does not affect the overall complexity.

Thorup and Zwick~\cite{ThZw01} proved that the expected number of iterations is $O(\log n)$. 
They show that if the set $A$ in each iteration is     chosen from $W$
uniformly at random with probability $p$ then in each iteration with
probability $1/2$ the size of $W$ decreases by a factor of $2$. Thus, the fact that  the set $A$ from which we start is different does not
affect the correctness proof, since the set $A$ is still constructed in each iteration in the same way. 

Therefore, we conclude that there are only $O(\log n)$ iterations in expectation.
This implies that a set $A$ of expected size $(np \log n)$ is returned in $O(n/p^2 \log n)$ expected time.  The algorithm stops only when there are no large clusters, thus the bound on the cluster size follows. As we mentioned above the algorithm starts with bunches that satisfy the required bound and their
size can only decrease afterwards, thus the bound on the bunches follows.
\end{proof}

\begin{algorithm}
\caption{almost $3/2$-Approximation for Diameter}\label{A-diam}
\begin{algorithmic}[1]
\Procedure{$3/2$-Approx-Diam}{G}
\State $M$ - $n\times n$ matrix whose entries are set to $n$
\State $A =$ CENTER$(G,1/\sqrt n)$\;

\For{\texttt{every $w\in V$}}
\Comment{Step 1}
\For{\texttt{every $\langle u, v\rangle\in C_A(w) \times C_A(w)$, s.t. $u\neq v\;$}}
    \State $M(u,v) = \min (M(u,v), d(u,w)+d(v,w))$
\EndFor
\EndFor

\For{\texttt{every $\langle u, v\rangle \in V \times V$, s.t. $M(u,v)=n$}}
\Comment{Step 2}
    \State $M(u,v)  = d(u,A)+d(v,A)-1$
\EndFor

\State $H$ - an additive 2 spanner of $G$
\Comment{Step 3}
\For{\texttt{every $u \in A$}}
    \State compute shortest paths tree for $u$ in $H$ and set $\epsilon_H(u)$, the Eccentricity of $u$ in $H$
\EndFor

\State $D_1 = \max_{\langle u, v\rangle \in V \times V} M(u,v)$
\State $D_2 = \max_{u\in A} \epsilon_H(u)$
\State $\hat{D} = \max(D_1,D_2-2)$

\State \Return $\hat{D}$

\EndProcedure
\end{algorithmic}
\end{algorithm}

We can now turn to describe the new Diameter algorithm.
The algorithm works as follows. All entries of an $n \times n$ matrix $M$ are set to $n$. A set $A$ of centers is computed using the algorithm of Thorup and Zwick~\cite{ThZw01}.
For every vertex $w\in V$ and every pair $\langle u, v\rangle\in C_A(w)\times C_A(w)$  the algorithm sets $M(u,v)$ to $\min (M(u,v), d(u,w)+d(v,w))$ (Step 1).
Next, the algorithm searches the matrix $M$ for entries whose value is still $n$. Given  a pair $\langle u, v\rangle\in V \times V$ for which $M(u,v)=n$ the algorithm sets $M(u,v)$ to
$d(u,A)+d(v,A)-1$ (Step 2). Finally, the algorithm computes  an additive 2 spanner $H$ of the input graph $G$ and for every $u\in A$  it computes $\epsilon_H(u)$, the Eccentricity of $u$ in $H$ (Step 3).
The algorithm outputs the maximum between $D_1$ and $D_2-2$, where $D_1$ is  $\max_{\langle u, v\rangle \in V \times V} M(u,v)$ and $D_2$ is $\max_{u\in A} \epsilon_H(u)$.

Next, we bound the value returned by Algorithm~\ref{A-diam}.

\begin{theorem}\label{thm:densediam}
For an unweighted undirected graph $G$ with diameter $D=3h+z$ where $h$ is a positive integer and $z\in [0,1,2]$, the value $\hat{D}$ returned by Algorithm~\ref{A-diam} satisfies:

$$
	\begin{array}{ll}
		2h-1  & \mbox{if } z \in [0,1] \\
		2h & \mbox{if } z=2
	\end{array}\leq \hat{D} \leq  D
$$
\end{theorem}
\begin{proof}

We start with the following Lemma:

\begin{lemma}\label{L-INTER}
Let $u,v\in V$ and let $P(u,v)$ be a shortest path between $u$ and $v$.
If $B_A(u) \cap B_A(v)\neq \emptyset$ then $(B_A(u) \cap B_A(v)) \cap P(u,v) \neq \emptyset$.
\end{lemma}
\begin{proof}
If $v\in B_A(u)$ then the claim trivially holds so we can assume that $v\notin B_A(u)$.
Let $w$ be the vertex farthest from $u$ that is in $B_A(u)  \cap P(u,v)$.  From the definition of $w$ it follows that $d(u,w)=d(u,A)-1$. Assume, towards a contradiction, that $(B_A(u) \cap B_A(v)) \cap P(u,v) = \emptyset$.
This implies that $w\notin B_A(v)$ and $d(v,A)-1<d(v,w)=d(u,v)-d(u,w)$. However, since $B_A(u) \cap B_A(v)\neq \emptyset$ there is a vertex $w'$ such that $d(u,w')\leq d(u,w)$ and $d(v,w')\leq d(v,A)-1 < d(u,v)-d(u,w)$. This implies that $d(u,w')+d(v,w')< d(u,v)$, a contradiction to the triangle inequality.
\end{proof}

\begin{lemma}\label{L-NO-INT}
Let $u,v\in V$. If $B_A(u) \cap B_A(v)= \emptyset$ then $d(u,A)+d(v,A)-1 \leq d(u,v)$.
\end{lemma}
\begin{proof}
Notice first that $B_A(u)$ (resp., $B_A(v)$) contains all the vertices at distance $d(u,A)-1$ (resp., $d(v,A)-1$). Let $P(u,v)$ be a shortest path between $u$ and $v$.
Let $w$ be the vertex farthest from $u$ on $P(u,v)$ that is also in $B_A(u)$. Similarly, let $w'$ be the vertex farthest from $v$ on $P(u,v)$ that is also in $B_A(v)$. Since $B_A(u) \cap B_A(v)= \emptyset$ it holds that $w\neq w'$.
Therefore: $$d(u,v) = d(u,A)-1+d(v,A)-1+d(w,w')\geq d(u,A)+d(v,A)-1.$$\end{proof}

Let $a$ and $b$ be the Diameter endpoints, that is $d(a,b)=D=3h+z$, where $z\in [0,1,2]$. Let $P(a,b)$ be a shortest path between $a$ and $b$.

Assume first that $B_A(a) \cap B_A(b)\neq \emptyset$. It follows from Lemma~\ref{L-INTER} that there is a vertex $w\in P(a,b)$ such that $\langle a, b\rangle \in C(w) \times C(w)$.
Therefore,  $M(a,b)=D$ after Step 1.
After the update in Step 2 it follows from Lemma~\ref{L-NO-INT} that $M(u,v)\leq d(u,v)$ for every $u ,v\in V$. Therefore, the maximum value in the matrix is $d(a,b)$ and $D_1=D$.
Let $x = \argmax_{y\in A} \epsilon_H(y)$. Since $H$ is an additive 2 spanner it holds that $\epsilon_H(x)\leq D+2$, hence, we have $D_2 \leq D$ and the algorithm returns the exact value of the Diameter.

Assume now that $B_A(a) \cap B_A(b) = \emptyset$.
From the discussion of the previous case it follows that in this case $\hat{D}\leq D$ as well. Thus, it is only left to prove the lower bound. 
Assume that $z\in [0,1]$. Consider first the case that $d(a,A)\geq h$ and $d(b,A)\geq h$ then from Lemma~\ref{L-NO-INT} it follows that $M(a,b)\geq 2h-1$ after Step 2 and $D_1$ is at least $2h-1$.
If this is not the case then either $d(a,A)< h$ or $d(b,A)< h$ (or both). Assume, wlog, that $d(a,A)< h$. In this case the Eccentricity in $H$ of at least one vertex from $A$ is at least $2h+1$  and hence $D_2-2$ is at least $2h-1$.

Assume now that $z=2$.
If either $d(a,A)\geq h$ and  $d(b,A)> h$ or $d(a,A)>h$ and  $d(b,A)\geq h$ then from Lemma~\ref{L-NO-INT} it follows that $M(a,b)\geq 2h$ after Step 2 and $D_1$ is at least $2h$. If this is not the case then either $d(a,A)\leq h$ or $d(b,A)\leq h$ (or both). Assume, wlog, that $d(a,A)\leq h$. In this case the Eccentricity in $H$ of at least one vertex from $A$ is at least $2h+2$  and hence $D_2-2$ is at least $2h$.
\end{proof}

We now turn to we analyze the running time of Algorithm~\ref{A-diam}.

\begin{theorem}
For an unweighted undirected graph $G$, the expected running time of  Algorithm~\ref{A-diam} is $O(n^2 \log n)$.
\end{theorem}
\begin{proof}
The set $A$ is computed by the center algorithm presented in Algorithm~\ref{A-center} with $p=1/\sqrt n$. From Lemma~\ref{L-A-center} it follows that the size of the set $A$ is $O(\sqrt n \log n)$ and its construction time is $O(n^2\log n)$ in expectation. For every $w\in V$ the size of $C_A(w)$ is $O(\sqrt n)$. Therefore, Step 1 takes $O(n \times |C_A(w)|^2)=O(n^2)$.
Step 2 takes $O(n^2)$ time as well.  In Step 3 we first compute an additive 2 spanner $H$ on $O(n^{1.5})$ edges. Knudsen~\cite{Knudsen17}, following Dor, Halperin and Zwick~\cite{DorHZ00} showed how to do this in $O(n^2)$ time. We also compute $|A|$ shortest paths trees in $H$. As $H$ has $O(n^{1.5})$ edges, this step  takes $O(n^2 \log n)$ time.
\end{proof}

\subsection{Near linear almost 5/3-approximation for Eccentricities}

\begin{algorithm}
\caption{almost $5/3$-Approximation for all Eccentricities}\label{A-ecc}
\begin{algorithmic}[1]
\Procedure{$5/3$-Approx-Ecc}{G}
%
%
%
\State Run lines 2-11 of Algorithm~\ref{A-diam}, with $H$ augmented with shortest paths trees for $B_A(u) \cup \{ p(u) \}$ for every $u\in V$
\For{\texttt{every $u \in V$}}
    \State $\epsilon_1(u) = \max_{v\in V} M(u,v)$
    \State $\epsilon_2(u) = \epsilon_H(p(u))-d(u,p(u))-2$
    \State $\epsilon_3(u) = d_H(u,y)-2$, where $y=\argmax_{x\in A} d_H(u,x)$
    \State $\epsilon'(u) = \max(\epsilon_1(u),\epsilon_2(u),\epsilon_3(u))$
\EndFor
\EndProcedure
\end{algorithmic}
\end{algorithm}

Next, we show how to update Algorithm~\ref{A-diam} to obtain an almost $5/3$ approximation for all Eccentricities.
We run lines 2-11 of Algorithm~\ref{A-diam}. The only difference is that $H$ is augmented with the edges of the shortest paths tree that span the set $B_A(u) \cup \{ p(u) \}$ for every $u\in V$.
Then, for every $u\in V$ we  compute $\epsilon_1(u)$, $\epsilon_2(u)$ and $\epsilon_3(u)$, which are defined as follows:
$\epsilon_1(u) = \max_{v\in V} M(u,v)$,  $\epsilon_2(u) = \epsilon_H(p(u))-d(u,p(u))-2$ and  $\epsilon_3(u) = d_H(u,y)-2$, where $y=\argmax_{x\in A} d_H(u,x)$.
The algorithm sets $\epsilon'(u)$ to $\max \{ \epsilon_1(u), \epsilon_2(u), \epsilon_3(u) \}$ for every $u\in V$ as an estimation to $\epsilon(u)$. The pseudo-code is given in Algorithm~\ref{A-ecc}.

We now prove:

\begin{theorem}\label{thm:fastdenseecc}
For an unweighted undirected graph $G$, for every $u\in V$,  Algorithm~\ref{A-ecc} computes in $O(n^2 \log n)$ expected time  a value $\epsilon'(u)$ that satisfies:  $\frac{3\epsilon(u)}{5}-1 \leq \epsilon'(u)\leq  \epsilon(u).$
\end{theorem}
\begin{proof}
We start by analyzing the running time. Lines 2-11 of the algorithm are the same as Algorithm~\ref{A-diam}, with one difference, the spanner $H$ is augmented with the edges of a shortest paths tree rooted at $u$ that span the set $B_A(u) \cup \{ p(u) \}$, for every $u\in V$.
This adds at most $O(n^{1.5})$ edges to $H$ and hence the cost of these lines remain $O(n^2 \log n)$ time in expectation. The computation of $\epsilon_1(u)$, $\epsilon_2(u)$ and $\epsilon_3(u)$ for every $u\in V$ costs $O(n^2)$ time in total.

Let $u\in V$ be an arbitrary vertex and let $\epsilon(u)=d(u,t)$.  We now turn to bound $\epsilon'(u)$.

In our analysis we will use the following simple observation:
\begin{observation}\label{O-ecc}
In an undirected graph it holds for every $u,v \in V$ that $\epsilon(u)\geq \epsilon(v)-d(u,v)$.
\end{observation}

It is straightforward to see that both  $\epsilon_2(u)$  and $\epsilon_3(u)$ are at most $\epsilon(u)$.
Recall that $\epsilon_2(u) = \epsilon_H(p(u))-d(u,p(u))-2\leq \epsilon(p(u))-d(u,p(u))\leq \epsilon(u)$
and $\epsilon_3(u) = d_H(u,y)-2\leq d(u,y)\leq \epsilon(u)$.

We distinguish between two cases.

{\bf Case 1:} $B_A(u) \cap B_A(t)\neq \emptyset$. It follows from Lemma~\ref{L-INTER} that $P(u,t)\cap (B_A(u) \cap B_A(t)) \neq \emptyset$ and $M(u,t)= \epsilon(u)$.
From Lemma~\ref{L-NO-INT} it follows that $M(u,w)\leq d(u,w)$ for every $w\in V$ after Step 2. Therefore, $\epsilon_1(u)=\epsilon (u)$. Since $\epsilon_2(u)\leq \epsilon(u)$ and $\epsilon_3(u)\leq \epsilon(u)$ we get that $\epsilon'(u)=\epsilon(u)$.

{\bf Case 2:} $B_A(u) \cap B_A(t) =  \emptyset$. Consider first the case that $d(u,p(u))\leq \frac{\epsilon(u)}{5}-1$.
From Observation~\ref{O-ecc} we get that $\epsilon_H(p(u))\geq \epsilon_H(u)-d_H(u,p(u))$. As we augmented $H$ with a shortest paths tree that spans $B_A(u) \cup \{ p(u) \}$ we have $d(u,p(u))=d_H(u,p(u))$ and we get  $\epsilon_H(p(u))\geq \epsilon_H(u)-d(u,p(u))$.
Hence, we get that $\epsilon_2(u) = \epsilon_H(p(u))-d(u,p(u))-2 \geq  \epsilon_H(u)-2d(u,p(u))-2\geq \epsilon(u)-2d(u,p(u))-2$.
As before we have $\epsilon_2(u) \leq  \epsilon(u)$. Using $d(u,p(u))\leq \frac{\epsilon(u)}{5}-1$ we get that:

 $$\epsilon_2(u) \geq \epsilon(u)-\frac{2\epsilon(u)}{5}\geq \frac{3\epsilon(u)}{5}.$$

Assume now that  $d(u,p(u))\geq  \frac{\epsilon(u)}{5}$. This means that $d(u,A)-1\geq \frac{\epsilon(u)}{5}-1$.

Let $S$ be the set of all vertices $v\in V$ such that $B_A(u) \cap B_A(v)=  \emptyset$, that is, $S=V \setminus \cup_{w\in B_A(u)} C_A(w)$.
Let $t' = \argmax_{x\in S} d(x,A)-1$. If $d(t',A)-1 \geq   \frac{2\epsilon(u)}{5}-1$ we get from Lemma~\ref{L-NO-INT} that $M(u,t')\geq \frac{3\epsilon(u)}{5}-1$.
Assume now that  $d(t',A) <   \frac{2\epsilon(u)}{5}$. As $t'$ is the farthest vertex from $A$ we get that $d(t,p(t))<\frac{2\epsilon(u)}{5}$ and $d(u,p(t))>\frac{3\epsilon(u)}{5}$.
Therefore, $\epsilon_3(u) = d_H(u,y)-2\geq d(u,p(t))-2\geq \frac{3\epsilon(u)}{5}-1$.

From Lemma~\ref{L-INTER} and Lemma~\ref{L-NO-INT} it follows that $\epsilon_1(u)\leq \epsilon(u)$ and the bound follows.
\end{proof}

%% file: diameter-alg-mmult.tex
\subsection{Algorithms for dense graphs using matrix multiplication}

Here we will give $O(n^{2.05})$ time
approximation algorithms for Diameter and Eccentricities in dense unweighted undirected graphs. The approximation guarantees of these algorithms are slightly better than those in our $O(n^2\log n)$ time algorithm. In fact, the guarantees are exactly the same as in the $\tO(m\sqrt n)$ time algorithms for Diameter and Eccentricities of Roditty and Vassilevska W.~\cite{RV13} and Cairo et al.~\cite{cairo}. Specifically, we prove the following theorem.
 \begin{theorem}\label{thm:205}
There is an $O(n^{2.045})$ time randomized algorithm that with high probability outputs an almost $3/2$-approximation $\tilde{D}$ to the Diameter $D$ and almost $5/3$-approximations $e(v)$ to all Eccentricities $\epsilon(v)$ in an unweighted undirected graph:
\begin{enumerate}
\item $\frac{2D-1}{3} \leq \tilde{D}\leq D$.
\item For every vertex $v$, $\frac{3\epsilon(v)-1}{5}\leq \tilde{\epsilon}(v)\leq \epsilon(v)$.
\end{enumerate}
\end{theorem}

To achieve this, we give an efficient implementation using fast matrix multiplication of the $\tO(m\sqrt n)$ time algorithms of \cite{cairo} and \cite{RV13}.

The main overhead of the $\tilde{O}(m\sqrt n)$ time algorithms \cite{cairo,RV13} is in computing the distances from a set $S$ of $O(\sqrt n\log n)$ vertices. Computing the set $S$ itself can be done in linear time. In particular, $S$ is defined as the union of a set $W$, a set $T$, and a vertex $w$. The set $W$ is simply a random sample. The vertex $w$ is the farthest vertex from $W$, which can be computed in linear time via BFS from a dummy vertex adjacent to every vertex in $W$. The set $T$ is defined as the closest $\sqrt n$ vertices to $w$, and can be computed by BFS from $w$. After one knows all distances from every $s\in S$ to every $v\in V$, it takes linear time to output the Diameter and Eccentricity estimates.

The main idea of our algorithms is as follows. If the Diameter is of size $\leq O(\log n)$, then one does not need all distances between $S$ and $V$, but only those that are $O(\log n)$. Small distances are easy to compute with matrix multiplication. Let $A$ be the adjacency matrix and $A_S$ be its submatrix formed by just the rows in $S$. Then we can find the distances for all pairs in $S\times V$ at distance $\leq t$ by computing $A_S\times A^{t-1}$, which can be computed by performing $t-1$ matrix products of dimension $|S|\times n$ by $n\times n$, and this can be accomplished in $O(tn^{2.05})$ time~\cite{legallrectnew,legallrect}. If on the other hand the Diameter is $D\geq 100\log n$, then one can use an $\tilde{O}(n^2)$ time algorithm by Dor et al.~\cite{DorHZ00} to compute estimates of all pairwise distances with an additive error at most $4\log n$. The maximum distance estimate computed, minus $4\log n$, will be between $0.96 D$ and $D$, giving a really good approximation already. A similar argument works for Eccentricities, and also for $S$-$T$ Diameter.


Below we recap the guarantees of the $\Ot(m\sqrt n)$ time approximation algorithms of \cite{cairo, RV13}.

\begin{theorem}[\cite{cairo, RV13}] \label{cairoalgs} 
The following can be computed in $\tilde{O}(m\sqrt{n})$ time with high probability:
\begin{enumerate}
\item an estimate $\hat{D}$ of the graph Diameter $D$, such that $\frac{2}{3} D-\frac{1}{3}\leq \hat{D}\leq D$,
\item for every vertex $v$, an estimate $e(v)$ of its Eccentricity $\epsilon(v)$, such that $\frac{3}{5} \epsilon(v)-\frac{1}{5}\leq e(v)\leq \epsilon(v)$.
\end{enumerate}
\end{theorem}

Using Seidel's algorithm~\cite{Seidel} we can compute all the distances exactly, and hence the above parameters as well, all in $O(n^\omega)$ time for $\omega< 2.373$. We will show that for dense graphs, we can obtain the same approximation guarantees as in Theorem~\ref{cairoalgs}, in time $O(n^{2.05})$.

Let us compare to our $O(n^2\log n)$ time algorithms. For Diameter $D=3h+z$, the $O(n^2\log n)$ time algorithm returns an estimate $2h-1$ when $z=0,1$ and $2h$ when $z=2$. The estimate $\hat{D}$ here is $\geq (2D-1)/3= 2h+ (2z-1)/3$, which is $\geq 2h$ when $z=0$ and $\geq 2h+1$ when $z=1,2$.

For Eccentricities, the $O(n^2\log n)$ time algorithm returns estimates $e(v)\geq 3\epsilon(v)/5 -1$, and here we return a better estimate $e(v)\geq (3\epsilon -1)/5$.

We will rely on two known algorithms. The first is from a paper by Dor, Halperin and Zwick~\cite{DorHZ00} on additive approximations of All-Pairs Shortest Paths (APSP). Among many other results, \cite{DorHZ00} show that in $\tilde{O}(n^2)$ time, one can compute for all pairs of vertices $u,v$, an estimate $d'(u,v)$ of their distance $d(u,v)$ so that $d(u,v)\leq d'(u,v)\leq d(u,v)+a\log n$ for an explicit constant $a\leq 4$.

The second is an algorithm for the following truncated multi-source shortest paths problem: given an integer $Q$, a graph $G=(V,E)$ and a set $S$, compute the distances $d(s,v)$ for every $s\in S$ and $v\in V$ for which $d(s,v)\leq Q$. 

The algorithm uses fast matrix multiplication and is quite straightforward. Let $A$ be the $n\times n$ Boolean matrix with rows and columns indexed by $V$, so that $A[u,v]=1$ if there is an edge between $u$ and $v$ or $u=v$, and $A[u,v]=0$ otherwise; i.e. $A$ is the adjacency matrix added to the identity matrix.
Let $A_S$ be the $|S|\times n$ submatrix of $A$ consisting of the rows indexed by vertices of $S$.
For an integer $i\geq 1$, let $A^i$ be the $i$-th power of $A$ under the Boolean matrix product. Here, $A^i[u,v]=1$ if and only if the distance between $u$ and $v$ is at most $i$. Define $A^0$ as the identity matrix. Consider $A_S \cdot A^i$ for any choice of $i\geq 0$ (under the Boolean matrix product). Here, $(A_S \cdot A^i)[s,v]=1$ if and only if the distance between $s$ and $v$ is at most $i+1$. 
Thus, if we compute $D_i:=A_S\cdot A^i$ for every $1\leq i< Q$, we would know the distance from every $s\in S$ to every $v\in V$, whenever this distance is at most $Q$.
Computing these matrix products can easily be done by performing the following $Q-1$ Boolean products of an $|S|\times n$ matrix by an $n\times n$ matrix: let $D_0=A_S$; then for each $i$ from $1$ to $t-1$, compute $D_i:=D_{i-1}\cdot A$. Thus, the running time is $O(Q\cdot \mathcal{M}(|S|,n,n))$ where $\mathcal{M}(|S|,n,n)$ is the runtime of multiplying an $|S|\times n$ matrix by an $n\times n$ matrix.

Armed with these two algorithms, let us recap Roditty et al.'s (and Cairo et al.'s) approximation algorithm and see how to modify it.
The algorithm proceeds as follows: Let $D$, $R$ and $\epsilon(v)$ denote the Diameter and Radius of $G$ and the Eccentricity of vertex $v$, respectively.\\
%
%
%
%
%

\begin{algorithm}
\caption{RV/CGR Algorithm}\label{cairoalg}
\begin{algorithmic}[1]
\State Using BFS (see \cite{RV13} and \cite{cairo}), in $O(m+n)$ time compute $W,w,T$, where $W\subseteq V$ is a uniformly chosen subset of size $O(\sqrt n \log n)$, $w$ is the furthest vertex from $W$ and $T$ are the closest $\sqrt n$ vertices to $w$.
Let $S=\{w\}\cup W\cup T$.

\State For every $s\in S$ and every $v\in V$, compute the distance $d(s,v)$ between $s$ and $v$; set $\epsilon(s)=\max_v d(s,v)$.

\State Set $\tilde{D}=\max_{x\in S} \epsilon(x)$. 


\State Set for every $v\in V$, $\tilde{\epsilon}(v)=\max\{d(w,v),\max_{x\in W} d(x,v), \max_{x\in T} (\epsilon(x)-d(x,v))\}$.
\end{algorithmic}
\end{algorithm}

The runtime bottleneck in the above algorithm is step (2) which
runs in $\tilde{O}(mn^{1/2})$ time if one uses BFS through each vertex of $S$.
Let us describe how to modify the algorithm. We will replace (2) with a truncated distance computation and also use the algorithm of Dor, Halperin and Zwick to handle large distances that we might have ignored in the truncated computation.\\

\begin{algorithm}
\caption{Our Modified Approximation.}\label{ouralg}
\begin{algorithmic}[1]
\Procedure{FasterApproximation}{}
\BState \emph{First part: Handle Large Distances}:
\State \multiline{Use Dor, Halperin and Zwick's algorithm to compute distance estimates $d'(\cdot,\cdot)$ so that for every $u,v\in V$, $d(u,v)\leq d'(u,v)\leq d(u,v)+a\log n$. 
Let $X=3a\log n$.}
\State Set $\tilde{D}_1 = \max_{u,v\in V} d'(u,v) - a\log n$.
\State For every $v\in V$, set $\tilde{\epsilon}_1(v) = \max_{u} d'(u,v) - a\log n$.
\BState \emph{Second Part: Handle Small Distances}: 
\State \multiline{Using BFS (see \cite{RV13} and \cite{cairo}), in $O(m+n)$ time compute $W,w,T$, where $W\subseteq V$ is a uniformly chosen subset of size $O(\sqrt n \log n)$, $w$ is the furthest vertex from $W$ and $T$ are the closest $\sqrt n$ vertices to $w$.
Let $S=\{w\}\cup W\cup T$.}
\State \multiline{Let $Q=2(X+a\log n)=8a\log n$. 
For every $s\in S$ and every $v\in V$ whose distance $d(s,v)$ is at most $Q$, compute $d(s,v)$. Let $d_{\leq}(s,v)$ denote $d(s,v)$ if we have computed it, and $\infty$ otherwise. Set $\epsilon_{\leq}(s)=\max_v d_{\leq}(s,v)$. }
\State Set $\tilde{D}_2=\max_{x\in S} \epsilon_{\leq}(x)$. 
\State \multiline{$\forall v\in V$, set $\tilde{\epsilon}_2(v)=\max\{d_\leq(w,v), \max_{x\in W} d_\leq(x,v),  \max_{y\in T} (\epsilon_\leq(y)-d_\leq(y,v))\}$. If $\epsilon_\leq(y)$ and $d_\leq(y,v))$ are both infinite, $\tilde{\epsilon}_2(v)$ is set to $\infty$.}
\BState \emph{Third Part: Set $\tilde{D},\tilde{\epsilon}(\cdot)$}:
\State If $\tilde{D}_1\geq X$, 
set $\tilde{D}=\tilde{D}_1$, and otherwise set $\tilde{D}=\tilde{D}_2$. 
\State \multiline{For every $v\in V$, if there exists some $x\in S$ such that $d'(x,v)\geq X+a\log n$, then set $\tilde{\epsilon}(v)=\tilde{\epsilon}_1(v)$, and otherwise $\tilde{\epsilon}(v)=\tilde{\epsilon}_2(v)$.}
\EndProcedure
\end{algorithmic}
\end{algorithm}

Consider our modified algorithm, \textsc{FasterApproximation}.
Now we will prove several claims.

\begin{claim} The running time of algorithm \textsc{FasterApproximation} is $\tilde{O}(\mathcal{M}(\sqrt{n},n,n))$.
\end{claim}

\begin{proof}
The Dor, Halperin, Zwick part of the algorithm (Step 3) runs in $\tilde{O}(n^2)$ time. Step 8 runs in $O((X+a\log n)\cdot \mathcal{M}(|S|,n,n))$ time where $S=\{w\}\cup W\cup T$, using the iterated rectangular matrix product algorithm. Recall that 
$X+a\log n = O(\log n)$. Thus Step 8 runs in $\tilde{O}(\mathcal{M}(|S|,n,n))$ time.
Since $|S|=\tilde{O}(\sqrt n)$ and we can partition an $|S|\times n\times n$ matrix product into $\polylog n$, $n^{1/2}\times n\times n$ matrix products, the runtime of the step is $\tilde{O}(\mathcal{M}(n^{1/(2)},n,n))$. Steps 10 and 13 run in $O(n|S|)<\tilde{O}(n^2)$ time. The rest of the steps run in linear time. Since $\mathcal{M}(n^{1/2},n,n)\geq n^2$ (one must at least read the input), the total running time is  $\tilde{O}(\mathcal{M}(n^{1/2},n,n))$.
\end{proof}

\begin{claim}
$\frac{2D-1}{3} \leq \tilde{D}\leq D$.
\end{claim}

\begin{proof}
Suppose that $\tilde{D}_1\geq X$. The algorithm returns $\tilde{D}=\tilde{D}_1=\max_{u,v} d'(u,v)-a\log n$. By the guarantee on $d'$, we have $D-a\log n \leq \tilde{D}_1\leq D$. Hence $\tilde{D}\geq D(1-(a\log n)/D)\geq D(1-(a\log n)/X) = 
2D/3 \geq (2D-1)/3$.

Suppose now that $\tilde{D}_1< X$. This means that $D<X+a\log n$ and every distance in the graph is $\leq X+a\log n$. In the second part of the algorithm we set $Q=2(X+a\log n)$, and hence every distance is computed exactly: for every $s\in S$, $v\in V$, $d_\leq(s,v)=d(s,v)$. Hence the second part of the algorithm will be identical to the RV/CGR algorithm and hence we get the same guarantees: $(2D-1)/3\leq \tilde{D}\leq D$.
\end{proof}

%
%
%
\begin{claim}
For every vertex $v$, $\frac{3\epsilon(v)-1}{5} \epsilon(v)\leq \tilde{\epsilon}(v)\leq \epsilon(v)$.
\end{claim}

\begin{proof} Fix $v$.
Suppose first that there exists some $x$ such that $d'(x,v)\geq X+a\log n$. Then $\epsilon(v)\geq\tilde{\epsilon}_1(v)=\max_u d'(u,v) - a\log n \geq \epsilon(v)-a\log n = \epsilon(v)(1-a\log n/\epsilon(v))$. Since $\epsilon(v)\geq d(x,v)\geq d'(x,v)-a\log n \geq X$, we get that $\tilde{\epsilon}_1(v)\geq \epsilon(v)(1-a\log n / X) = 2\epsilon(v)/3$. 

Now suppose that for all $x\in V$, $d'(x,v)<X+a\log n$. Then, also for all $x\in V$, $d(x,v)<X+a\log n$ and $\epsilon(v)<X+a\log n$.
Consider all the quantities needed in the second part of the algorithm to compute $\tilde{\epsilon}_2(v)$: 
\begin{itemize}
\item $d_\leq(w,v)$: since $\forall x\in V$, $d(x,v)<X+a\log n$, $d_\leq(w_i,v)=d(w_i,v)$ for each $w_i$;
\item $d_\leq(x,v)$ for every $x\in W$: as above, $d_\leq(x,v)=d(x,v)$;
\item $\epsilon_\leq(x)-d_\leq(x,v)$ for all $x\in T$: here, $\epsilon(x)\leq \epsilon(v)+d(x,v)\leq 2\epsilon(v)<2(X+a\log n)$. Since we compute all distances from vertices in $S$ up to $2(X+a\log n)$ and $x\in S$, $\epsilon_\leq(x)=\epsilon(x)$. Also as in the above bullets, $d_\leq(x,v)=d(x,v)$. 
\end{itemize}
Thus all the quantities needed are the correct ones and $\tilde{\epsilon}(v)=\tilde{\epsilon}_2(v)$ inherits the same guarantees as in the algorithm by Cairo et al.
\end{proof}

From Le Gall and Urrutia~\cite{legallrectnew} (see also,~\cite{legallrect}) we obtain that $\mathcal{M}(\sqrt{n},n,n)\leq O(n^{2.044183})$. This completes the proof of Theorem~\ref{thm:205}

Finally we note that our approach also works to speed up our almost $2$-approximation algorithm for $S$-$T$ Diameter as well, giving an $O(n^{2.045})$ time almost-$2$ approximation algorithm. The main reason is that, like in the Diameter approximation algorithm, if the $S$-$T$ Diameter is very large (say $D_{S,T}>100a\log n$), then the $+a\log n$ APSP algorithm with $a\log n$ subtracted will return an estimate that is at least $D_{S,T}-a\log n>0.99D_{S,T}$. On the other hand, our $S$-$T$ Diameter approximation algorithm only needs to know the distances up to $D_{S,T}$ to compute an estimate of $D_{S,T}$, and so if $D_{S,T}\leq 100a\log n$, then we only need to compute $O(\log n)$ matrix products of dimension $O(\sqrt n \log n)\times n\times n$ again.

%

%% file: main-journal.bbl
\newcommand{\etalchar}[1]{$^{#1}$}
\begin{thebibliography}{LWCW16}

\bibitem[AB17]{AbBo16}
Amir Abboud and Greg Bodwin.
\newblock The 4/3 additive spanner exponent is tight.
\newblock {\em Journal of the ACM (JACM)}, 64(4):1--20, 2017.

\bibitem[ACIM99]{aingworth}
D.~Aingworth, C.~Chekuri, P.~Indyk, and R.~Motwani.
\newblock Fast estimation of diameter and shortest paths (without matrix
  multiplication).
\newblock {\em SIAM J. Comput.}, 28(4):1167--1181, 1999.

\bibitem[ADD{\etalchar{+}}93]{AlthoferDDJS93}
I.~Alth\"{o}fer, G.~Das, D.~Dobkin, D.~Joseph, and J.~Soares.
\newblock On sparse spanners of weighted graphs.
\newblock {\em Discrete \& Computational Geometry}, 9(1):81--100, 1993.

\bibitem[AG11]{AbrahamG11}
Ittai Abraham and Cyril Gavoille.
\newblock On approximate distance labels and routing schemes with affine
  stretch.
\newblock In {\em International Symposium on Distributed Computing (DISC)},
  pages 404--415. Springer, 2011.

\bibitem[AG13]{AG13}
Rachit Agarwal and Philip~Brighten Godfrey.
\newblock Brief announcement: a simple stretch 2 distance oracle.
\newblock In {\em {ACM} Symposium on Principles of Distributed Computing,
  {PODC} '13, Montreal, QC, Canada, July 22-24, 2013}, pages 110--112, 2013.

\bibitem[AGM97]{AlonGM97}
Noga Alon, Zvi Galil, and Oded Margalit.
\newblock On the exponent of the all pairs shortest path problem.
\newblock {\em J. Comput. Syst. Sci.}, 54(2):255--262, 1997.

\bibitem[AVW16]{AbboudWW16}
Amir Abboud, Virginia {Vassilevska Williams}, and Joshua~R. Wang.
\newblock Approximation and fixed parameter subquadratic algorithms for radius
  and diameter in sparse graphs.
\newblock In {\em Proceedings of the Twenty-Seventh Annual {ACM-SIAM} Symposium
  on Discrete Algorithms, {SODA} 2016, Arlington, VA, USA, January 10-12,
  2016}, pages 377--391, 2016.

\bibitem[AWY15]{AbboudWY15}
Amir Abboud, Richard~Ryan Williams, and Huacheng Yu.
\newblock More applications of the polynomial method to algorithm design.
\newblock In {\em Proceedings of the Twenty-Sixth Annual {ACM-SIAM} Symposium
  on Discrete Algorithms, {SODA} 2015, San Diego, CA, USA, January 4-6, 2015},
  pages 218--230, 2015.

\bibitem[BCH{\etalchar{+}}15]{diam-prac6}
Michele Borassi, Pierluigi Crescenzi, Michel Habib, Walter~A Kosters, Andrea
  Marino, and Frank~W Takes.
\newblock Fast diameter and radius bfs-based computation in (weakly connected)
  real-world graphs: With an application to the six degrees of separation
  games.
\newblock {\em Theoretical Computer Science}, 586:59--80, 2015.

\bibitem[BGSU08]{BaswanaGSU08}
Surender Baswana, Akshay Gaur, Sandeep Sen, and Jayant Upadhyay.
\newblock Distance oracles for unweighted graphs: Breaking the quadratic
  barrier with constant additive error.
\newblock In {\em Automata, Languages and Programming, 35th International
  Colloquium, {ICALP} 2008, Reykjavik, Iceland, July 7-11, 2008, Proceedings,
  Part {I:} Tack {A:} Algorithms, Automata, Complexity, and Games}, pages
  609--621, 2008.

\bibitem[BK10]{BaKa10}
S.~Baswana and T.~Kavitha.
\newblock Faster algorithms for all-pairs approximate shortest paths in
  undirected graphs.
\newblock {\em SIAM J. Comput.}, 39(7):2865--2896, 2010.

\bibitem[BKMP10]{BaswanaKMP10}
Surender Baswana, Telikepalli Kavitha, Kurt Mehlhorn, and Seth Pettie.
\newblock Additive spanners and (alpha, beta)-spanners.
\newblock {\em {ACM} Trans. Algorithms}, 7(1):5:1--5:26, 2010.

\bibitem[Bon20]{4ov}
{\'E}douard Bonnet.
\newblock Inapproximability of diameter in super-linear time: Beyond the 5/3
  ratio.
\newblock {\em arXiv preprint arXiv:2008.11315}, 2020.

\bibitem[BS06]{BaswanaS06}
Surender Baswana and Sandeep Sen.
\newblock Approximate distance oracles for unweighted graphs in expected
  \emph{O}(\emph{n}\({}^{\mbox{2}}\)) time.
\newblock {\em {ACM} Trans. Algorithms}, 2(4):557--577, 2006.

\bibitem[BS07]{BaswanaS07}
Surender Baswana and Sandeep Sen.
\newblock A simple and linear time randomized algorithm for computing sparse
  spanners in weighted graphs.
\newblock {\em Random Struct. Algorithms}, 30(4):532--563, 2007.

\bibitem[CG20]{choudhary2020extremal}
Keerti Choudhary and Omer Gold.
\newblock Extremal distances in directed graphs: tight spanners and
  near-optimal approximation algorithms.
\newblock In {\em Proceedings of the Fourteenth Annual ACM-SIAM Symposium on
  Discrete Algorithms (SODA)}, pages 495--514. SIAM, 2020.

\bibitem[CGR16]{cairo}
Massimo Cairo, Roberto Grossi, and Romeo Rizzi.
\newblock New bounds for approximating extremal distances in undirected graphs.
\newblock In {\em Proceedings of the Twenty-Seventh Annual {ACM-SIAM} Symposium
  on Discrete Algorithms, {SODA} 2016, Arlington, VA, USA, January 10-12,
  2016}, pages 363--376, 2016.

\bibitem[CGS15]{cyganbaur}
Marek Cygan, Harold~N. Gabow, and Piotr Sankowski.
\newblock Algorithmic applications of baur-strassen's theorem: Shortest cycles,
  diameter, and matchings.
\newblock {\em J. ACM}, 62(4):28:1--28:30, September 2015.

\bibitem[Cha12]{chan06j}
Timothy~M. Chan.
\newblock All-pairs shortest paths for unweighted undirected graphs in
  \emph{o}(\emph{mn}) time.
\newblock {\em {ACM} Trans. Algorithms}, 8(4):34:1--34:17, 2012.

\bibitem[Che13]{Chechik13}
Shiri Chechik.
\newblock New additive spanners.
\newblock In {\em Proceedings of the twenty-fourth annual ACM-SIAM symposium on
  Discrete algorithms (SODA)}, pages 498--512. SIAM, 2013.

\bibitem[Che15]{Chechik15}
Shiri Chechik.
\newblock Approximate distance oracles with improved bounds.
\newblock In {\em Proceedings of the Forty-Seventh Annual {ACM} on Symposium on
  Theory of Computing, {STOC} 2015, Portland, OR, USA, June 14-17, 2015}, pages
  1--10, 2015.

\bibitem[Chu87]{fanchung}
F.R.K Chung.
\newblock Diameters of graphs: Old problems and new results.
\newblock {\em Congressus Numerantium}, 60:295–--317, 1987.

\bibitem[CLR{\etalchar{+}}14]{ChechikLRSTW14}
Shiri Chechik, Daniel~H. Larkin, Liam Roditty, Grant Schoenebeck, Robert~Endre
  Tarjan, and Virginia {Vassilevska Williams}.
\newblock Better approximation algorithms for the graph diameter.
\newblock In {\em Proceedings of the Twenty-Fifth Annual {ACM-SIAM} Symposium
  on Discrete Algorithms, {SODA} 2014, Portland, Oregon, USA, January 5-7,
  2014}, pages 1041--1052, 2014.

\bibitem[CW16]{ChanW16}
Timothy~M. Chan and Ryan Williams.
\newblock Deterministic apsp, orthogonal vectors, and more: Quickly
  derandomizing razborov-smolensky.
\newblock In {\em Proceedings of the Twenty-Seventh Annual {ACM-SIAM} Symposium
  on Discrete Algorithms, {SODA} 2016, Arlington, VA, USA, January 10-12,
  2016}, pages 1246--1255, 2016.

\bibitem[CZ01]{CoZw01}
E.~Cohen and U.~Zwick.
\newblock All-pairs small-stretch paths.
\newblock {\em J. Algorithms}, 38(2):335--353, 2001.

\bibitem[DHZ00]{DorHZ00}
D.~Dor, S.~Halperin, and U.~Zwick.
\newblock All-pairs almost shortest paths.
\newblock {\em SIAM J. Comput.}, 29(5):1740--1759, 2000.

\bibitem[DW20]{dalirrooyfard2020tight}
Mina Dalirrooyfard and Nicole Wein.
\newblock Tight conditional lower bounds for approximating diameter in directed
  graphs.
\newblock {\em arXiv preprint arXiv:2011.03892}, 2020.

\bibitem[EP04]{ElkinP04}
Michael Elkin and David Peleg.
\newblock (1+epsilon, beta)-spanner constructions for general graphs.
\newblock {\em SIAM J. Comput.}, 33(3):608--631, 2004.

\bibitem[GU18]{legallrectnew}
Fran{\c{c}}ois~Le Gall and Florent Urrutia.
\newblock Improved rectangular matrix multiplication using powers of the
  coppersmith-winograd tensor.
\newblock In {\em Proceedings of the Twenty-Ninth Annual ACM-SIAM Symposium on
  Discrete Algorithms}, pages 1029--1046. SIAM, 2018.

\bibitem[Hir98]{hirschsat}
Edward~A Hirsch.
\newblock Two new upper bounds for sat.
\newblock In {\em Proceedings of the ninth annual ACM-SIAM symposium on
  Discrete algorithms (SODA)}, pages 521--530. Society for Industrial and
  Applied Mathematics, 1998.

\bibitem[IPZ01]{ipz2}
R.~Impagliazzo, R.~Paturi, and F.~Zane.
\newblock Which problems have strongly exponential complexity?
\newblock {\em J. Comput. Syst. Sci.}, 63(4):512--530, 2001.

\bibitem[KI92]{katoh}
Naoki Katoh and Kazuo Iwano.
\newblock Finding k farthest pairs and k closest/farthest bichromatic pairs for
  points in the plane.
\newblock In {\em Proceedings of the Eighth Annual Symposium on Computational
  Geometry}, SCG '92, pages 320--329, 1992.

\bibitem[Knu17]{Knudsen17}
Mathias B{\ae}k~Tejs Knudsen.
\newblock Additive spanners and distance oracles in quadratic time.
\newblock In {\em 44th International Colloquium on Automata, Languages, and
  Programming, {ICALP} 2017, July 10-14, 2017, Warsaw, Poland}, pages
  64:1--64:12, 2017.

\bibitem[{Le }12]{legallrect}
Fran{\c{c}}ois {Le Gall}.
\newblock Faster algorithms for rectangular matrix multiplication.
\newblock In {\em 53rd Annual {IEEE} Symposium on Foundations of Computer
  Science, {FOCS} 2012, New Brunswick, NJ, USA, October 20-23, 2012}, pages
  514--523, 2012.

\bibitem[{Le }14]{legallmult}
Fran{\c{c}}ois {Le Gall}.
\newblock Powers of tensors and fast matrix multiplication.
\newblock In {\em International Symposium on Symbolic and Algebraic
  Computation, {ISSAC} '14, Kobe, Japan, July 23-25, 2014}, pages 296--303,
  2014.

\bibitem[Li20]{3vs5}
Ray Li.
\newblock Settling {SETH} vs. approximate sparse directed unweighted diameter
  (up to ({NU}){NSETH}).
\newblock {\em arXiv preprint arXiv:2008.05106}, 2020.

\bibitem[LWCW16]{diam-prac3}
T.~C. Lin, M.~J. Wu, W.~J. Chen, and B.~Y. Wu.
\newblock Computing the diameters of huge social networks.
\newblock In {\em 2016 International Computer Symposium (ICS)}, pages 6--11,
  2016.

\bibitem[LWW18]{lincolnsoda}
Andrea Lincoln, Virginia~Vassilevska Williams, and Ryan Williams.
\newblock Tight hardness for shortest cycles and paths in sparse graphs.
\newblock In {\em Proceedings of the Twenty-Ninth Annual ACM-SIAM Symposium on
  Discrete Algorithms (SODA)}, pages 1236--1252. SIAM, 2018.

\bibitem[Pet04]{Pettie04}
S.~Pettie.
\newblock A new approach to all-pairs shortest paths on real-weighted graphs.
\newblock {\em Theor. Comput. Sci.}, 312(1):47--74, 2004.

\bibitem[PPSZ05]{PPSZ05}
R.~Paturi, P.~Pudl\'{a}k, M.~E. Saks, and F.~Zane.
\newblock An improved exponential-time algorithm for $k$-{SAT}.
\newblock {\em J. ACM}, 52(3):337--364, 2005.

\bibitem[PR05]{PettieR05}
Seth Pettie and Vijaya Ramachandran.
\newblock A shortest path algorithm for real-weighted undirected graphs.
\newblock {\em {SIAM} J. Comput.}, 34(6):1398--1431, 2005.

\bibitem[PR10]{patrascuroditty}
Mihai Patrascu and Liam Roditty.
\newblock Distance oracles beyond the thorup-zwick bound.
\newblock In {\em 2010 IEEE 51st Annual Symposium on Foundations of Computer
  Science (FOCS)}, pages 815--823. IEEE, 2010.

\bibitem[PRT12a]{PatrascuRT12}
Mihai Patrascu, Liam Roditty, and Mikkel Thorup.
\newblock A new infinity of distance oracles for sparse graphs.
\newblock In {\em 2012 IEEE 53rd Annual Symposium on Foundations of Computer
  Science (FOCS)}, pages 738--747. IEEE, 2012.

\bibitem[PRT12b]{diam-prac1}
David Peleg, Liam Roditty, and Elad Tal.
\newblock Distributed algorithms for network diameter and girth.
\newblock In {\em Automata, Languages, and Programming: 39th International
  Colloquium, ICALP 2012, Warwick, UK, July 9-13, 2012, Proceedings, Part II},
  pages 660--672, 2012.

\bibitem[RV13]{RV13}
Liam Roditty and Virginia {Vassilevska Williams}.
\newblock Fast approximation algorithms for the diameter and radius of sparse
  graphs.
\newblock In {\em Proceedings of the 45th annual ACM symposium on Symposium on
  theory of computing}, STOC '13, pages 515--524, New York, NY, USA, 2013. ACM.

\bibitem[RW19]{rubinstein2019seth}
Aviad Rubinstein and Virginia~Vassilevska Williams.
\newblock Seth vs approximation.
\newblock {\em ACM SIGACT News}, 50(4):57--76, 2019.

\bibitem[Sch99]{Scho99sat}
T~Schoning.
\newblock A probabilistic algorithm for k-sat and constraint satisfaction
  problems.
\newblock In {\em 40th Annual Symposium on Foundations of Computer Science
  (FOCS)}, pages 410--414. IEEE, 1999.

\bibitem[Sei95]{Seidel}
Raimund Seidel.
\newblock On the all-pairs-shortest-path problem in unweighted undirected
  graphs.
\newblock {\em Journal of computer and system sciences}, 51(3):400--403, 1995.

\bibitem[Som16]{Sommer16}
Christian Sommer.
\newblock {All-Pairs Approximate Shortest Paths and Distance Oracle
  Preprocessing}.
\newblock In {\em 43rd International Colloquium on Automata, Languages, and
  Programming (ICALP 2016)}, volume~55 of {\em Leibniz International
  Proceedings in Informatics (LIPIcs)}, pages 55:1--55:13, 2016.

\bibitem[Sto10]{stothers}
A.~Stothers.
\newblock On the complexity of matrix multiplication.
\newblock {\em Ph.D. Thesis, U. Edinburgh}, 2010.

\bibitem[SZ99]{sz99}
Avi Shoshan and Uri Zwick.
\newblock All pairs shortest paths in undirected graphs with integer weights.
\newblock In {\em 40th Annual Symposium on Foundations of Computer Science
  (CFOCS)}, pages 605--614. IEEE, 1999.

\bibitem[Tho99]{thorup1999undirected}
Mikkel Thorup.
\newblock Undirected single-source shortest paths with positive integer weights
  in linear time.
\newblock {\em Journal of the ACM (JACM)}, 46(3):362--394, 1999.

\bibitem[TZ01]{ThZw01}
Mikkel Thorup and Uri Zwick.
\newblock Compact routing schemes.
\newblock In {\em Proceedings of the thirteenth annual ACM symposium on
  Parallel algorithms and architectures (SPAA)}, pages 1--10, 2001.

\bibitem[TZ05]{ThZw05}
Mikkel Thorup and Uri Zwick.
\newblock Approximate distance oracles.
\newblock {\em Journal of the ACM (JACM)}, 52(1):1--24, 2005.

\bibitem[{Vas}15]{ipecsurvey}
Virginia {Vassilevska Williams}.
\newblock Hardness of easy problems: Basing hardness on popular conjectures
  such as the strong exponential time hypothesis (invited talk).
\newblock In {\em 10th International Symposium on Parameterized and Exact
  Computation, {IPEC} 2015, September 16-18, 2015, Patras, Greece}, pages
  17--29, 2015.

\bibitem[Wil05]{TCS05}
R.~Williams.
\newblock A new algorithm for optimal $2$-constraint satisfaction and its
  implications.
\newblock {\em Theor. Comput. Sci.}, 348(2--3):357--365, 2005.

\bibitem[Wil12]{williams2012multiplying}
Virginia~Vassilevska Williams.
\newblock Multiplying matrices faster than coppersmith-winograd.
\newblock In {\em Proceedings of the forty-fourth annual ACM symposium on
  Theory of computing}, pages 887--898. ACM, 2012.

\bibitem[Wil14]{ryanapsp}
Ryan Williams.
\newblock Faster all-pairs shortest paths via circuit complexity.
\newblock In {\em Symposium on Theory of Computing, {STOC} 2014, New York, NY,
  USA, May 31 - June 03, 2014}, pages 664--673, 2014.

\bibitem[Wil18]{williams2018some}
Virginia~Vassilevska Williams.
\newblock On some fine-grained questions in algorithms and complexity.
\newblock In {\em Proceedings of the ICM}, volume~3, pages 3431--3472. World
  Scientific, 2018.

\bibitem[Woo06]{Woodruff}
D.~P. Woodruff.
\newblock Lower bounds for additive spanners, emulators, and more.
\newblock In {\em Proceedings of the 47th Annual IEEE Symposium on Foundations
  of Computer Science}, FOCS '06, pages 389--398, 2006.

\bibitem[WW10]{focsy}
Virginia~Vassilevska Williams and Ryan Williams.
\newblock Subcubic equivalences between path, matrix and triangle problems.
\newblock In {\em 2010 IEEE 51st Annual Symposium on Foundations of Computer
  Science (FOCS)}, pages 645--654. IEEE, 2010.

\bibitem[Zwi02]{zwickbridge}
U.~Zwick.
\newblock All pairs shortest paths using bridging sets and rectangular matrix
  multiplication.
\newblock {\em J. ACM}, 49(3):289--317, 2002.

\end{thebibliography}
